\documentclass[journal,onecolumn,draftclsnofoot]{IEEEtran}

\ifCLASSINFOpdf
  
\else
  
\fi

\usepackage{amsmath}
\usepackage{amsthm}
\usepackage{cases}
\usepackage{amsfonts}
\usepackage{cite}
\usepackage{amssymb}
\usepackage{enumerate}
\usepackage{graphicx}
\usepackage{float}
\usepackage{caption}
\usepackage{verbatim}
\usepackage{dblfloatfix}
\usepackage[cmintegrals]{newtxmath}
\usepackage[ruled,linesnumbered]{algorithm2e}

\DeclareMathOperator\supp{supp}
\newtheorem{theorem}{Theorem}
\newtheorem{lemma}{Lemma}
\newtheorem{definition}{Definition}
\newtheorem{corollary}{Corollary}
\newtheorem{remark}{Remark}
\newtheorem{proposition}{Proposition}
\newtheorem{example}{Example}
\begin{document}

\title{On Strong Secrecy for Multiple Access Channels with States and Causal CSI}

\author{%
\IEEEauthorblockN{Yiqi Chen\IEEEauthorrefmark{1}\IEEEauthorrefmark{2},
                Tobias Oechtering\IEEEauthorrefmark{2},
                Mikael Skoglund\IEEEauthorrefmark{2},
                and Yuan Luo\IEEEauthorrefmark{1}}\\
\IEEEauthorblockA{\IEEEauthorrefmark{1}%
              Shanghai Jiao Tong University,  
              Shanghai,
              China, 
              \{chenyiqi,yuanluo\}@sjtu.edu.cn}\\
\IEEEauthorblockA{\IEEEauthorrefmark{2}%
              KTH Royal Institute of Technology,  
              100 44 Stockholm,
              Sweden, 
              \{oech,skoglund\}@kth.se}
\thanks{Corresponding author: Yuan Luo. Part of the earlier results in this paper were accepted for presentation at the 2023 IEEE International Symposium on Information Theory (ISIT)[arXiv version: https://arxiv.org/abs/2305.01393].} 
}


\maketitle
\thispagestyle{empty}
\pagestyle{empty}
\begin{abstract}
  Strong secrecy communication over a discrete memoryless state-dependent multiple access channel (SD-MAC) with an external eavesdropper is investigated. The channel is governed by discrete memoryless and i.i.d. channel states and the channel state information (CSI) is revealed to the encoders in a causal manner. Inner and outer bounds are provided. To establish the inner bound, we investigate coding schemes incorporating wiretap coding and secret key agreement between the sender and the legitimate receiver. Two kinds of block Markov coding schemes are proposed. The first one is a new coding scheme using backward decoding and Wyner-Ziv coding and the secret key is constructed from a lossy description of the CSI. The other one is an extended version of the existing coding scheme for point-to-point wiretap channels with causal CSI. A numerical example shows that the achievable region given by the first coding scheme can be strictly larger than the second one. However, these two schemes do not outperform each other in general and there exists some numerical examples that in different channel models each coding scheme achieves some rate pairs that cannot be achieved by another scheme. Our established inner bound reduces to some best-known results in the literature as special cases. We further investigate some capacity-achieving cases for state-dependent multiple access wiretap channels (SD-MAWCs) with degraded message sets. It turns out that the two coding schemes are both optimal in these cases.
\end{abstract}


%
\IEEEpeerreviewmaketitle

\section{Introduction}
Secure communication over a discrete memoryless channel (DMC) was first studied in \cite{wyner1975wire} where the sender communicates to the legitimate receiver over the main channel in the presence of an external eavesdropper through a degraded version of the main channel. The model was further extended to more general case called broadcast channels with confidential messages in \cite{csiszar1978broadcast}. Following these landmark papers, secrecy capacity results of different generalized models have been reported in recent years. In \cite{liang2008multiple}, multiple access channel with confidential messages was investigated. Each sender can receive a noisy version of the input sequence from the other and is treated as an eavesdropper by another user. Capacity-equivocation regions and perfect secrecy capacity regions of these models were characterized. Multiple access wiretap channel with an external eavesdropper under strong secrecy constraint was considered in \cite{yassaee2010multiple} and an achievable region was provided. Strong secrecy of MAC with common messages/conference encoders was investigated in \cite{wiese2013strong}. Relay channels with confidential messages were studied in \cite{lai2008relay}\cite{oohama2007capacity}. Interference channels with confidential messages were studied in \cite{liu2008discrete}\cite{liang2009capacity}. For continuous alphabets, Gaussian wiretap channels were studied in \cite{leung1978gaussian}.

Coding for channels with random parameters was first studied in \cite{shannon1958channels}. The transition probability of the channel is governed by channel state, and the channel state information (CSI) is revealed to the encoder in a causal manner. For the case that CSI is revealed to the encoder in a noncausal way, \cite{gel1980coding} solved the problem using random binning technique. Although it was proved by Shannon that Shannon strategy is optimal for a point-to-point state-dependent DMC (SD-DMC), in \cite{lapidoth2012multiple}, the authors showed that Shannon strategy is suboptimal for state-dependent multiple access channels (SD-MAC). It was proved that using block Markov coding with Wyner-Ziv coding and backward decoding achieves a strictly larger achievable region than Shannon strategy in some cases. SD-MAC with independent states at each sender was investigated in \cite{lapidoth2012multiple2}. A more general channel model is that the channel states vary in an unknown manner, which is called arbitrarily varying channels (AVCs), and was first considered in \cite{blackwell1959capacity}. Arbitrarily varying channels with CSI at the encoder were studied in \cite{pereg2018arbitrarily}\cite{pereg2019arbitrarily}.

Secure communication over channels with CSI attracts a lot of attention recently. Existing results show that knowing CSI at legitimate users' side can improve the secrecy achievable rate. One widely used method to achieve secure communication is using wiretap codes\cite{wyner1975wire}, which introduces some additional local randomness in the encoding procedure and protects the messages using the statistical properties of the channel itself. Wiretap channel with noncausal CSI at the encoder was studied in \cite{chen2008wiretap} and lower and upper bounds of the secrecy capacity were presented by combining random binning and wiretap coding. The model was then revisited in \cite{dai2012some} and further studied in \cite{goldfeld2019wiretap} with a more stringent secrecy constraint. AVCs with secrecy constraints were investigated in \cite{chen2021strong}\cite{chen2022strong}. In \cite{koga2013information}, the authors investigated the minimum size of uniform random numbers that approximate a given output distribution through a channel and named it channel resolvability. It provides the base of a general framework of wiretap code design and was proved in \cite{bloch2013strong} that strong secrecy can be achieved based on channel resolvability. Such output distribution analysis based wiretap codes design was further extended in \cite{goldfeld2016arbitrarily}\cite{goldfeld2016semantic} and semantic security was achieved.

For a state-dependent channel, knowing CSI can help secure transmission over the channel since the system participants have some additional resources to construct a secret key and encrypt part of our messages by using Shannon one-time pad cipher \cite{shannon1949communication}. This combination of wiretap codes and secret key agreement has been used in recent works on state-dependent wiretap channels (SD-WTCs). In \cite{chia2012wiretap}, the authors addressed secure communication over wiretap channels with causal CSI at both the encoder and the decoder. The coding scheme uses block Markov coding and generates the secret key in each block upon observing the channel state sequence of the last block. Lower and upper bounds of the secrecy capacity were provided with a secrecy capacity result when the main channel is less noisy than the wiretap channel. The author of \cite{fujita2016secrecy} provided a lower bound of the physically degraded wiretap channels when causal CSI at the encoder and an upper bound when noncausal CSI at the encoder. The condition that the inner and outer bounds meet each other was also given. The authors of \cite{sasaki2019wiretap} extended the above mentioned works about wiretap channels with causal CSI at the encoder by considering strong secrecy constraint. A block Markov coding scheme with key agreement\cite{csiszar2011information} and forward decoding was adopted. Inner and outer bounds with some capacity-achieving cases were provided in the paper. The results were further strengthened in \cite{sasaki2021wiretap} to semantic secrecy constraint by using soft-covering method\cite{goldfeld2019wiretap}.

The main contribution of this paper is an achievable region and its corresponding coding schemes for the SD-MAC with causal CSI at encoders and an external eavesdropper. In addition to extending the coding scheme giving the state-of-the-art strong secrecy result for the SD-WTC\cite{sasaki2019wiretap} to multi-user case, we further propose a new coding scheme using lossy compression and block Markov coding with backward decoding. We notice that in \cite{sonee2014wiretap} a superposition coding scheme with lossy compression and backward decoding was proposed with a different key construction method aiming at a weak secrecy result. However, the secrecy analysis of the coding scheme is not convincing to us. More details are discussed in Remark \ref{rem: worst case assumption}.  The coding scheme in \cite{sasaki2019wiretap} is extended to the multi-user case. It uses state sequences to construct secret keys, and lossless compression is adopted to reconstruct state sequences at the decoder side. The motivation of our new coding scheme is the fact that the lossless compression step causes some additional cost at the decoder side and one can decrease this cost by using lossy compression and constructing secret keys based on the lossy descriptions of the state sequences. It was also proved in \cite{lapidoth2012multiple} that the Shannon strategy is not optimal for SD-MAC when causal CSI at both encoders and block Markov coding with lossy compression achieves a strictly larger achievable rate region. We show by different numerical examples that our new coding scheme achieves some rate pairs that cannot be achieved by the extended coding scheme of that in \cite{sasaki2019wiretap}, and the two coding schemes do not outperform each other in general. Several different capacity-achieving cases are also provided.

The rest of the paper is organized as follows. In Section II, we provide the notations and definitions of the channel model considered in this paper. In Section III, we summarize our capacity results, which include an inner bound and an outer bound. The inner bound consists of three different regions, each corresponding to a different coding scheme. The first two coding schemes (we refer to them as Coding scheme  1 and Coding scheme 2, respectively) combine wiretap coding and secret key agreement while the third coding scheme only uses wiretap coding. Explanations of these regions are also presented in the section. Sections IV and V provide coding schemes for the inner bound result. The coding scheme in Section IV is a block Markov coding with backward decoding, where the secret key is constructed based on the lossy description of the channel state sequence observed by the encoders according to the secret key agreement\cite{csiszar2000common}. Section V generalizes the coding scheme in \cite{sasaki2019wiretap} to a multi-user case. Section VI provides some numerical examples and applications of our capacity results. One first numerical example gives two channel models in which one coding scheme achieves some rate points that cannot be achieved by another scheme, the second example shows that the achievable region of Coding scheme 1 can be strictly larger than the region of Coding scheme 2, and the third example is a `state-reproducing' channel model where the inner and outer bounds meet each other for both Coding schemes  1 and 2. Then, four capacity-achieving cases including one-sender model, causal CSI at only one side, causal CSI at one side and strictly causal CSI at another side and the state-independent channel model are considered. Section VII concludes this paper.
\section{Definitions}
\subsection{Notations}\label{sec: notations}
Throughout this paper, random variables and corresponding sample values are denoted by capital letters and lowercase letters, e.g. $X$ and $x$. Sets are denoted by calligraphic letters. Capital and lowercase letters in boldface represent n-length random and sample sequences, respectively, e.g. $\boldsymbol{X}=(X_1,X_2,\dots,X_n)$ and $\boldsymbol{x}=(x_1,x_2,\dots,x_n)$. Let $\mathcal{X}^n$  be the n-fold Cartesian product of $\mathcal{X}$, which is the set of all possible $\boldsymbol{x}$. To denote substrings, let $X^{b}=(X_1,X_2,\dots,X_b)$ and $X^{[b+1]}=(X_{b+1},X_{b+2},\dots,X_n)$. We use $P_X$ to denote the probability mass functions (PMFs) of the random variable $X$ and $P_{XY},P_{X|Y}$ to denote the joint PMFs and conditional PMFs, respectively. The set of all PMFs over a given set $\mathcal{X}$ is written by $\mathcal{P}(\mathcal{X})$. For a sequence $\boldsymbol{x}$ generated i.i.d. according to some distribution $P_X$, we denote $P_X^n(\boldsymbol{x})=\prod_{i=1}^n P_X(x)$.

For given random variable $X$ with distribution $P_X$, we use $\mathbb{E}[X]=\sum_{x\in\mathcal{X}}P_X(x)\cdot x$ to represent its expectation. For two distributions $P_X,Q_X\in\mathcal{P}(\mathcal{X})$, the KL Divergence between $P_X$ and $Q_X$ is denoted by 
\begin{align*}
  D(P_X||Q_X)=\sum_{x\in\mathcal{X}}P_X(x)\log\frac{P_X(x)}{Q_X(x)}.
\end{align*}

For any distribution $P_X$ on $\mathcal{X}$, a sequence $\boldsymbol{x}\in\mathcal{X}^n$ is $\delta-$typical if 
  \begin{align*}
    \left| \frac{1}{n}N(x|\boldsymbol{x}) - P_X(x) \right| \leq \delta
  \end{align*}
  for any $x\in\mathcal{X}$, where $N(x|\boldsymbol{x})$ is the number of symbol $x$ in sequence $\boldsymbol{x}$. The set of such sequences is denoted by $T^n_{P_X,\delta}$. 

  For a conditional distribution $P_{Y|X}$, a sequence $\boldsymbol{y}\in\mathcal{Y}^n$ is $\delta-$conditional typical if
  \begin{align*}
    \left| \frac{1}{n}N(x,y|\boldsymbol{x},\boldsymbol{y}) - \frac{1}{n}N(x|\boldsymbol{x})P_{Y|X}(y|x) \right| \leq \delta
  \end{align*}
  for any $x\in\mathcal{X},y\in\mathcal{Y}$, where $N(x,y|\boldsymbol{x},\boldsymbol{y})$ is the number of symbol pair $(x,y)$ in $(\boldsymbol{x},\boldsymbol{y})$. The set of such sequences is denoted by $T^n_{P_{Y|X},\delta}[\boldsymbol{x}]$.
The definition of joint typical sequences $T^n_{P_{XY},\delta}$ follows similarly and further define set $T^n_{P_{XY},\delta}[\boldsymbol{x}]:= \{\boldsymbol{y}: (\boldsymbol{x},\boldsymbol{y})\in T^n_{P_{XY},\delta}\}$. Some useful properties of typical sequences are provided in Appendix \ref{app: properties of typical sequences}.
\subsection{Channel Model}
In this subsection, we give the definition of the channel model considered in this paper and the coding scheme used for the communication.
\begin{figure}[h]
  \centering
  \includegraphics[scale=0.7]{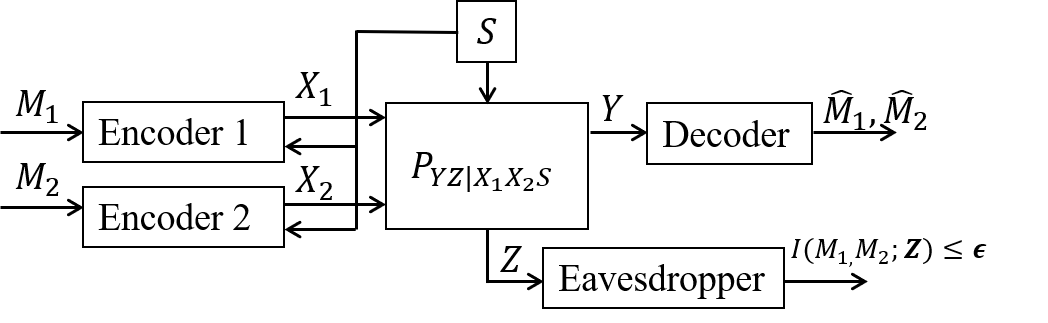}
  \caption{SD-MAWC with Causal CSI at Encoders.}
\end{figure}
\begin{definition}\label{def: channel model}
  A memoryless discrete multiple access wiretap channel with state is defined by a stochastic transition matrix $\mathcal{W}:\mathcal{X}_1\times\mathcal{X}_2\times\mathcal{S}\to\mathcal{Y}\times\mathcal{Z}$, where $\mathcal{X}_1$ and $\mathcal{X}_2$ are finite input alphabets, $\mathcal{Y}, \mathcal{Z}$ are finite output alphabets, $\mathcal{S}$ is finite state alphabet. The transition probability from input sequences $(\boldsymbol{x}_1,\boldsymbol{x}_2)$ to the output sequences $(\boldsymbol{y},\boldsymbol{z})$ given state sequence $\boldsymbol{s}$ is 
  \begin{align*}
    P^n_{YZ|X_1X_2S}(\boldsymbol{y},\boldsymbol{z}|\boldsymbol{x}_1,\boldsymbol{x}_2,\boldsymbol{s}) = \prod_{i=1}^n P_{YZ|X_1X_2S}(y_i,z_i|x_{1i},x_{2i},s_i),
  \end{align*}
  where $\boldsymbol{y}$ is the output of the main channel and $\boldsymbol{z}$ is the output of the wiretap channel. The channel state sequences are generated by a discrete memoryless source such that $P_S^n(\boldsymbol{s})=\prod_{i=1}^n P_S(s_i)$.
\end{definition}
\begin{remark}
  The model defined in Definition \ref{def: channel model} is an SD-MAC with common states. In \cite{lapidoth2012multiple2}, the authors considered the SD-MACs with double states that the channels are governed by two independent memoryless channel states, each of which is revealed to a different sender in a strictly causal/causal way. In this paper, we only consider the common state case.
\end{remark}
\begin{definition}
  A  code $(2^{nR_1},2^{nR_2},f^n_1,f^n_2,\phi)$ consists of message sets $\mathcal{M}_1=[1:2^{nR_1}],\mathcal{M}_2=[1:2^{nR_2}]$, sets of encoders $f_{i,j}:\mathcal{M}_i \times \mathcal{S}^j \to \mathcal{X}_i, i=1,2,j=1,\dots,n$ and a decoder $\phi:\mathcal{Y}^n\to\mathcal{M}_1\times\mathcal{M}_2$. To transmit message $M_i$, the sender sends codeword $\boldsymbol{X}_i$ with each component $X_{ij}$ generated according to $f_{i,j}(M_i,\boldsymbol{S}^j)$. The receiver receives $\boldsymbol{Y}$ and estimates the message by the decoder $(\hat{M}_1,\hat{M}_2)=\phi(\boldsymbol{Y})$. The average decoding error of a code is 
  \begin{align*}
    P_e = \frac{1}{|\mathcal{M}_1||\mathcal{M}_2|}\sum_{m_1\in\mathcal{M}}\sum_{m_2\in\mathcal{M}_2}Pr\{(M_1,M_2)\neq(\hat{M}_1,\hat{M}_2)|M_1=m_1,M_2=m_2\}.
  \end{align*}
The information leakage of the code is defined as 
\begin{align*}
  I(M_1,M_2;\boldsymbol{Z}),
\end{align*}
where $\boldsymbol{Z}$ is the output of the wiretap channel.
\end{definition}

\begin{definition}
  A rate pair $(R_1,R_2)$ is said to be achievable for the SD-MAWC with causal channel state information at encoders if for any $\epsilon>0$ there exists an $n_0$ such that for any $n>n_0$ there exists a code $(2^{nR_1},2^{nR_2},f^n_1,f^n_2,\phi)$ with
  \begin{align*}
    &P_e \leq \epsilon,\\
    &I(M_1,M_2;\boldsymbol{Z}) \leq \epsilon,
  \end{align*}
  where the second inequality is the strong secrecy constraint.
  The secrecy capacity of the SD-MAWC is the closure of all sets of achievable rate pairs.
\end{definition}

\section{Main result}\
This section presents the main capacity results of SD-MAWCs with causal CSI. Theorem \ref{the: mawc with casual csi} gives an inner bound of the secrecy capacity of SD-MAWCs with causal CSI, and an outer bound is provided in Theorem \ref{the: outer bound}. We use $U,U_1,U_2,V$ as auxiliary random variables.

We first introduce the secrecy achievable region of multiple access wiretap channels without channel state. Let $\mathcal{R}_{0i}(P_{U_1U_2X_1X_2}),i=1,2,3$ be sets of real number pairs $(R_1,R_2)$ such that
\begin{equation*}
  \begin{split}
    &\mathcal{R}_{01}(P_{U_1U_2X_1X_2})=\left\{ 
      \begin{aligned}
        &R_1>0,R_2>0;\\
        &R_1 \leq I(U_1;Y|U_2)-I(U_1;Z);\\
        &R_2 \leq I(U_2;Y|U_1)-I(U_2;Z);\\  
        &R_1+R_2 \leq I(U_1,U_2;Y)-I(U_1,U_2;Z);\\
      \end{aligned}
    \right.\\
    &\mathcal{R}_{02}(P_{U_1U_2X_1X_2})=\left\{ 
      \begin{aligned}
        &R_1=0,R_2>0;\\
        &R_2 \leq I(U_2;Y|U_1)-I(U_2;Z|U_1);\\  
      \end{aligned}
    \right.\\
    &\mathcal{R}_{03}(P_{U_1U_2X_1X_2})=\left\{ 
      \begin{aligned}
        &R_1>0,R_2=0;\\
        &R_1 \leq I(U_1;Y|U_2)-I(U_1;Z|U_2);\\  
      \end{aligned}
    \right.
  \end{split}
\end{equation*}
where the mutual information is according to the joint distribution $P_{U_1U_2X_1X_2YZ}=P_{U_1}P_{U_2}P_{X_1|U_1}P_{X_2|U_2}P_{YZ|X_1X_2}$.
Let $C^{WCSI}$ be the secrecy capacity region of multiple access wiretap channels without channel state. 
\begin{theorem}[\cite{molavianjazi2009secure}]\label{the: mawc strong secrecy}
  The secrecy capacity region of multiple access wiretap channels satisfies
  \begin{align*}
    \left( \bigcup_{P_{U_1U_2X_1X_2}} \mathcal{R}_{01}(P_{U_1U_2X_1X_2})\cup \mathcal{R}_{02}(P_{U_1U_2X_1X_2})\cup \mathcal{R}_{03}(P_{U_1U_2X_1X_2}) \right) \subseteq C^{WCSI}.
  \end{align*}
\end{theorem}
Theorem \ref{the: mawc strong secrecy} is derived by using traditional wiretap codes, where the codebook is divided into some subcodebooks, each corresponding to a message. To transmit a message, the senders randomly select codewords from the corresponding subcodebooks. This random selection introduces some additional randomness in the encoding procedure, which is not available at the receiver side, and hence confuses the eavesdropper.
Region $\mathcal{R}_{02}$ occurs when the size of Sender 1's codebook satisfies $\widetilde{R}_1 \leq I(U_1;Z)$. In this case, Sender 1 cannot transmit any message secretly and hence, the secure achievable rate is 0. Region $\mathcal{R}_{03}$ follows similarly. 

In the rest of this section we provide our capacity results, which are the union of three regions $\mathcal{R}_i,i=1,2,3$ achieved by different coding schemes, each consisting of three sub-regions $\mathcal{R}_{ij},j=1,2,3$. The relation between region $\mathcal{R}_{i1}$ and $\mathcal{R}_{i2},\mathcal{R}_{i3}$ is the same as that in Theorem \ref{the: mawc strong secrecy} depending on the use of wiretap coding. The following table describes the main differences between these regions. More detailed discussions and the explanation of parameters are given in the remarks after Theorem \ref{the: mawc with casual csi}.
\begin{table}[h]
  \centering
  \caption{Comparison between different achievable regions}
  \label{table: Comparison between different achievable regions}
  \begin{tabular}{|c|c|c|c|}
    \hline
      & Secret Key Construction Method 1 & Secret Key Construction Method 2 & No Secret Key \\
     \hline
     Wiretap Coding at Both Encoders & $\mathcal{R}_{11}$ &$\mathcal{R}_{21}$ &$\mathcal{R}_{31}$ \\
    \hline
    Wiretap Coding Not Available at Encoder 1 & $\mathcal{R}_{12}$ & $\mathcal{R}_{22}$  & $\mathcal{R}_{32}$  \\
    \hline
    Wiretap Coding Not Available at Encoder 2 & $\mathcal{R}_{13}$ & $\mathcal{R}_{23}$&$\mathcal{R}_{33}$ \\
    \hline
  \end{tabular}
\end{table}

Given $P_S$ and $P_{YZ|X_1X_2S}$, define region $\mathcal{R}_{1i}(P_{VUU_1U_2X_1X_2|S}),i=1,2,3$, as sets of non-negative real number pairs $(R_1,R_2)$ such that
\begin{equation*}
  \mathcal{R}_{11}(P_{VUU_1U_2X_1X_2|S})=\left\{
    \begin{aligned}
      &R_1 \leq \min\{I(U_1;Y|V,U,U_2)-I(U_1;Z|S,U)+R_{11},I(U_1;Y|V,U,U_2)\};\\
      &R_2 \leq \min\{I(U_2;Y|V,U,U_1)-I(U_2;Z|S,U)+R_{21},I(U_2;Y|V,U,U_1)\};\\  
      &R_1 + R_2 \leq \min\{R_{SUM}-I(U_1,U_2;Z|S,U)+R_{11}+R_{21},R_{SUM}\}\\
      &R_{11}+R_{21} \leq I(V;Y)-I(V;U,Z).
    \end{aligned}
  \right.
\end{equation*}
\begin{equation*}
  \begin{split}
    &\mathcal{R}_{12}(P_{VUU_1U_2X_1X_2|S})=\left\{
    \begin{aligned}
      &R_1 \leq \min\{I(U_1;Y|V,U,U_2),R_{11}\};\\
      &R_2 \leq \min\{I(U_2;Y|V,U,U_1)-I(U_2;Z|S,U,U_1)+R_{21},I(U_2;Y|V,U,U_1)\};\\
      &R_1 + R_2 \leq \min\{R_{SUM},R_{SUM} - I(U_2;Z|S,U,U_1) + R_{21}, I(U_2;Y|V,U,U_1)-I(U_2;Z|S,U,U_1)+R_{11}+R_{21}\}\\
      &R_{11}+R_{21}\leq I(V;Y)-I(V;Z,U,U_1),
    \end{aligned}
  \right. \\
  &\text{where $R_{SUM}=\min\{I(U_1,U_2;Y|V,U),I(V,U,U_1,U_2;Y)-I(V;S)\}.$}\\
  &\mathcal{R}_{13}(P_{VUU_1U_2X_1X_2|S})=\left\{
    \begin{aligned}
      &R_1 \leq \min\{I(U_1;Y|V,U,U_2)-I(U_1;Z|S,U,U_2)+R_{11},I(U_1;Y|V,U,U_2)\};\\
      &R_2 \leq \min\{I(U_2;Y|V,U,U_1),R_{21}\};\\
      &R_1 + R_2 \leq \min\{R_{SUM},R_{SUM} - I(U_1;Z|S,U,U_2) + R_{11}, I(U_1;Y|V,U,U_2)-I(U_1;Z|S,U,U_2)+R_{11}+R_{21}\}\\
      &R_{11}+R_{21}\leq I(V;Y)-I(V;Z,U,U_2),
    \end{aligned}
  \right. 
  \end{split}
\end{equation*}
where $R_{SUM}=\min\{I(U_1,U_2;Y|V,U),I(V,U,U_1,U_2;Y)-I(V;S)\}$ and joint distribution such that $P_SP_{VUU_1U_2X_1X_2YZ|S}=P_{S}P_{V|S}P_{U}P_{U_1|U}P_{U_2|U}$ $P_{X_1|UU_1S}P_{X_2|UU_2S}P_{YZ|X_1X_2S}$. The explanation of $(R_{11},R_{21})$ is given in Remark \ref{rem: explanation of R1}.

 Let $\mathcal{R}_1$ be the convex hull of 
\begin{align*}
  \bigcup_{\substack{P_{VUU_1U_2X_1X_2|S}}} \bigcup_{i=1}^3\mathcal{R}_{1i}(P_{VUU_1U_2X_1X_2|S}).
\end{align*}

Further define $\mathcal{R}_{2i}(P_{U_1U_2X_1X_2|S}),i=1,2,3,$ as sets of non-negative real number pairs $(R_1,R_2)$ such that
\begin{equation*}
  \mathcal{R}_{21}(P_{U_1U_2X_1X_2|S})=\left\{
    \begin{aligned}
      &R_1 \leq \min\{I(U_1;Y|U_2)-I(U_1;Z|S)+R_{11},I(U_1;Y|U_2)-R_{12}\};\\
      &R_2 \leq \min\{I(U_2;Y|U_1)-I(U_2;Z|S)+R_{21},I(U_2;Y|U_1)-R_{22}\};\\ 
      &R_1+R_2 \leq \min\{I(U_1,U_2;Y)-I(U_1,U_2;Z|S)+H(S|Z)-H(S|U_1U_2Y),I(U_1,U_2;Y)-H(S|U_1U_2Y)\};\\
      &R_{11}+R_{12}+R_{21}+R_{22} \leq H(S|Z),\\
      &R_{12}+R_{22} \geq H(S|Y,U_1,U_2).
    \end{aligned}
  \right.
\end{equation*}
\begin{equation*}
  \begin{split}
    &\mathcal{R}_{22}(P_{U_1U_2X_1X_2|S})=\left\{
    \begin{aligned}
      &R_1 \leq \min\{R_{11},I(U_1;Y|U_2)-R_{12}\};\\
      &R_2 \leq \min\{I(U_2;Y|U_1)-I(U_2;Z|S,U_1)+R_{21},I(U_2;Y|U_1)-R_{22}\};\\  
      &R_1 + R_2 \leq \min\{I(U_2;Y|U_1)-I(U_2;Z|S,U_1)+H(S|Z,U_1)-H(S|U_1U_2Y),I(U_1,U_2;Y)-H(S|U_1U_2Y)\}\\
      &R_{11}+R_{12}+R_{21}+R_{22} \leq H(S|Z,U_1),\\
      &R_{12}+R_{22} \geq H(S|Y,U_1,U_2).
    \end{aligned}
  \right.\\
  &\mathcal{R}_{23}(P_{U_1U_2X_1X_2|S})=\left\{
    \begin{aligned}
      &R_1 \leq \min\{I(U_1;Y|U_2)-I(U_1;Z|S,U_2)+R_{11},I(U_1;Y|U_2)-R_{12}\};\\
      &R_2 \leq \min\{R_{21},I(U_2;Y|U_1)-R_{22}\};\\  
      &R_1 + R_2 \leq \min\{I(U_1;Y|U_2)-I(U_1;Z|S,U_2)+H(S|Z,U_2)-H(S|U_1U_2Y),I(U_1,U_2;Y)-H(S|U_1U_2Y)\}\\
      &R_{11}+R_{12}+R_{21}+R_{22} \leq H(S|Z,U_2),\\
      &R_{12}+R_{22} \geq H(S|Y,U_1,U_2).
    \end{aligned}
  \right.
  \end{split}
\end{equation*}
where the joint distribution $P_SP_{U_1U_2X_1X_2YZ|S}=P_SP_{U_1}P_{U_2}P_{X_1|U_1S}P_{X_2|U_2S}P_{YZ|X_1X_2S}$. Similarly, let $\mathcal{R}_2$ be the convex hull of
\begin{align*}
  \bigcup_{\substack{P_{U_1U_2X_1X_2|S}}} \bigcup_{i=1}^3\mathcal{R}_{2i}(P_{U_1U_2X_1X_2|S}).
\end{align*}
Finally, define $\mathcal{R}_{3i}(P_{VUU_1U_2X_1X_2|S}),i=1,2,3,$ as the sets of non-negative real number pairs $(R_1,R_2)$ such that
\begin{equation*}
  \begin{split}
    &\mathcal{R}_{31}(P_{VUU_1U_2X_1X_2|S})=\left\{ 
      \begin{aligned}
        &R_1>0,R_2>0;\\
        &R_1 \leq I(U_1;Y|V,U,U_2)-I(U_1;Z);\\
        &R_2 \leq I(U_2;Y|V,U,U_1)-I(U_2;Z);\\  
        &R_1+R_2 \leq I(U_1,U_2;Y|V,U)-I(U_1,U_2;Z);\\
        &R_1+R_2 \leq I(V,U,U_1,U_2;Y)-I(U_1,U_2;Z)-I(V;S);
      \end{aligned}
    \right.\\
    &\mathcal{R}_{32}(P_{VUU_1U_2X_1X_2|S})=\left\{ 
      \begin{aligned}
        &R_1=0,R_2>0;\\
        &R_2 \leq I(U_2;Y|V,U,U_1)-I(U_2;Z|U_1);\\  
      \end{aligned}
    \right.\\
    &\mathcal{R}_{33}(P_{VUU_1U_2X_1X_2|S})=\left\{ 
      \begin{aligned}
        &R_1>0,R_2=0;\\
        &R_1 \leq I(U_1;Y|V,U,U_2)-I(U_1;Z|U_2);\\  
      \end{aligned}
    \right.
  \end{split}
\end{equation*}
where the joint distribution $P_SP_{VUU_1U_2X_1X_2YZ|S}=P_SP_{V|S}P_UP_{U_1|U}P_{U_2|U}P_{X_1|UU_1S}P_{X_2|UU_2S}P_{YZ|X_1X_2S}$. Let $\mathcal{R}_3$ be the convex hull of 
\begin{align*}
  \bigcup_{\substack{P_{VUU_1U_2X_1X_2|S}}} \bigcup_{i=1}^3\mathcal{R}_{3i}(P_{VUU_1U_2X_1X_2|S})
\end{align*}
\begin{theorem}\label{the: mawc with casual csi}
  Let $\mathcal{R}^{CSI-E}$ be the convex hull of $\mathcal{R}_1 \cup \mathcal{R}_2 \cup \mathcal{R}_3$.
  The secrecy capacity of multiple access wiretap channels with causal channel state information at encoders $C^{CSI-E}$ satisfies
  \begin{align*}
    \mathcal{R}^{CSI-E} \subseteq C^{CSI-E}.
  \end{align*}
\end{theorem}
The achievability proof of region $\mathcal{R}_1$ is given in Section \ref{sec: coding scheme for R11}, and the achievability proof of region $\mathcal{R}_2$ is provided in Section \ref{codin scheme for R21}. The achievability of region $\mathcal{R}_3$ follows from setting the secret key rates of the two senders $R_{11}=R_{21}=0$ in region $\mathcal{R}_1$ and skipping \emph{Key Message Codebook Generation} in Section \ref{sec: coding scheme for R11} since the key rate is set to zero.

Regions $\mathcal{R}_{i},i=1,2,3$, are achieved by using three different coding schemes. We explain the difference between those regions and coding schemes in the following remarks. 

Throughout the paper, we refer to the \textbf{coding scheme for region $\mathcal{R}_i$ as Coding scheme $i,i=1,2,3.$}
\begin{remark}\label{remark: remark for three regions}
  The region in Theorem \ref{the: mawc with casual csi} is the convex hull of the union of three different regions, each consisting of three sub-regions. Here we explain the differences between them. According to the use of the secret key, regions are divided into two groups: $(\mathcal{R}_1,\mathcal{R}_2)$ and $\mathcal{R}_3$. In $\mathcal{R}_3$ only wiretap codes are used to protect the messages. Regions $\mathcal{R}_1$ and $\mathcal{R}_2$ correspond to two different ways to construct secret keys, one using Wyner-Ziv coding and one using Slepian-Wolf coding. Hence, the three different regions in Theorem \ref{the: mawc with casual csi} correspond to the different use of secret keys. The model considered in this paper is a multiple access channel, where the communication conditions of each sender may be different. For the case that Sender 1 cannot use wiretap coding to protect the messages, region $\mathcal{R}_{1}$ reduces to region $\mathcal{R}_{12}$, and for the case that Sender 2 cannot use wiretap coding, region $\mathcal{R}_{1}$ reduces to region $\mathcal{R}_{13}$. Regions $\mathcal{R}_{2}$ and $\mathcal{R}_3$ follow similarly.
\end{remark}
\begin{remark}\label{rem: explanation of R1}
Explanation of region $\mathcal{R}_1$: Region $\mathcal{R}_1$ is a union of three regions.
 The coding scheme for achieving $\mathcal{R}_1$ uses wiretap codes and secret key agreement. The region
\begin{equation*}
    \mathcal{R}_{11}'(P_{VUU_1U_2X_1X_2|S})=\left\{
      \begin{aligned}
        &R_1>0,R_2>0;\\
        &R_1 \leq I(U_1;Y|V,U,U_2);\\
        &R_2 \leq I(U_2;Y|V,U,U_1);\\  
        &R_1+R_2 \leq I(U_1,U_2;Y|V,U);\\
        &R_1+R_2 \leq I(V,U,U_1,U_2;Y)-I(V;S);
      \end{aligned}
    \right.
\end{equation*}
given $P_S$ and $P_{YZ|X_1X_2S}$ for joint distribution such that $P_{VUU_1U_2X_1X_2|S}=P_{V|S}P_UP_{U_1|U}P_{U_2|U}P_{X_1|UU_1S}P_{X_2|UU_2S}$ is an achievable region for multiple access channels with causal channel state information at encoders given in \cite{lapidoth2012multiple}. It is proved that such a region outperforms the region obtained by using Shannon strategy. The coding scheme for this region considers a block Markov coding with backward decoding. The auxiliary random variable $V$ is used to describe a lossy description of the channel state and also for constructing secret keys to encrypt messages, and $U$ represents a common message based on the state information from the last block observed by both encoders. This common message is used to convey the Wyner-Ziv coding index of the lossy description. The subtraction terms $I(U_1;Z|S,U),I(U_2;Z|S,U)$ and $I(U_1,U_2;Z|S,U)$ are the secrecy cost in the wiretap codes. The term $I(V;Y)-I(V;Z,U)$ is the gain of using a secret key constructed from $V$. The secret key is divided into two independent parts $(R_{11},R_{21})$ and assigned to the senders. We consider here the worst case assuming that the eavesdropper always has knowledge of $U$. Note that the model considered in this paper is a common state model, i.e. the state sequence is the same for Sender 1 and Sender 2. Hence, the secret key cannot be used to encrypt both senders' messages directly. Otherwise, the encryption will introduce dependence between the two independent messages which leads to more information leakage to the eavesdropper since in this case observing the encrypted message of Sender 2 may give some information about the secret key, which is also used to encrypt the message of Sender 1, and leads to the information leakage of Sender 1's message. Instead, we split the secret key into two independent parts and use them to encrypt senders' messages separately. Region $\mathcal{R}_{12}$ occurs in the case that the size of Sender 1's secret message codebook $2^{n\widetilde{R}_1}$ satisfies $\widetilde{R}_1 < I(U_1;Z|S,U)$. As we explained before, $I(U_1;Z|S,U)$ is the local randomness introduced at Sender 1 to confuse the eavesdropper. The constraint implies the size of the codebook is within the decoding ability of the eavesdropper. We put an additional $U$ in the condition by the worst-case assumption that the eavesdropper is capable to decode the common message encoded using $U$. Although similar to region $\mathcal{R}_{02}$ that Sender 1 cannot use wiretap codes to randomize its codewords in this case, Sender 1 can still achieve a positive secrecy achievable rate by using a secret key. Similarly, for the case that the size of Sender 2's secret message codebook $2^{n\widetilde{R}_2}$ satisfies $\widetilde{R}_2 < I(U_2;Z|S,U)$, region $\mathcal{R}_{13}$ is achievable.
\end{remark}

\begin{remark}\label{rem: worst case assumption}
  In this remark, we discuss our worst-case assumption that the eavesdropper always has the common message $U$. The common message is used to convey the Wyner-Ziv index of the lossy description of the state sequence from the last block. Hence, at Block $b$ the selection of the common message codeword is determined by the Wyner-Ziv index, denoted by $k_{b-1}$, of the lossy description $\boldsymbol{v}_{b-1}$ which describes the state sequence from Block $b-1$. To make sure each index has a corresponding codeword, the size of the common message codebook $2^{nR_0}$ satisfies $R_0>I(V;S)-I(V;Y)$ by Wyner-Ziv coding theorem. However, this one-to-one mapping between the index and the common message codeword makes it possible for the eavesdropper to decode the common codeword directly. By channel coding theorem, this is determined by the terms $I(V;S)-I(V;Y)$ and $I(U,Z)$, where it is in general unclear which term is larger. For the case that $I(U,Z)>I(V;S)-I(V;Y)$, the eavesdropper can find the unique $\boldsymbol{u}(k_{b-1})$ and also the Wyner-Ziv index. Thus, it is necessary to make the worst-case assumption that the eavesdropper always has the common message codeword\footnote{We think this argument is missing in \cite{sonee2014wiretap}, where the Wyner-Ziv index is used as the secret key to encrypt the message and conveyed by the cloud center codeword in their superposition coding scheme. As discussed in the remark, we think the eavesdropper may be able to decode this cloud center codeword and thereby have access to the secret key.}.
\end{remark}
\begin{remark}
  Explanation of region $\mathcal{R}_2$: Region $\mathcal{R}_{2}$ is a union of three regions, where regions $\mathcal{R}_{22}$ and $\mathcal{R}_{23}$ occur when one of the senders cannot use wiretap coding. The basic idea of the coding scheme is also a combination of wiretap codes and secret key agreement, which is similar to the coding scheme in \cite{sasaki2019wiretap}. Secret key splitting is still used here to encrypt senders' messages. However, region $\mathcal{R}_{2}$ considers a block Markov coding scheme with forward decoding. If we remove the secrecy constraint and the secret key agreement, the reduced region is in fact the achievable rate region of an SD-MAC with causal CSI at encoders using Shannon strategy, which is proved not optimal for the considered model \cite{lapidoth2012multiple}. Setting $U_1=\emptyset$, region $\mathcal{R}_2$ reduces to the secrecy achievable rate given in \cite{sasaki2019wiretap}. Regions $\mathcal{R}_1$ and $\mathcal{R}_2$ do not include each other in general. In Section \ref{sec: examples}, we provide a numerical example to show that in some cases, points falling into region $\mathcal{R}_1$ do not fall into region $\mathcal{R}_2$ and vice versa.
\end{remark}
\begin{remark}
  In \cite{lapidoth2012multiple}, it is proved that the achievable region for SD-MACs with strictly causal/causal CSI can be larger by an improved coding scheme using block Markov coding with Gasptar's results on the compression of correlated sources with side information\cite{gastpar2004wyner}. The encoder compresses the information about the state sequence and input sequence of the last block to help decoding. This improved coding scheme is also used for SD-MACs with double states and outperforms the Shannon strategy\cite{lapidoth2012multiple2}. However, the compression step compresses codewords from the last block, which are related to the transmitted messages, and the secret keys constructed based on these lossy descriptions are also related to the previous messages. This relation between the constructed secret keys and transmitted messages makes the information leakage can not be bounded. It is still an open problem to apply this improved coding scheme to SD-MAWCs with causal CSI at encoders.
\end{remark}
\begin{remark}
  Explanation of region $\mathcal{R}_3$: As before, region $\mathcal{R}_3$ is the union of three regions. The region is obtained by block Markov coding with backward decoding\cite{lapidoth2012multiple} and wiretap coding\cite{molavianjazi2009secure}. Regions $\mathcal{R}_{32}$ and $\mathcal{R}_{33}$ correspond to the case that one of the senders does not use wiretap coding. Note that in Coding scheme 3, there is no secret key agreement between the senders and the receiver. Hence, no message can be secretly transmitted if wiretap coding is not used.  If we remove the randomness in the state by setting $S$ to be deterministic, block coding with backward decoding is not necessary and region $\mathcal{R}_3$ reduces to Theorem \ref{the: mawc strong secrecy}.
\end{remark}
In Section \ref{sec: examples}, numerical examples show that some points falling in region $\mathcal{R}_1$ cannot be achieved by the coding scheme used for region $\mathcal{R}_2$ and vice versa. When $I(V;Y)-I(V;Z,U)<0$ or $H(S|Z)-H(S|UY)\leq 0$, the sender cannot construct any secret key, and only use wiretap codes to protect the messages. In this case, regions $\mathcal{R}_1$ and $\mathcal{R}_2$ are empty and block Markov coding with wiretap coding can be used to derive the achievable region $\mathcal{R}_3$. Further note that by setting $V=U=\emptyset$ and $S$ to be deterministic, Theorem \ref{the: mawc with casual csi} reduces to Theorem \ref{the: mawc strong secrecy}.

The next theorem states an outer bound of the SD-MACs with noncausal CSI (and also holds for the causal case since one can regard the causal case coding as a special case of the noncausal case by using state information in a causal manner. Hence, any code for the causal case can also be used for the noncausal case). Define region $\mathcal{O}(P_{UU_1U_2X_1X_2|S})$ as 
\begin{equation*}
  \mathcal{O}(P_{UU_1U_2X_1X_2|S}) = \left\{
    \begin{aligned}
      &R_1 \leq I(U_{1};Y|U_{2},U) - I(U_{1};Z|U_{2},U)\\
      &R_2 \leq I(U_{2};Y|U_{1},U) - I(U_{2};Z|U_{1},U)\\
      &R_1 + R_2 \leq I(U_{1},U_{2};Y|U) - I(U_{1},U_{2};Z|U),
    \end{aligned} 
  \right.
\end{equation*}
where the mutual information is according to the joint distribution $P_SP_{UU_1U_2X_1X_2|S}=P_SP_{UU_1U_2|S}P_{X_1|UU_1S}P_{X_2|UU_2S}$. Let $\mathcal{O}^{NCSI-E}:= \cup_{P_{UU_1U_2X_1X_2|S}}\mathcal{O}(P_{UU_1U_2X_1X_2|S})$
\begin{theorem}\label{the: outer bound}
  The secrecy capacity of the multiple access wiretap channel with noncausal CSI satisfies
  \begin{align*}
    C^{NCSI-E} \subseteq \mathcal{O}^{NCSI-E}.
  \end{align*}
\end{theorem}
The proof is given in Appendix \ref{app: proof of outer bound}. The capacity of the model considered in this paper is still an open problem since there is a gap between the inner and outer bounds. In fact, the capacity of multiple access channels with causal/strictly causal CSI without secrecy constraint has also not yet been found\cite{lapidoth2012multiple}, and only capacity results of some special cases have been reported. For example, the capacity of deterministic MAC with strictly causal CSI at encoders was reported in \cite{li2012multiple}, and the capacity of causal asymmetric SD-MAC with degraded message sets was provided in \cite{somekh2008cooperative}. In Section \ref{sec: examples}, we identify some special cases where the secrecy capacity is achieved by considering degraded message sets, and our results recover some existing works in the literature.

\section{coding scheme for $\mathcal{R}_{1}$}\label{sec: coding scheme for R11}
In this section, we provide the coding scheme for region $\mathcal{R}_{1}$ and its reliability and security analysis. Block coding schemes are widely used for different types of discrete memoryless channels \cite{lapidoth2012multiple,chia2012wiretap,fujita2016secrecy,sasaki2019wiretap}, where the state sequence from the last block is used to encrypt part of the message in the current block and block coding also helps the reliable transmission by compressing the state and codewords from the last block\cite{lapidoth2012multiple}\cite{lapidoth2012multiple2}. We consider a block coding scheme with backward decoding, which has been used for SD-MACs with causal CSI at encoders \cite{lapidoth2012multiple}. For the secrecy part, the main difference between the coding scheme below and that in \cite{sasaki2019wiretap} is that instead of constructing the secret key from state sequences directly, we construct the secret key from a lossy description of the state sequence. 

The coding scheme with backward decoding uses $B+1$ blocks. In block $2\leq b \leq B+1$, upon observing the channel state sequence $\boldsymbol{s}_{b-1}$ in the last block $b-1$, the encoder finds a lossy description $\boldsymbol{v}_{b-1}$ of the state sequence and constructs a secret key from $\boldsymbol{v}_{b-1}$. In Blocks 1 and $B+1$, no meaningful message is transmitted. We further assume that $I(U_1,U_2;Y)>\mu>0$ for some positive number $\mu$, otherwise, as shown by \cite[(73b)]{lapidoth2012multiple}, the rates of both senders are zero. The codeword length of blocks 1 to $B$ is $n$ and the length of the last block is  
\begin{align}
  \label{eq: codeword length of last block}\widetilde{n}= n\frac{R_0}{\mu},
\end{align}
where $R_0>0$ is a positive real number such that $2^{nR_0}$ is the size of the auxiliary message codebook defined in the following paragraphs. The last block is only used to transmit the information of state sequence $\boldsymbol{s}_{B}$ and the setting of $\widetilde{n}$ enables the decoder to decode the information correctly. Assume $I(V;Y)-I(V;Z)>\tau_1> 0$ with Markov chain relation $V-S-(Y,Z)$. Otherwise, the secret key rate is zero and the resulting achievable region is $\mathcal{R}_{3}$. The senders cooperate with each other since both of them observe the same channel states, and send a common message to convey the lossy description of the state sequence they observe. Based on this lossy description of the channel state, the decoder is able to decode more messages and this leads to a larger achievable rate region.

The backward decoding follows the principles in \cite[Appendix G]{lapidoth2012multiple} with an additional secret key decryption due to the secret key agreement. The decoder first decodes Block $B+1$ to get the index of the state lossy description in Block $B$. In the subsequent decoding block $b,b\in[1:B]$, the decoder first finds the lossy description of the current state sequence and then performs joint typical decoding to decode the common and private messages with the help of this lossy description. The information about the lossy description of state sequence $\boldsymbol{s}_{b-1}$ is contained in the common message and hence, help the decoding in the next decoding block.

Before giving the coding scheme, we further present an intuitive explanation of the secret key agreement in the coding scheme. For simplicity, we only consider the encryption and decryption operations of Sender 1. In each block $b, b\in[2:B]$, upon observing the channel state sequence $\boldsymbol{s}_{b-1}$ and its lossy description $\boldsymbol{v}_{b-1}$, the sender generates a secret key $k_{1,b}$ by a mapping $\kappa: \mathcal{T} \to \mathcal{K}_1$, where $\mathcal{T}$ is the index set of $\boldsymbol{V}$ in Wyner-Ziv subcodebooks. This key is used to encrypt the message $m_{1,b}$ by computing $c_{1,b}=m_{1,b} \oplus k_{1,b}$, where $\oplus$ is the modulo sum. In the backward decoding step block $b, b\in[1:B]$, the decoder has the knowledge about the Wyner-Ziv index $k_{0,b+1}$ from the last block $b+1$. Now it is possible to find $\boldsymbol{v}_{b}$, which is the lossy description of the state sequence $\boldsymbol{s}_b$ in the current block. The decoder then reproduces the secret key $k_{1,b+1}$ used for block $b+1$ and decrypts the message $m_{1,b+1}$ by computing $m_{1,b+1}=c_{1,b+1} \ominus k_{1,b+1}$.

\begin{figure}
  \centering
  \includegraphics[scale=0.65]{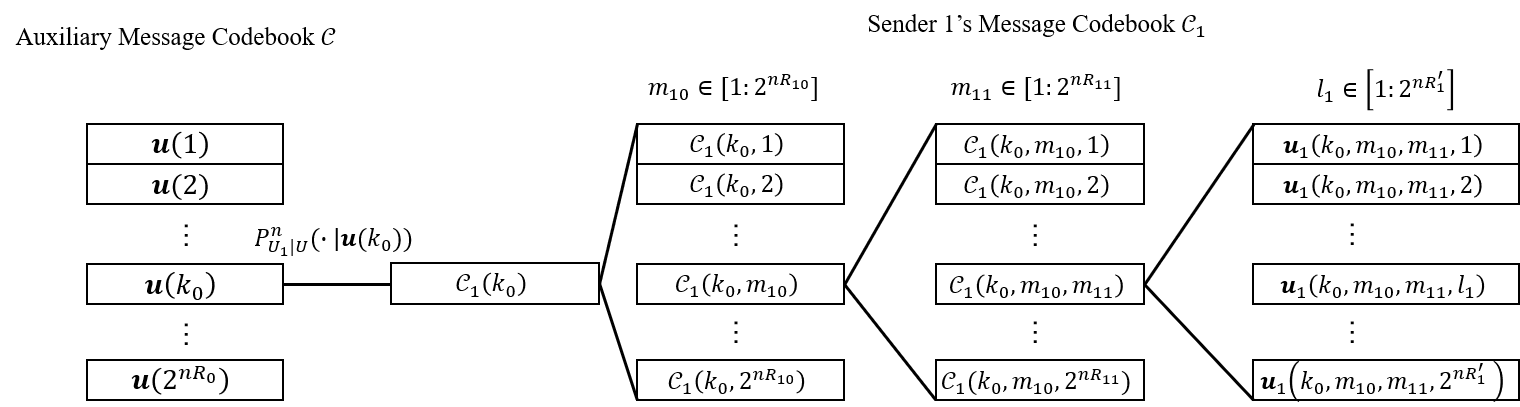}
  \caption{Codebook partition of Sender 1's message codebook: The encoder first chooses the common message $\boldsymbol{u}(k_0)$ by observing the state sequence from the last block and finding its Wyner-Ziv index $k_0$. Then, the encoder uses the codebook generated by this common message codeword and divides each message $M_1$ into two independent parts $(M_{10},M_{11})$. The codebook is partitioned into two-layer sub-codebooks. The first layer is indexed by message $M_{10}$ and the second layer is indexed by the encrypted version of message $M_{11}$.}
  \label{fig: coding scheme 1 codebook partition}
\end{figure}

Given $P_S$ and channel $P_{YZ|X_1X_2S}$, let $R_0, \widetilde{R}_1,\widetilde{R}_2,R_{10},R_{11},R_{20},R_{21},R_{K_1}$ be positive real numbers with constraints
\begin{align*}
  R_0 \geq I(V;S)-I(V;Y),&\\
  \widetilde{R}_1 \leq I(U_1;Y|V,U,U_2),&\\
  \widetilde{R}_2 \leq I(U_2;Y|V,U,U_1),&\\
  \widetilde{R}_1 - R_{10} > I(U_1;Z|S,U),&\\
  \widetilde{R}_2 - R_{20} > I(U_2;Z|S,U),&\\
  \widetilde{R}_1 + \widetilde{R}_2 - R_{10} - R_{20} > I(U_1,U_2;Z|S,U),&\\
  R_{K_1}=R_{11} + R_{21} \leq I(V;Y) - I(V;Z),&\\
  R_1 = R_{10}+R_{11},&\\
  R_{2}=R_{20}+R_{21}.&
\end{align*}
under fixed joint distribution $$P_{S}P_{V|S}P_UP_{U_1|U}P_{U_2|U}P_{X_1|UU_1S}P_{X_2|UU_2S}P_{YZ|X_1X_2S}.$$
\begin{figure}
  \centering
  \includegraphics[scale=0.8]{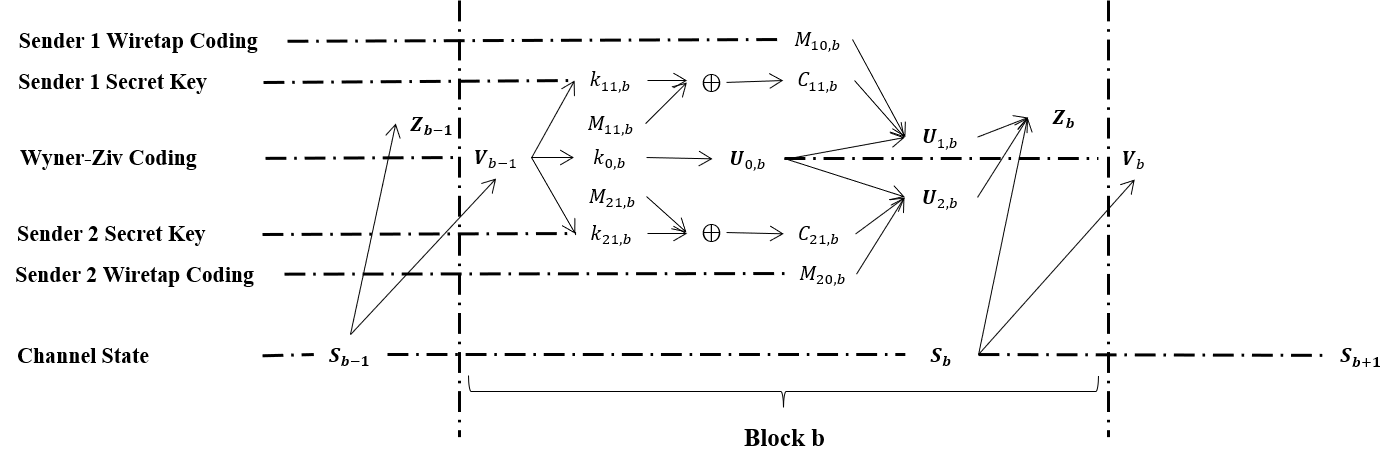}
  \caption{Encoding of Coding scheme  1. The arrows represent the dependent relation between random variables and the dashed lines represent the random variables connected on the same horizontal line.}
  \label{fig: coding shcheme 1 encoding}
\end{figure}
\begin{figure}
  \centering
  \includegraphics[scale=0.8]{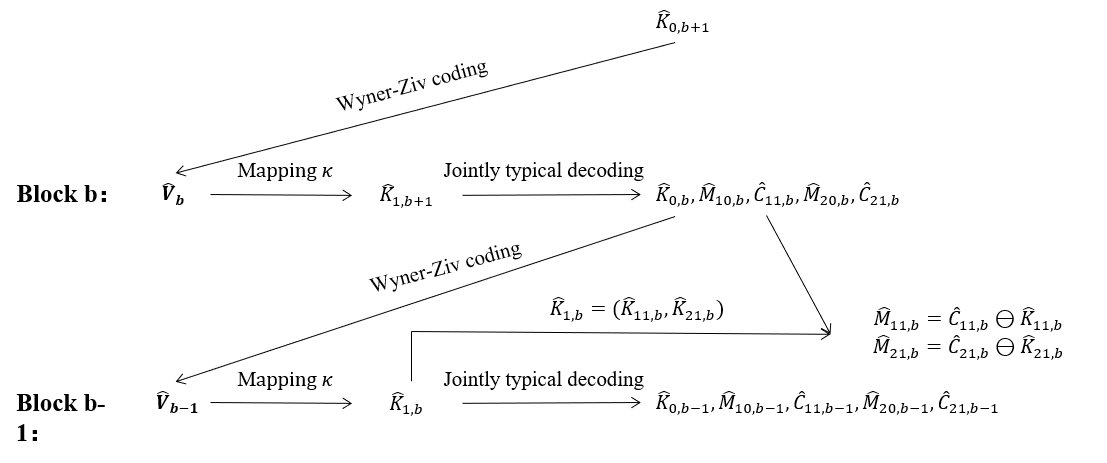}
  \caption{Backward decoding of the coefficients in Coding scheme  1.}
  \label{fig: coding shcheme 1 decoding}
\end{figure}
\emph{Key Message Codebook Generation: } Given $\tau>0$ and $\delta>0$, let $R_K=I(V;S)+\tau$. In each block $2\leq b \leq B+1$, the sender generates a codebook $\mathcal{C}_{K_b}=\{\boldsymbol{v}(l)\}_{l=1}^{2^{nR_{K}}}$ consists of $2^{nR_K}$ codewords, each i.i.d. generated according to distribution $P_V$ such that $P_V(v)=\sum_{s\in\mathcal{S}}P_S(s)P_{V|S}(v|s)$ for any $v\in\mathcal{V}$. Partition the codebook $\mathcal{C}_{K_b}$ into $2^{nR_{K_0}}$ sub-codebooks $\mathcal{C}_{K_b}(k_{0,b})$, where $k_{0,b}\in[1:2^{nR_{K_0}}]$ and $R_{K_0}=I(V;S)-I(V;Y)+2\tau$. Let $\mathcal{T}$ be the index set of codewords in each subcodebook $\mathcal{C}_{K_b}(k_{0,b})$ such that $|\mathcal{T}|=|\mathcal{C}_{K_b}(k_{0,b})|$ for any $k_{0,b}\in[1:2^{nR_{K_0}}]$. For each codebook $\mathcal{C}_{K_b},$ construct a secret key mapping $\kappa:\mathcal{T}\to [1:2^{nR_{K_1}}].$ Denote the resulted secret key by $K_{1,b}$.

\emph{Auxiliary Message Codebook Generation: } For each block $b$, generate auxiliary message codebook $\mathcal{C}_b=\{\boldsymbol{u}(m_0)\}_{m_0=1}^{2^{nR_0}}$ i.i.d. according to distribution $P_{U}$, where $R_0=I(V;S)-I(V;Y)+3\tau$.  

\emph{Message Codebook Generation:}
\begin{enumerate}
  \item Blocks $b\in[1:B]$: For each $m_0$, generate codebook $\mathcal{C}_{1,b}(m_0)=\{\boldsymbol{u}_1(m_0,l)\}_{l=1}^{2^{n\widetilde{R}_1}}$ containing $2^{n\widetilde{R}_1}$ codewords, each i.i.d. generated according to distribution $P_{U_1|U}$. Partition each $\mathcal{C}_{1,b}(m_0)$ into $2^{nR_{10}}$ subcodebooks $\mathcal{C}_{1,b}(m_0,m_{10})$, where $m_{10}\in[1:2^{nR_{10}}].$ For each subcodebook $\mathcal{C}_{1,b}(m_0,m_{10})$, partition it into two-layer subcodebooks $\mathcal{C}_{1,b}(m_0,m_{10},m_{11})=\{\boldsymbol{u}_1(m_0,m_{10},m_{11},l_1)\}_{l_1=1}^{2^{nR_1'}},$ where $m_{11}\in[1:2^{nR_{11}}],R_1' := \widetilde{R}_1-R_{10}-R_{11}$. Likewise, generate codebook $\mathcal{C}_{2,b}=\{\boldsymbol{u}_2(m_0,l)\}_{l=1}^{2^{n\widetilde{R}_2}}$ with codewords i.i.d. generated according to $P_{U_2|U}$, and then partition it into two-layer sub-codebooks $\mathcal{C}_{2,b}(m_0,m_{20},m_{21})=\{\boldsymbol{u}_2(m_0,m_{20},m_{21},l_2)\}_{l_2=1}^{2^{nR_2'}}$, where $m_{20}\in[1:2^{nR_{20}}],m_{21}\in[1:2^{nR_{21}}],R_2':=\widetilde{R}_2-R_{20}-R_{21}$. The codebook partition is also presented in Fig. \ref{fig: coding scheme 1 codebook partition}.
  \item Block $B+1$. For $k=1,2$, generate codebooks $\mathcal{C}_{k,B+1}$ as above with codeword length $\widetilde{n}$ defined as in \eqref{eq: codeword length of last block}.
\end{enumerate}


The above codebooks are all generated randomly and independently. Denote the set of random codebooks in each block $b$ by $\bar{\textbf{C}}_b$.

\emph{Encoding: } The encoding process is illustrated in Fig. \ref{fig: coding shcheme 1 encoding}.
\begin{enumerate}
  \item Block 1. Setting $m_{0,1}=m_{10,1}=m_{20,1}=m_{11,1}=m_{21,1}=1$, the encoder $j$ picks an index $l_j\in[1:2^{nR_j'}]$ uniformly at random, $j=1,2$. The codeword $\boldsymbol{x}_j$ is generated by $(\boldsymbol{u}(1),\boldsymbol{u}_j(1,1,1,l_j),\boldsymbol{s})$ according to $P^n_{X_j|UU_jS}(\boldsymbol{x}_j|\boldsymbol{u},\boldsymbol{u}_j,\boldsymbol{s})=\prod_{i=1}^n P_{X_j|UU_jS}(x_{ji}|u_i,u_{ji},s_i),j=1,2$. Here we omit the indices of the codewords. 
  \item Blocks $b\in[2:B]$. Upon observing the state sequence $\boldsymbol{s}_{b-1}$ in the last block, the encoders find a sequence $\boldsymbol{v}_{b-1}$ such that $(\boldsymbol{s}_{b-1},\boldsymbol{v}_{b-1})\in T^n_{P_{SV},\delta}$ and set $m_{0,b}=k_{0,b}$, where $k_{0,b}$ is the index of subcodebook $\mathcal{C}_{k_{b-1}}(k_{0,b})$ containing $\boldsymbol{v}_{b-1}$. We also write sequence $\boldsymbol{v}_{b-1}$ as $\boldsymbol{v}(k_{0,b},t_b)$ if $\boldsymbol{v}_{b-1}$ is the $t_b$-th sequence in sub-codebook $\mathcal{C}_{K_{b-1}}(k_{0,b})$. Generate the secret key $k_{1,b}=\kappa(t_b)$ and then split it into two independent parts $(k_{11,b},k_{21,b})\in[1:2^{nR_{11}}]\times[1:2^{nR_{21}}]$. To transmit message $m_{1,b}$, Encoder $1$ splits it into two independent parts $(m_{10,b},m_{11,b})$ and computes $c_{11,b} = m_{11,b} \oplus k_{11,b} \pmod{2^{nR_{11}}}$. The encoder  selects an index $l_{1,b}\in[1:2^{nR_1'}]$ uniformly at random and generates the codeword $\boldsymbol{x}_1$ by $P^n_{X_1|UU_1S}(\boldsymbol{x}_1|\boldsymbol{u}(k_{0,b}),\boldsymbol{u}_1(k_{0,b},m_{10,b},c_{11,b},l_{1,b}),\boldsymbol{s})$.
  Likewise, the codeword $\boldsymbol{x}_2$ for Sender 2 is generated by $P^n_{X_2|UU_2S}(\boldsymbol{x}_2|\boldsymbol{u}(k_{0,b}),\boldsymbol{u}_2(k_{0,b},m_{20,b},c_{21,b},l_{2,b}),\boldsymbol{s}),$ where $c_{21,b} = m_{21,b} \oplus k_{21,b} \pmod{2^{nR_{21}}}$.
  \item Block $B+1$. Upon observing the state sequence $\boldsymbol{s}_{B}$ in the last block, the encoder finds a sequence $\boldsymbol{v}(k_{0,B+1},t_{B+1})$ such that $(\boldsymbol{s}_{B},\boldsymbol{v}(k_{0,B+1},t_{B+1}))\in T^n_{P_{SV},\delta}$. The encoders then set $m_{0,B+1}=m_{10,B+1}=m_{11,B+1}=m_{20,B+1}=m_{21,B+1}=1$ and generate codewords $\boldsymbol{x}_1$ and $\boldsymbol{x}_2$ according to distributions $P^n_{X_1|UU_1S}(\boldsymbol{x}_1|\boldsymbol{u}(k_{0,B+1}),\boldsymbol{u}_1(k_{0,B+1},1,1,1),\boldsymbol{s})$ and $P^n_{X_2|UU_2S}(\boldsymbol{x}_2|\boldsymbol{u}(k_{0,B+1}),\boldsymbol{u}_2(k_{0,B+1},1,1,1),\boldsymbol{s})$.

\end{enumerate}

\emph{Backward Decoding: } The backward decoding process of the coefficients is illustrated in Fig. \ref{fig: coding shcheme 1 decoding}.
\begin{enumerate}
  \item Block $B+1$. The decoder looks for a unique $\hat{k}_{0,B+1}$ such that $(\boldsymbol{u}_1(\hat{k}_{0,B+1},1,1,1),\boldsymbol{u}_2(\hat{k}_{0,B+1},1,1,1),\boldsymbol{y}_{B+1})\in T^n_{P_{U_1U_2Y},\delta}$.
  \item Blocks $b\in[1:B]$. The decoder has the knowledge about $\hat{k}_{0,b+1}$ from the last block. It tries to find a unique $\boldsymbol{v}_{b}=\boldsymbol{v}(\hat{k}_{0,b+1},\hat{t}_{b+1})$ such that $(\boldsymbol{v}(\hat{k}_{0,b+1},\hat{t}_{b+1}),\boldsymbol{y}_b)\in T^n_{P_{VY},\delta}$. Using $(\hat{k}_{11,b+1},\hat{k}_{21,b+1})=\kappa(\hat{t}_{b+1})$, the decoder now computes $\hat{m}_{11,b+1}=\hat{c}_{11,b+1} \ominus \hat{k}_{11,b+1} \pmod{2^{nR_{11}}}$ and $\hat{m}_{21,b+1}=\hat{c}_{21,b+1} \ominus \hat{k}_{21,b+1} \pmod{2^{nR_{21}}}$. For block $b,b\in[2:B]$, with the help of $\boldsymbol{v}(\hat{k}_{0,b+1},\hat{t}_{b+1})$, the decoder looks for a unique tuple $(\hat{k}_{0,b},\hat{m}_{10,b},\hat{m}_{20,b},\hat{c}_{11,b},\hat{c}_{21,b},\hat{l}_{1,b},\hat{l}_{2,b})$ such that 
  \begin{align*}
    &(\boldsymbol{v}_b,\boldsymbol{u}(\hat{k}_{0,b}),\boldsymbol{u}_1(\hat{k}_{0,b},\hat{m}_{10,b},\hat{c}_{11,b},\hat{l}_{1,b}),\boldsymbol{u}_2(\hat{k}_{0,b},\hat{m}_{20,b},\hat{c}_{21,b},\hat{l}_{2,b}),\boldsymbol{y}_b)\in T^n_{P_{VUU_1U_2Y},\delta}.
  \end{align*}
  \item Block 1. The messages transmitted in Block 1 are dummy messages. Hence, the decoding of Block 1 need not be performed.
\end{enumerate}


\textbf{Error Analysis: } The proof of the reliability is similar to that in \cite[Appendix G]{lapidoth2012multiple}. To bound the decoding error probability of the coding scheme, we define the following events.
\begin{align*}
  &\textbf{E}_{1,b} = \{(\boldsymbol{V}_{b},\boldsymbol{S}_{b})\in T^n_{P_{SV},\delta} \text{for some $\boldsymbol{V}_b\in \mathcal{C}_{K_b}$}\},\\
  &\textbf{E}_{2,B+1}(k_{0,B+1}) = \{(\boldsymbol{U}_1(k_{0,B+1},1,1,1),\boldsymbol{U}_2(k_{0,B+1},1,1,1),\boldsymbol{Y}_{B+1})\in T^n_{P_{U_1U_2Y},\delta}\},\\
  &\textbf{E}_{3,b}(t_{b+1}) = \{(\boldsymbol{V}(k_{0,b+1},t_{1,b+1}),\boldsymbol{Y}_{b})\in T^n_{P_{VY},\delta}\}\\
  &\textbf{E}_{4,b}(k_{0,b},\widetilde{m}_{1,b},\widetilde{m}_{2,b}) = \{(\boldsymbol{V}(k_{0,b+1},t_{1,b+1}),\boldsymbol{U}(k_{0,b}),\boldsymbol{U}_1(\widetilde{m}_{1,b}),\boldsymbol{U}_2(\widetilde{m}_{2,b}),\boldsymbol{Y}_{b})\in T^n_{P_{VUU_1U_2Y},\delta}\}
\end{align*}
where $\widetilde{m}_{i,b}=(k_{0,b},m_{i0,b},c_{i1,b},l_{i,b}),i=1,2$.
Further define
\begin{align*}
  &\textbf{D}_{1} = \bigcap_{b=1}^{B} \textbf{E}_{1,b},\\
  &\textbf{D}_{2,B+1} =  \textbf{E}_{2,B+1}(k_{0,B+1})^c \bigcup_{k_{0,B+1}'\neq k_{0,B+1}} \textbf{E}_{2,b}(k_{0,B+1}'),\\
  &\textbf{D}_{3,b} =  \textbf{E}_{3,b}(t_{b+1})^c \bigcup_{t_{b+1}'\neq t_{b+1}}\textbf{E}_{3,b}(t_{b+1}'),\\
  &\textbf{D}_{4,b} =  \textbf{E}_{4,b}(k_{0,b},\widetilde{m}_{1,b},\widetilde{m}_{2,b})^c \bigcup_{(k_{0,b}',\widetilde{m}_{1,b}',\widetilde{m}_{2,b}')\neq (k_{0,b},\widetilde{m}_{1,b},\widetilde{m}_{2,b})} \textbf{E}_{1,b}(k_{0,b}',\widetilde{m}_{1,b}',\widetilde{m}_{2,b}'),
\end{align*}

For event $\textbf{D}_1,$ note that $Pr\{\textbf{D}_1^c\}\leq \cup_{b=1}^{B-1}Pr\{\textbf{E}_{1,b}^c\}$. By covering lemma\cite{el2011network}, it follows that $Pr\{\textbf{D}_1^c\}\to 0$ as $n\to\infty$ when $R_K = I(S;V)+\tau.$
It also follows that $Pr\{\textbf{D}_{2,B+1}\}\to 0$ as $n\to\infty$ given $R_0 < \frac{\widetilde{n}}{n}I(U_1,U_2;Y)$.

Further note that $Pr\{\textbf{D}_{3,b}\}\leq Pr\{ \textbf{E}_{3,b}(t_{b+1})^c\} + \bigcup_{t_{1,b+1}'\neq t_{1,b+1}}Pr\{\textbf{E}_{3,b}(t_{b+1}')\}$. Applying Markov Lemma\cite{thomas2006elements} yields $Pr\{\textbf{E}_{3,b}(t_{b+1})^c\}\to 0$ as $n\to \infty$. Note that in Block $b$ the decoder has the knowledge about $k_{0,b+1}$ from the decoding of Block $b+1$, so it tries to find a unique $t_{b+1}$ in $\mathcal{C}_{K_{b+1}}$ such that $(\boldsymbol{v}(k_{0,b+1},t_{b+1}),\boldsymbol{y}_b)\in T^n_{P_{VY},\delta}$. By Wyner-Ziv Theorem\cite{wyner1976rate}, it follows that $\bigcup_{t_{0,b+1}' \neq t_{0,b+1}}Pr\{\textbf{E}_{3,b}(t_{0,b+1}')\}\to 0$ since 
\begin{align*}
  R_K - R_{K_0} = I(V;Y) - \tau.
\end{align*}

Finally, by standard joint typicality argument\cite{el2011network}, we have $Pr\{D_{4,b}\}\to 0$ as $n\to\infty$ if 
\begin{align*}
  \widetilde{R}_1 &\leq I(U_1;Y|V,U,U_2),\\
  \widetilde{R}_2 &\leq I(U_2;Y|V,U,U_1),\\
  R_0 &\leq I(U,U_1,U_2;Y|V),\\
  R_0 + \widetilde{R}_1 &\leq I(U,U_1,U_2;Y|V),\\
  R_0 + \widetilde{R}_2 &\leq I(U,U_1,U_2;Y|V),\\
  \widetilde{R}_1 + \widetilde{R}_2 &\leq I(U_1,U_2;Y|V,U),\\
  R_0 + \widetilde{R}_1 + \widetilde{R}_2 &\leq I(U,U_1,U_2;Y|V).
\end{align*}

\textbf{Security Analysis: }In the above coding scheme, each message $m_1$ is split into two independent parts $(m_{10},m_{11})$, where the first part $m_{10}$ is protected by using wiretap codes and the statistical property of the channel itself, and another part $m_{11}$ is protected by secret key agreement between the sender and the legitimate receiver. To show the proposed coding scheme can achieve strong secrecy, we first analyze the performance of the wiretap codes. Let $P^n_{YS|U\mathcal{C}\mathcal{C}_1\mathcal{C}_2}$ be distribution induced by given codebooks $\mathcal{C},\mathcal{C}_1$ and $\mathcal{C}_2$ such that $\mathcal{C}_1$ and $\mathcal{C}_2$ are generated according to distributions $P_{U_1|U}^n(\cdot|\boldsymbol{u})$ and $P_{U_2|U}^n(\cdot|\boldsymbol{u})$, respectively for $\boldsymbol{u}\in\mathcal{C}$ and
\begin{align*}
  P^n_{YS|U\mathcal{C}\mathcal{C}_1\mathcal{C}_2}(\boldsymbol{y},\boldsymbol{s}|\boldsymbol{u})&=\frac{1}{|\mathcal{C}_1||\mathcal{C}_2|}\sum_{i=1}^{|\mathcal{C}_1|}\sum_{j=1}^{|\mathcal{C}_2|}P^n_{YS|UU_1U_2}(\boldsymbol{y},\boldsymbol{s}|\boldsymbol{u},\boldsymbol{u}_1(i),\boldsymbol{u}_2(j))\\
  &=\frac{1}{|\mathcal{C}_1||\mathcal{C}_2|}\sum_{i=1}^{|\mathcal{C}_1|}\sum_{j=1}^{|\mathcal{C}_2|}P^n_{Y|UU_1U_2S}(\boldsymbol{y}|\boldsymbol{u},\boldsymbol{u}_1(i),\boldsymbol{u}_2(j),\boldsymbol{s})P_S^n(\boldsymbol{s}).
\end{align*}
Further define distribution $P^n_{YS|U}$ on $\mathcal{Y}^n\times\mathcal{S}^n$ such that
\begin{align*}
  P^n_{YS|U}(\boldsymbol{y},\boldsymbol{s}|\boldsymbol{u})=\sum_{\boldsymbol{u}_1\in\mathcal{U}^n_1}\sum_{\boldsymbol{u}_2\in\mathcal{U}^n_2}P^n_{U_1|U}(\boldsymbol{u}_1|\boldsymbol{u})P^n_{U_2|U}(\boldsymbol{u}_2|\boldsymbol{u})P^n_{YS|U,U_1,U_2}(\boldsymbol{y},\boldsymbol{s}|\boldsymbol{u},\boldsymbol{u}_{1},\boldsymbol{u}_{2}).
\end{align*}
\begin{lemma}\label{lem: wiretap codes}
  Let $(\textbf{C},\textbf{C}_1,\textbf{C}_2)$ be a pair of random codebooks and $P_{UU_1U_2}=P_{U}P_{U_1|U}P_{U_2|U}$ be the joint distribution such that
  \begin{align*}
    &Pr\{\textbf{C}=\mathcal{C}\} = \prod_{l=1}^{2^{nR_0}}\prod_{i=1}^n P_U(u_i(l)),\\
    &Pr\{\textbf{C}_1=\mathcal{C}_1(l),\textbf{C}_2=\mathcal{C}_2(l)|\textbf{C}=\mathcal{C}\}=\prod_{i=1}^2\prod_{l_i=1}^{2^{nR_i}}\prod_{j=1}^{n}P_{U_i|U}(u_{ij}(l_i)|\boldsymbol{u}(l)),
  \end{align*}
for some $\boldsymbol{u}(l)$ in $\mathcal{C}$, where $R_1,R_2$ are positive real numbers such that
\begin{align*}
  &R_1 > I(U_1;Z|S,U),\\
  &R_2 > I(U_2;Z|S,U),\\
  &R_1 + R_2 > I(U_1,U_2;Z|S,U).
\end{align*}
Let $(\boldsymbol{U}_1(l,l_1),\boldsymbol{U}_2(l,l_2))$ be randomly selected sequences and $\boldsymbol{Z}^n$ be the output sequence such that
\begin{align*}
  &Pr\{\boldsymbol{U}(l)=\boldsymbol{u}(l),\boldsymbol{U}_1(l,l_1)=\boldsymbol{u}_1(l,l_1),\boldsymbol{U}_2(l,l_2)=\boldsymbol{u}_2(l,l_2),\boldsymbol{Z}=\boldsymbol{z},\boldsymbol{S}=\boldsymbol{s}|\mathcal{C},\mathcal{C}_1,\mathcal{C}_2\}\\
  &=\frac{1}{|\mathcal{C}_1||\mathcal{C}_2|}\bar{P}(l)\prod_{i=1}^n P_{Z|UU_1U_2S}(z_i|u_i,u_{1i},u_{2i},s_i)P_S(s_i),
\end{align*}
where $\bar{P}$ is some distribution on $\mathcal{C}$.
Then, it follows that
\begin{align*}
  \mathbb{E}_{\textbf{C}\textbf{C}_1\textbf{C}_2}[D(P^n_{ZS|U\textbf{C}\textbf{C}_1\textbf{C}_2}||P^n_{ZS|U})] \to 0
\end{align*}
exponentially fast, where $D(P||Q)$ is KL divergence.
\end{lemma}
The lemma is proved by a similar technique of channel resolvability\cite{helal2020cooperative}. We give the proof in Appendix \ref{app: proof of wiretap codes} for completeness. Define the security index $\mathbb{S}$ of random variable $K$ against $Z$\cite{csiszar2011information} as $\mathbb{S}(K|Z)=\log|\mathcal{K}|-H(K|Z)$, where $\mathcal{K}$ is the range of $K$. The secret key is constructed by the following lemma.

\begin{lemma}\label{lem: secret key}
  Suppose $V-S-(YZ)$ forms a Markov chain.
  Let $\mathcal{C}_{\mathcal{V}}=\{\boldsymbol{v}(i,j)\}_{i\in[1:2^{nR_1}],j\in[1:2^{nR_2}]}$ be a codebook containing $2^{n\widetilde{R}}$ codewords, where $\widetilde{R}=I(V;S)+\tau, R_1=I(V;S)-I(V;Y)+2\tau,R_2=I(V;Y)-\tau$, $\tau>0$. There exists a mapping $\kappa:[1:2^{nR_2}]\to\{1,...,k\}$, where $k=I(V;Y)-I(V;Z)$ such that
  \begin{align*}
    \mathbb{S}(\kappa(J)|\boldsymbol{Z},I) \leq \epsilon
  \end{align*}
  for some $\epsilon>0$ that can be arbitrarily small. In fact, by setting the size of the codebook as above, we can further construct a partition on the codebook such that
  \begin{align*}
    \mathbb{S}(I|\boldsymbol{Z}) \leq \epsilon.
  \end{align*}
\end{lemma}
The proof is given in Appendix \ref{app: proof of secret key}. When applying this lemma to the coding scheme, $\boldsymbol{Z}$ is the output of the wiretap channel, $\mathcal{C}_{\mathcal{V}}$ is the key message codebook $\mathcal{C}_K$ in \emph{Key Message Codebook Generation},  $I$ and $J$ index the lossy description $\boldsymbol{V}_{b-1}(I,J)$ of the state sequence $\boldsymbol{S}_{b-1}$ from the last block.
\begin{remark}
  The first part of the lemma is the direct part of Theorem 17.21 in \cite{csiszar2011information}, where the sender and the legitimate receiver try to generate a key and keep it secret from the eavesdropper who can observe the channel output and an index $I$. The lemma gives the secret key with rate $R=I(V;Y)-I(V;Z)$. Setting the random variable $V=S$ in Lemma \ref{lem: secret key}, the randomness extracted is $H(S|Z)-H(S|Y)$. Compared to the randomness extraction in \cite[Lemma 17.5]{csiszar2011information}, the additional term $H(S|Y)$ is due to the assumption that the eavesdropper always observes index $I$.
\end{remark}
\begin{remark}\label{remark: codebook markov chain}
  Note that in block $b$, $\boldsymbol{Z}^{b-1}$ affects $\boldsymbol{Z}_b$ by $K_{0,b}$ and $K_{1,b}$ and inversely $\boldsymbol{Z}_b$ affects $\boldsymbol{Z}^{b-1}$ by $K_{0,b}$ and $C_{11,b}$. The key $K_{0,b}$ is used to determine the auxiliary sequence $\boldsymbol{U}(K_{0,b})$ and then generates the message codebook $\textbf{C}_1(K_{0,b}),\textbf{C}_2(K_{0,b})$. Hence, given codebook $\bar{\textbf{C}}^{B+1}$ used in $B+1$ blocks, $\boldsymbol{Z}^{b-1}$ affects $\boldsymbol{Z}_b$ only by $K_{1,b}$
\end{remark}
Note that no meaningful message is sent in Block 1 and Block B+1. Hence, we set $M_{1,1}=M_{2,1}=M_{1,B+1}=M_{2,B+1}=\text{Const}$. To bound the information leakage, let $M^{B+1}=[M_1,M_2,\dots,M_{B+1}],M^{[b+1]}=[M_{b+1},M_{b+2},\dots,M_{B+1}]$, $\bar{\textbf{C}}^{B+1}$ be the codebooks used in $B+1$ blocks. It follows that
\begin{align*}
  &I(M_1^{B+1},M_2^{B+1};\boldsymbol{Z}^{B+1}|\bar{\textbf{C}}^{B+1})\\
  &\leq I(M_1^{B+1},M_2^{B+1};\boldsymbol{Z}^{B+1},\boldsymbol{U}^{B+1}|\bar{\textbf{C}}^{B+1})\\
  &=\sum_{b} I(M_{1,b},M_{2,b};\boldsymbol{Z}^{B+1},\boldsymbol{U}^{B+1}|M_1^{[b+1]},M_2^{[b+1]},\bar{\textbf{C}}^{B+1})\\
  &\overset{(a)}{\leq} \sum_{b} I(M_{1,b},M_{2,b};\boldsymbol{Z}^{B+1},\boldsymbol{U}^{B+1}|\boldsymbol{S}_b,M_1^{[b+1]},M_2^{[b+1]},\bar{\textbf{C}}^{B+1})\\
  &\overset{(b)}{=}\sum_{b} I(M_{1,b},M_{2,b};\boldsymbol{Z}^{b},\boldsymbol{U}^{b}|\boldsymbol{S}_b,\bar{\textbf{C}}^{B+1}),
\end{align*}
where $(a)$ follows by the independence between $(M_{1,b},M_{2,b})$ and $\boldsymbol{S}_b$, $(b)$ follows by the independence between $(M_{1,b},M_{2,b})$ and $(M_1^{[b+1]},M_2^{[b+1]},\boldsymbol{Z}^{[b+1]},\boldsymbol{U}^{[b+1]})$ given $(\boldsymbol{Z}^b,\boldsymbol{U}^b,\boldsymbol{S}_b)$.

To bound $I(M_{1,b},M_{2,b};\boldsymbol{Z}^{b},\boldsymbol{U}^{b}|\boldsymbol{S}_{b},\bar{\textbf{C}}^{B+1})$, it follows that
\begin{align*}
  &I(M_{1,b},M_{2,b};\boldsymbol{Z}^{b},\boldsymbol{U}^{b}|\boldsymbol{S}_{b},\bar{\textbf{C}}^{B+1})\\
  &=I(M_{1,b},M_{2,b};\boldsymbol{Z}^{b-1},\boldsymbol{U}^{b-1}|\boldsymbol{S}_{b},\bar{\textbf{C}}^{B+1}) + I(M_{1,b},M_{2,b};\boldsymbol{Z}_{b},\boldsymbol{U}_{b}|\boldsymbol{Z}^{b-1},\boldsymbol{U}^{b-1},\boldsymbol{S}_{b},\bar{\textbf{C}}^{B+1})\\
  &=I(M_{1,b},M_{2,b};\boldsymbol{Z}_{b},\boldsymbol{U}_{b}|\boldsymbol{Z}^{b-1},\boldsymbol{U}^{b-1},\boldsymbol{S}_{b},\bar{\textbf{C}}^{B+1})\\
  &=\underbrace{I(M_{10,b},M_{20,b};\boldsymbol{Z}_{b},\boldsymbol{U}_{b}|\boldsymbol{Z}^{b-1},\boldsymbol{U}^{b-1},\boldsymbol{S}_{b},\bar{\textbf{C}}^{B+1})}_{=:I_1}+\underbrace{I(M_{11,b},M_{21,b};\boldsymbol{Z}_{b},\boldsymbol{U}_{b}|\boldsymbol{Z}^{b-1},\boldsymbol{U}^{b-1},\boldsymbol{S}_{b},M_{10,b},M_{20,b},\bar{\textbf{C}}^{B+1})}_{=:I_2}.
\end{align*}
The above $I_1$ can be bounded by 
\begin{align*}
  I_1 &\leq I(M_{10,b},M_{20,b},\boldsymbol{Z}^{b-1},\boldsymbol{U}^{b-1};\boldsymbol{Z}_{b},\boldsymbol{U}_{b}|\boldsymbol{S}_{b},\bar{\textbf{C}}^{B+1})\\
  &= I(M_{10,b},M_{20,b};\boldsymbol{Z}_{b},\boldsymbol{U}_{b}|\boldsymbol{S}_{b},\bar{\textbf{C}}^{B+1}) + I(\boldsymbol{Z}^{b-1},\boldsymbol{U}^{b-1};\boldsymbol{Z}_{b},\boldsymbol{U}_{b}|\boldsymbol{S}_{b},M_{10,b},M_{20,b},\bar{\textbf{C}}^{B+1})\\
  &\leq I(M_{10,b},M_{20,b};\boldsymbol{Z}_{b},\boldsymbol{U}_{b}|\boldsymbol{S}_{b},\bar{\textbf{C}}^{B+1}) + I(\boldsymbol{Z}^{b-1},\boldsymbol{U}^{b-1};\boldsymbol{Z}_{b},\boldsymbol{U}_{b},\boldsymbol{S}_{b},M_{10,b},M_{20,b},K_{1,b}|\bar{\textbf{C}}^{B+1})\\
  &\overset{(a)}{=}\underbrace{I(M_{10,b},M_{20,b};\boldsymbol{Z}_{b},\boldsymbol{U}_{b}|\boldsymbol{S}_{b},\bar{\textbf{C}}^{B+1})}_{=:I_{11}} + \underbrace{I(\boldsymbol{Z}^{b-1},\boldsymbol{U}^{b-1};K_{1,b},K_{0,b}|\bar{\textbf{C}}^{B+1})}_{=:I_{12}},
\end{align*}
where $(a)$ follows by the Markov chain $(\boldsymbol{Z}^{b-1},\boldsymbol{U}^{b-1}) - (K_{0,b},K_{1,b},\bar{\textbf{C}}^{B+1}) - (\boldsymbol{Z}_{b},\boldsymbol{S}_{b},M_{10,b},M_{20,b})$.
Now by Lemma \ref{lem: wiretap codes},
\begin{align}
  I_{11}&=I(M_{10,b},M_{20,b};\boldsymbol{Z}_{b},\boldsymbol{U}_{b}|\boldsymbol{S}_{b},\bar{\textbf{C}}^{B+1})\\
  &\overset{(a)}{=}I(M_{10,b},M_{20,b};\boldsymbol{Z}_{b},\boldsymbol{S}_{b}|\boldsymbol{U}_{b},\bar{\textbf{C}}^{B+1})\notag\\
  &=\sum_{m_{10,b}}\sum_{m_{20,b}}Pr\{M_{10,b}=m_{10,b},M_{20,b}=m_{20,b}\} \mathbb{E}_{\textbf{C}\textbf{C}_1\textbf{C}_2}\left[D(P_{\boldsymbol{Z}\boldsymbol{S}|\boldsymbol{U}\textbf{C}\textbf{C}_1\textbf{C}_2m_{10}m_{20}}||P_{\boldsymbol{Z}\boldsymbol{S}|\boldsymbol{U}\textbf{C}\textbf{C}_1\textbf{C}_2})\right]\notag\\
  &\overset{(b)}{\leq}\sum_{m_{10,b}}\sum_{m_{20,b}}Pr\{M_{10,b}=m_{10,b},M_{20,b}=m_{20,b}\} \mathbb{E}_{\textbf{C}\textbf{C}_1\textbf{C}_2}\left[D(P_{\boldsymbol{Z}\boldsymbol{S}|\boldsymbol{U}\textbf{C}\textbf{C}_1\textbf{C}_2m_{10}m_{20}}||P_{\boldsymbol{Z}\boldsymbol{S}|\boldsymbol{U}})\right]\notag\\
  \label{neq: I_11}&=\mathbb{E}_{M_{10}M_{20}}\mathbb{E}_{\textbf{C}\textbf{C}_1\textbf{C}_2}\left[D(P_{\boldsymbol{Z}\boldsymbol{S}|\boldsymbol{U}\textbf{C}\textbf{C}_1\textbf{C}_2M_{10}M_{20}}||P_{\boldsymbol{Z}\boldsymbol{S}|\boldsymbol{U}})\right]\overset{(c)}{\to} 0
\end{align}
where $(a)$ follows by the independence between $(M_{10,b},M_{20,b})$ and $\boldsymbol{S}_{b}$ given $\bar{\textbf{C}}^{B+1}$, $(b)$ follows by adding $D(P_{\boldsymbol{Z}\boldsymbol{S}|\boldsymbol{U}\textbf{C}\textbf{C}_1\textbf{C}_2}||P_{\boldsymbol{Z}\boldsymbol{S}|\boldsymbol{U}})\geq 0$ to the equation above it, 
$(c)$ is follows from Lemma \ref{lem: wiretap codes}. 

Applying Lemma \ref{lem: secret key} to our coding scheme by replacing $\boldsymbol{Z}$ by $(\boldsymbol{Z},\boldsymbol{U})$ and giving the role of $\widetilde{R},R_1,R_2,k$ to $R_K,R_{K_0},R_K-R_{K_0},R_{K_1}$ respectively yields
\begin{align*} 
  I(K_{1,b};K_{0,b}|\bar{\textbf{C}}^{B+1})\leq I(K_{1,b};\boldsymbol{Z}_{b-1},\boldsymbol{U}_{b-1},K_{0,b}|\bar{\textbf{C}}^{B+1})\leq\mathbb{S}(K_{1,b}|\boldsymbol{Z}_{b-1},\boldsymbol{U}_{b-1},K_{0,b})\leq \epsilon,\\
  I(K_{0,b};\boldsymbol{Z}_{b-1},\boldsymbol{U}_{b-1}|\bar{\textbf{C}^{B+1}})\leq \mathbb{S}(K_{0,b}|\boldsymbol{Z}_{b-1},\boldsymbol{U}_{b-1})\leq \epsilon
\end{align*}
where $\mathbb{S}$ is the security index and hence,
$I(K_{0,b},K_{1,b};\boldsymbol{Z}_{b-1},\boldsymbol{U}_{b-1})\leq 2\epsilon$
for some $\epsilon>0$. Now by the same recursion argument as in \cite{sasaki2019wiretap}, it follows that
\begin{align}
  &I(K_{0,b},K_{1,b};\boldsymbol{Z}^{b-1},\boldsymbol{U}^{b-1}|\bar{\textbf{C}}^{B+1})\notag\\
  &=I(K_{0,b},K_{1,b};\boldsymbol{Z}_{b-1},\boldsymbol{U}_{b-1}|\bar{\textbf{C}}^{B+1}) + I(K_{0,b},K_{1,b};\boldsymbol{Z}^{b-2},\boldsymbol{U}^{b-2}|\bar{\textbf{C}}^{B+1},\boldsymbol{Z}_{b-1},\boldsymbol{U}_{b-1})\notag\\
  &\leq I(K_{0,b},K_{1,b};\boldsymbol{Z}_{b-1},\boldsymbol{U}_{b-1}|\bar{\textbf{C}}^{B+1}) + I(K_{0,b},K_{1,b},K_{0,b-1},K_{1,b-1};\boldsymbol{Z}^{b-2},\boldsymbol{U}^{b-2}|\bar{\textbf{C}}^{B+1},\boldsymbol{Z}_{b-1},\boldsymbol{U}_{b-1})\notag\\
  &\overset{(a)}{=}I(K_{0,b},K_{1,b};\boldsymbol{Z}_{b-1},\boldsymbol{U}_{b-1}|\bar{\textbf{C}}^{B+1}) + I(K_{0,b-1},K_{1,b-1};\boldsymbol{Z}^{b-2},\boldsymbol{U}^{b-2}|\bar{\textbf{C}}^{B+1},\boldsymbol{Z}_{b-1},\boldsymbol{U}_{b-1})\notag\\
  &\overset{(b)}{\leq} I(K_{0,b},K_{1,b};\boldsymbol{Z}_{b-1},\boldsymbol{U}_{b-1}|\bar{\textbf{C}}^{B+1}) + I(K_{0,b-1},K_{1,b-1};\boldsymbol{Z}^{b-2},\boldsymbol{U}^{b-2}|\bar{\textbf{C}}^{B+1})\notag.
\end{align}
where $(a)$ follows by the independence between $(K_{0,b},K_{1,b})$ and $(\boldsymbol{Z}^{b-2},\boldsymbol{U}^{b-2})$ given $(\bar{\textbf{C}}^{B+1},\boldsymbol{Z}_{b-1},\boldsymbol{U}_{b-1},K_{0,b-1},K_{1,b-1})$, $(b)$ follows by the Markov chain $(\boldsymbol{Z}^{b-2},\boldsymbol{U}^{b-2})\to (\bar{\textbf{C}}^{B+1},K_{0,b-1},K_{1,b-1})\to(\boldsymbol{Z}_{b-1},\boldsymbol{U}_{b-1})$. 
Hence,
\begin{align}
  I(K_{0,b},K_{1,b};\boldsymbol{Z}^{b-1},\boldsymbol{U}^{b-1}|\bar{\textbf{C}}^{B+1}) &\leq \sum_{b=1}^{B+1} I(K_{0,b},K_{1,b};\boldsymbol{Z}_{b-1},\boldsymbol{U}_{b-1}|\bar{\textbf{C}}^{B+1})\notag\\
  \label{neq: I_12}&\leq 2(B+1)\epsilon.
\end{align}

Combining \eqref{neq: I_11} and \eqref{neq: I_12} gives
\begin{align*}
  I_{1}\leq (2B+3)\epsilon.
\end{align*}
To bound $I_{2}$, it follows that
\begin{align*}
  I_{2}&=I(M_{11,b},M_{21,b};\boldsymbol{Z}_{b},\boldsymbol{U}_b|\boldsymbol{Z}^{b-1},\boldsymbol{U}^{b-1},\boldsymbol{S}_{b},M_{10,b},M_{20,b},\bar{\textbf{C}}^{B+1})\\
  &\leq I(M_{11,b},M_{21,b},\boldsymbol{Z}^{b-1},\boldsymbol{U}^{b-1};\boldsymbol{Z}_{b},\boldsymbol{U}_b|\boldsymbol{S}_{b},M_{10,b},M_{20,b},\bar{\textbf{C}}^{B+1})\\
  &=I(M_{11,b},M_{21,b};\boldsymbol{Z}_{b},\boldsymbol{U}_b|\boldsymbol{S}_{b},M_{10,b},M_{20,b},\bar{\textbf{C}}^{B+1}) + I(\boldsymbol{Z}^{b-1},\boldsymbol{U}^{b-1};\boldsymbol{Z}_{b},\boldsymbol{U}_b|\boldsymbol{S}_{b},M_{10,b},M_{20,b},\bar{\textbf{C}}^{B+1},M_{11,b},M_{21,b})\\
  &\leq I(M_{11,b},M_{21,b};\boldsymbol{Z}_{b},\boldsymbol{U}_b,C_{11,b},C_{21,b}|\boldsymbol{S}_{b},M_{10,b},M_{20,b},\bar{\textbf{C}}^{B+1}) + I(\boldsymbol{Z}^{b-1},\boldsymbol{U}^{b-1};\boldsymbol{Z}_{b},\boldsymbol{U}_b|\boldsymbol{S}_{b},M_{10,b},M_{20,b},\bar{\textbf{C}}^{B+1},M_{11,b},M_{21,b})\\
  &=\underbrace{I(M_{11,b},M_{21,b};C_{11,b},C_{21,b}|\boldsymbol{S}_{b},M_{10,b},M_{20,b},\bar{\textbf{C}}^{B+1})}_{=:I_{21}} + \underbrace{I(M_{11,b},M_{21,b};\boldsymbol{Z}_{b},\boldsymbol{U}_b|\boldsymbol{S}_{b},M_{10,b},M_{20,b},\bar{\textbf{C}}^{B+1},C_{11,b},C_{21,b})}_{=:I_{22}}\\
  &\quad\quad\quad\quad\quad +\underbrace{I(\boldsymbol{Z}^{b-1},\boldsymbol{U}^{b-1};\boldsymbol{Z}_{b},\boldsymbol{U}_b|\boldsymbol{S}_{b},M_{10,b},M_{20,b},\bar{\textbf{C}}^{B+1},M_{11,b},M_{21,b})}_{=:I_{23}}.
\end{align*}
Then, it follows that
\begin{align*}
  I_{21}&= I(M_{11,b},M_{21,b};C_{11,b},C_{21,b}|\boldsymbol{S}_{b},M_{10,b},M_{20,b},\bar{\textbf{C}}^{B+1})\\
  &=I(M_{11,b},M_{21,b};C_{11,b},C_{21,b}|\bar{\textbf{C}}^{B+1})\\
  &= H(C_{11,b},C_{21,b}|\bar{\textbf{C}}^{B+1}) - H(C_{11,b},C_{21,b}|\bar{\textbf{C}}^{B+1},M_{11,b},M_{21,b})\\
  &\overset{(a)}{\leq} H(C_{11,b},C_{21,b}) - H(K_{11,b},K_{21,b}|\bar{\textbf{C}}^{B+1})\\
  &\overset{(b)}{\leq} H(C_{11,b},C_{21,b}) - H(K_{11,b},K_{21,b}|K_{0,b})\\
  &\overset{(c)}{=}D(P_{K_1|K_0} || P_C) \leq \mathbb{S}(K_{1,b}|\boldsymbol{Z},K_{0,b})\leq \epsilon
\end{align*}
where $(a)$ follows from the fact that $C_{i1,b}=M_{i1,b}\oplus K_{i1,b}$,and $M_{i1,b}$ is independent of $K_{i1,b},i=1,2$, $(b)$ follows by the Markov chain $K_{1,b}-K_{0,b}-\bar{\textbf{C}}^{B+1}$ and $K_{1,b}=(K_{11,b},K_{21,b})$, $(c)$ follows by the fact that $(C_{11,b},C_{21,b})$ is uniformly distributed on $[1:2^{nR_{11}}]\times[1:2^{nR_{21}}]$ and $P_C$ is a uniform distribution on $[1:2^{nR_{11}}]\times[1:2^{nR_{21}}]$.

Then, we bound $I_{22}$ as follows.
\begin{align*}
  I_{22}&=I(M_{11,b},M_{21,b};\boldsymbol{Z}_{b},\boldsymbol{U}_{b}|\boldsymbol{S}_{b},M_{10,b},M_{20,b},\bar{\textbf{C}}^{B+1},C_{11,b},C_{21,b})\\
  &=I(K_{11,b},K_{21,b};\boldsymbol{Z}_{b},\boldsymbol{U}_{b}|\boldsymbol{S}_{b},M_{10,b},M_{20,b},\bar{\textbf{C}}^{B+1},C_{11,b},C_{21,b})\\
  &\leq I(K_{11,b},K_{21,b};\boldsymbol{Z}_{b},\boldsymbol{U}_{b},K_{0,b}|\boldsymbol{S}_{b},M_{10,b},M_{20,b},\bar{\textbf{C}}^{B+1},C_{11,b},C_{21,b})\\
  &\overset{(a)}{=}I(K_{11,b},K_{21,b};K_{0,b}|\boldsymbol{S}_{b},M_{10,b},M_{20,b},\bar{\textbf{C}}^{B+1},C_{11,b},C_{21,b})\\
  &=I(K_{1,b};K_{0,b}|\bar{\textbf{C}}^{B+1})\leq \epsilon,
\end{align*}
where $(a)$ follows by the Markov chain $K_{1,b}-K_{0,b}-(\boldsymbol{Z}_b,\boldsymbol{U}_b)$ given $\boldsymbol{S}_{b},M_{10,b},M_{20,b},\bar{\textbf{C}}^{B+1},C_{11,b}$.

The last term $I_{23}$ is bounded by 
\begin{align*}
  I_{23} &= I(\boldsymbol{Z}^{b-1},\boldsymbol{U}^{b-1};\boldsymbol{Z}_{b},\boldsymbol{U}_{b}|\boldsymbol{S}_{b},M_{10,b},M_{20,b},\bar{\textbf{C}}^{B+1},M_{11,b},M_{21,b})\\
  &\leq I(\boldsymbol{Z}^{b-1},\boldsymbol{U}^{b-1};K_{0,b},K_{1,b},\boldsymbol{Z}_{b},\boldsymbol{U}_{b}|\boldsymbol{S}_{b},M_{10,b},M_{20,b},\bar{\textbf{C}}^{B+1},M_{11,b},M_{21,b})\\
  &\overset{(a)}{=}I(\boldsymbol{Z}^{b-1},\boldsymbol{U}^{b-1};K_{0,b},K_{1,b}|\boldsymbol{S}_{b},M_{10,b},M_{20,b},\bar{\textbf{C}}^{B+1},M_{11,b},M_{21,b})\\
  &=I(\boldsymbol{Z}^{b-1},\boldsymbol{U}^{b-1};K_{0,b},K_{1,b}|\bar{\textbf{C}}^{B+1}) \overset{(b)}{\leq} 2(B+1)\epsilon,
\end{align*}
where $(a)$ follows by the independence between $(\boldsymbol{Z}^{b-1},\boldsymbol{U}^{b-1})$ and $(\boldsymbol{Z}_{b},\boldsymbol{U}_{b})$ given $\boldsymbol{S}_{b},M_{10,b},M_{20,b},\bar{\textbf{C}}^{B+1},M_{11,b},K_{0,b},K_{1,b}$, $(b)$ follows by formula \eqref{neq: I_12}. Now, the information leakage is bounded by 
\begin{align*}
  I(M_1^{B+1},M_2^{B+1};\boldsymbol{Z}^{B+1},\bar{\textbf{C}}^{B+1})\leq(B+1)*(I_1+I_2)\leq (B+1)(4B+7)\epsilon.
\end{align*}
Using Fourier–Motzkin Elimination yields the achievable region $\mathcal{R}_{11}$.
 
By symmetry, it is sufficient to discuss regions $\mathcal{R}_{12}$ because region $\mathcal{R}_{13}$ follows similarly. The basic idea of region $\mathcal{R}_{12}$ is the same as the region $\mathcal{R}_2$ in \cite{molavianjazi2009secure}, where only one confidential message is protected by the wiretap code. The coding scheme is almost the same as that for $\mathcal{R}_{11}$ except that we set $R_{10}=0$. In this case, Lemma \ref{lem: wiretap codes} is not sufficient to protect the message of Sender 2 since the eavesdropper may have the information of the sequence transmitted by Sender 1.

\begin{lemma}\label{lem: wiretap codes 2}
  Let $(\textbf{C},\textbf{C}_1,\textbf{C}_2)$ be random codebooks generated as in Lemma \ref{lem: wiretap codes}, and $(\boldsymbol{U},\boldsymbol{U}_1,\boldsymbol{U}_2,\boldsymbol{S},\boldsymbol{Y})$ be random variables with joint distribution as in Lemma \ref{lem: wiretap codes}. If 
  \begin{align*}
    R_2 > I(U_2;Y|U,U_1,S),
  \end{align*}
  it follows that
  \begin{align*}
    \mathbb{E}_{\textbf{C}\textbf{C}_1\textbf{C}_2}[D(P^n_{YS|UU_1\textbf{C}\textbf{C}_1\textbf{C}_2}||P^n_{YS|UU_1})] \to 0.
  \end{align*}
\end{lemma}
The proof is similar to that of Lemma \ref{lem: wiretap codes} and is omitted here.
Applying Fourier–Motzkin Elimination and replacing $(Z,U)$ with $(Z,U,U_1)$ results in the desired region $\mathcal{R}_{12}$. Note that in $\mathcal{R}_{12}$, the gain of using the secret key changes to $I(V;Y)-I(V;Z,U,U_1)$. The additional $U_1$ in the mutual information term is because we do not use wiretap coding to randomize the input sequence $\boldsymbol{U}_1$, and hence the eavesdropper may have the information of the sequence. 

\section{coding scheme for $\mathcal{R}_{2}$}\label{codin scheme for R21}
In this section, we give the coding scheme for region $\mathcal{R}_{2}$. The coding scheme considered here is a block Markov coding scheme with forward decoding. We omit the error and security analysis since it is a combination of standard joint typical argument\cite{el2011network} and the secret key construction argument in \cite{sasaki2019wiretap}. Given $P_S$ and $P_{YZ|X_1X_2S}$, we consider real numbers $\widetilde{R}_1,\widetilde{R}_2,R_{10},R_{20},R_{11},R_{21},R_{12},R_{22}$ such that
\begin{equation}\label{ineq: constrints for coding scheme 2}
  \begin{split}
    &\widetilde{R}_1 \leq I(U_1;Y|U_2),\\
  &\widetilde{R}_2 \leq I(U_2;Y|U_1),\\
  &\widetilde{R}_1 + \widetilde{R}_2 \leq I(U_1,U_2;Y),\\
  &\widetilde{R}_1 - R_{10} > I(U_1;Z|S,U_2),\\
  &\widetilde{R}_2 - R_{20} > I(U_2;Z|S,U_1),\\
  &\widetilde{R}_1 + \widetilde{R}_2- R_{10}- R_{20}> I(U_1,U_2;Z|S),\\
  &R_{11} + R_{21} \geq H(S|Y,U_1,U_2),\\
  &R_{11} + R_{21} + R_{21} + R_{22} \leq H(S|Z).
  \end{split}
\end{equation}
under joint distribution $P_SP_{U_1}P_{U_2}P_{X_1|U_1S}P_{X_2|U_2S}P_{YZ|X_1X_2S}$.

\emph{Key Message Codebook Generation: } The secret key is generated in the same way as in the construction in \cite{sasaki2019wiretap} by using Slepian-Wolf Theorem\cite{slepian1973noiseless} and Csisz\'ar-K\"orner's key construction lemma \cite[Lemma 17.5]{csiszar2011information}. The state sequence $\boldsymbol{s}_{b-1}$ from block $b-1$ generates a pair of key $(k_{2,b},k_{1,b})\in[1:2^{n(R_{12}+R_{22})}]\times[1:2^{n(R_{11}+R_{21})}]$ by a pair of mapping $\sigma\times\kappa$ with $R_{11},R_{21}, R_{12},R_{22}$ defined in \eqref{ineq: constrints for coding scheme 2} such that the security index
\begin{align*}
  &\mathbb{S}(\sigma(\boldsymbol{S}_{b-1})\kappa(\boldsymbol{S}_{b-1})|\boldsymbol{Z}_{b-1})\leq \epsilon
\end{align*} 
for some arbitrarily small $\epsilon>0$.

\emph{Message Codebook Generation:} For each block $b$, let $\mathcal{C}_{1b}=\{\boldsymbol{u}_1(l)\}_{l=1}^{2^{n\widetilde{R}_1}}$ be a codebook consisting of $2^{n\widetilde{R}_1}$ codewords, each i.i.d. generated according to distribution $P_{U_1}$, where $\widetilde{R}_1>0$ is defined in \eqref{ineq: constrints for coding scheme 2}. 
For each codebook $\mathcal{C}_{1b}$, partition it into $2^{nR_{10}}$ subcodebooks $\mathcal{C}_{1b}(m_{10})$, where $m_{10}\in[1:2^{nR_{10}}]$. For each subcodebooks $\mathcal{C}_{1b}(m_{10})$, partition it into $2^{nR_{11}}$ two-layer subcodebooks $\mathcal{C}_{1b}(m_{10},m_{11})$, where $m_1\in[1:2^{nR_{11}}]$. For each subcodebooks $\mathcal{C}_{1b}(m_{10},m_{11})$, partition it into $2^{nR_{12}}$ three-layer subcodebooks $\mathcal{C}_{1b}(m_{10},m_{11},m_{12})=\{\boldsymbol{u}_1(m_{10},m_{11},m_{12},l_1)\}_{l_1=1}^{2^{nR_1'}}$, where $m_{12}\in[1:2^{nR_{12}}],R_1':= \widetilde{R}_1 - R_{10}-R_{11}-R_{12}$. 

Likewise, for codebook $\mathcal{C}_{2b}=\{\boldsymbol{u}_2(l)\}_{l=1}^{2^{n\widetilde{R}_2}}$ with codewords i.i.d. generated according to $P_{U_2}$, partition it into three-layer subcodebook $\mathcal{C}_{2b}(m_{20},m_{21},m_{22})$, where $m_{20}\in[1:2^{nR_{20}}],m_{21}\in[1:2^{nR_{21}}],m_{22}\in[1:2^{nR_{22}}]$.

The above codebooks are all generated randomly and independently. Denote the set of random codebooks in each block $b$ by $\bar{\textbf{C}}_b$.

\emph{Encoding: }  In the first block, setting $m_{10,1}=m_{20,1}=m_{11,1}=m_{21,1}=m_{12,1}=m_{22,1}=1$, Encoder 1 picks an index $l_1\in[1:2^{nR_1'}]$ uniformly at random and Encoder 2  picks an index $l_2\in[1:2^{nR_2'}]$ uniformly at random. The codeword $\boldsymbol{x}_1$ is generated according to $\prod_{i=1}^n P_{X_1|U_1S}(x_{1i}|u_{1i}(1,1,1,l_1),s_i)$ and  $\boldsymbol{x}_2$ is generated according to $\prod_{i=1}^n P_{X_2|U_2S}(x_{2i}|u_{2i}(1,1,1,l_2),s_i)$.

For block $2\leq b\leq B$, upon observing the state sequence $\boldsymbol{s}_{b-1}$ in the last block, the encoders generate the secret key $k_{1,b}=(k_{11,b},k_{21,b})$ and Slepian-Wolf index $k_{2,b}=(k_{12,b},k_{22,b})$ by mappings $\kappa$ and $\sigma$. To transmit message $m_{1,b}$, Encoder $1$ splits it into two independent parts $(m_{10,b},m_{11,b})$ and computes $c_{11,b} = m_{11,b} \oplus k_{11,b} \pmod{2^{nR_{11}}}$. Encoder 1 picks an index $l_1\in[1:2^{nR_1'}]$ uniformly at random and generates the codeword $\boldsymbol{x}_1$ by
\begin{align*}
  P^n_{X_1|U_1S}(\boldsymbol{x}_1|\boldsymbol{u}_1(m_{10,b},c_{11,b},k_{12,b},l_{1,b}),\boldsymbol{s}_b)=\prod_{i=1}^n P_{X_1|U_1S}(x_{1i}|u_{1i}(m_{10,b},c_{11,b},k_{12,b},l_{1,b}),s_{i,b}).
\end{align*}
Similarly, the codeword $\boldsymbol{x}_2$ for Sender 2 is generated by 
\begin{align*}
  P^n_{X_2|U_2S}(\boldsymbol{x}_2|\boldsymbol{u}_2(m_{20,b},c_{21,b},k_{22,b},l_{2,b}),\boldsymbol{s}_b)=\prod_{i=1}^n P_{X_2|U_2S}(x_{2i}|u_{2i}(m_{20,b},c_{21,b},k_{22,b},l_{2,b}),s_{i,b}).
\end{align*}

\emph{Decoding: } 
In Block $b, 1\leq b \leq B$, the decoder looks for a unique tuple $(\hat{m}_{10,b},\hat{m}_{20,b},\hat{c}_{11,b},\hat{c}_{21,b},\hat{k}_{12,b},\hat{k}_{22,b},\hat{l}_{1,b},\hat{l}_{2,b})$ such that 
\begin{align*}
  (\boldsymbol{u}_1(\hat{m}_{10,b},\hat{c}_{11,b},\hat{k}_{12,b},\hat{l}_{1,b}),\boldsymbol{u}_2(\hat{m}_{20,b},\hat{c}_{21,b},\hat{k}_{22,b},\hat{l}_{2,b}),\boldsymbol{y}_b)\in T^n_{P_{U_1U_2Y},\delta}
\end{align*}
for some $\delta>0$. In Block $b, 2\leq b \leq B$, the decoder then estimates $\hat{\boldsymbol{s}}_{b-1}$ according to 
\begin{align*}
  (\hat{k}_{12,b},\hat{k}_{22,b},\boldsymbol{u}_1(\hat{m}_{10,b-1},\hat{c}_{11,b-1},\hat{k}_{12,b-1},\hat{l}_{1,b-1}),\boldsymbol{u}_2(\hat{m}_{20,b-1},\hat{c}_{21,b-1},\hat{k}_{22,b-1},\hat{l}_{2,b-1}),\boldsymbol{y}_{b-1})
\end{align*}
and finds the secret key $(\hat{k}_{11,b},\hat{k}_{21,b})$ by $\kappa(\hat{\boldsymbol{s}}_{b-1})$.  Now the decoder computes $\hat{m}_{11,b}=\hat{c}_{11,b} \ominus \hat{k}_{11,b} \pmod{2^{nR_{11}}}$ and $\hat{m}_{21,b}=\hat{c}_{21,b} \ominus \hat{k}_{21,b} \pmod{2^{nR_{21}}}$. 

By standard joint typical argument, the decoding error is arbitrarily small if 
\begin{align*}
  &\widetilde{R}_1 \leq I(U_1;Y|U_2),\\
  &\widetilde{R}_2 \leq I(U_2;Y|U_1),\\
  &\widetilde{R}_1 + \widetilde{R}_2 \leq I(U_1,U_2;Y).
\end{align*}
Now applying Lemma \ref{lem: wiretap codes} with $\textbf{C}=\emptyset$ and the secrecy analysis in \cite{sasaki2019wiretap} completes the proof. 
Similar to the proof of region $\mathcal{R}_{1}$, regions $\mathcal{R}_{22}$ and $\mathcal{R}_{23}$ are derived by setting $R_{10}=0$ and $R_{20}=0$ in the above coding scheme, respectively.
\section{Examples and applications}\label{sec: examples}
In this section, we investigate some special cases of the channel model considered in this paper and apply the previous capacity results to these channel models.
The first example shows that in some channel models, achievable rate points falling in region $\mathcal{R}_1$ may not fall in region $\mathcal{R}_2$ and vice versa, and the second numerical example with degraded message sets shows region $\mathcal{R}_1$ can be strictly larger than $\mathcal{R}_2$.

Then, some capacity-achieving cases are discussed. We first consider the case that one of the senders is removed and the channel reduces to a point-to-point wiretap channel with causal CSI, which is studied in \cite{sasaki2019wiretap}. In this case, Coding scheme 2 reduces to the coding scheme used in \cite{sasaki2019wiretap}, and was proved to be optimal when the wiretap channel is a degraded version of the main channel and the CSI is also revealed to the legitimate receiver. In Section \ref{sec: one sender case}, we prove that our Coding scheme  1 is also optimal in this case. In Section \ref{sec: degraded message sets}, SD-MAWCs with degraded message sets and three types of channel state information are considered: causal CSI at one sender, causal CSI at one sender and strictly causal CSI at another sender, channel independent of the states. We find that both coding schemes are possible to be optimal when the CSI is revealed to the decoder. In Example \ref{example: state reproduced}, a numerical example shows that both coding schemes can be optimal if the decoder can reproduce the CSI itself.

\begin{example}\label{example: R1 > R_2}
  Consider a SD-MAWC with causal CSI at encoders where $\mathcal{X}_1=\mathcal{X}_2=\{0,1\},\mathcal{Y}=\mathcal{S}=\{0,1\}$, the channel model is described by
  \begin{equation*}
    \begin{split}
      &Y=\left\{
        \begin{aligned}
        X_1, \quad\quad\quad\quad \text{if}\;S=0,\\
        X_2, \quad\quad\quad\quad \text{if}\;S=1.\\
      \end{aligned}
        \right.\\
    &Z=X_2.
    \end{split}
  \end{equation*}
  In this channel model, rate pair $(\min\{1-p,1-h(p)\},0)$ is achieved by Coding scheme  1 but cannot be achieved by Coding scheme 2. If the main channel is described by 
  \begin{align*}
    Y=X_1\oplus X_2 \oplus S,
  \end{align*}
  then rate pair $(0,h(p))$ is achieved by Coding scheme 2 but cannot be achieved by Coding scheme  1.
\end{example}
\begin{proof}
  Setting $p=Pr\{S=1\}$. We first consider region $\mathcal{R}_1$. Set variables as follows,
  \begin{align*}
    V=S, U_2 = \emptyset, X_1=U_1, X_2=U
  \end{align*}
with $(U,U_1)$ and $S$ being independent and $U \backsim Bernoulli(\frac{1}{2})$ and $U_1 \backsim Bernoulli(\frac{1}{2})$. It follows that $I(U_1;Z|S,U)=I(U_1;Z|S,U,U_2)=I(U_1,U_2;Z|S,U)=0$. We further have
\begin{align*}
  &I(U_1;Y|V,U,U_2)=I(U_1;Y|S,U)= H(Y|S,U) = Pr\{S=1\}H(Y|U,S=1) = 1-p,\\
  &I(V,U,U_1,U_2;Y)-I(V;S) = H(Y)- H(S) = 1-h(p),
\end{align*}
where $h(p)=-p\log{p}-(1-p)\log(1-p)$. Note that the mutual information term $I(V;Y)=I(S;Y)\geq 0$. Hence, a rate pair $(R_1,0)$ satisfying $R_1 = \min\{1-p,1-h(p)\}$ is achievable. Now consider region $\mathcal{R}_2$. Restricting $R_2=0$ and assigning all secret key rate to Sender 1, it follows that
\begin{align*}
  I(U_1;Z|U_2,S)= I(U_1;X_2|U_2,S)=0, H(S|Z)=H(S|Z,U_1)=H(S),
\end{align*}
and the maximum rate that can be achieved by Sender 1 will not be greater than $I(U_1;Y|U_2)$.
Hence, it is sufficient to prove $I(U_1;Y|U_2)<\min\{1-p,1-h(p)\}$, where $I(U_1;Y|U_2)$ is in fact the constraint on $R_1$ in the achievable rate region of MAC with causal CSI using Shannon strategy. To see this, note that $I(U_1;Y|U_2)$ is a convex function of conditional distribution $P_{Y|U_1U_2}$, where $P_{Y|U_1U_2}(y|u_1,u_2)=\sum_{x_1,x_2,s}P_{Y|X_1X_2S}(y|x_1,x_2,s)P_{X_1|U_1S}(x_1|u_1,s)$ $P_{X_2|U_2S}(x_2|u_2,s)P_S(s)$ is a linear function of $P_{X_1|U_1S}$ and $P_{X_2|U_2S}$. Hence, it is also a convex function of $P_{X_1|U_1S}$ and $P_{X_2|U_2S}$, where the maximum is achieved in extreme points. Thus there is no loss of generality to replace the distribution $P_{X_1|U_1S}$ and $P_{X_2|U_2S}$ with deterministic functions $x_1(u_1,s)$ and $x_2(u_2,s)$, respectively. It is shown by Example 4 in \cite{lapidoth2012multiple} that $I(U_1;Y|U_2)<\min\{1-p,1-h(p)\}$ holds when $p$ is sufficiently large.

For the second channel model, consider the case that $R_1=0$. For Coding scheme  1, it follows that 
\begin{align*}
  I(U_2;Z|U,S) = I(U_2;X_2|U,S) \overset{(a)}{=} I(U_2;X_2|U,S,U_1) 
\end{align*}
where $(a)$ follows by the fact that $(U_2,X_2)$ is independent of $U_1$ given $(U,S)$. We further have 
\begin{align}
  I(U_2;Y|V,U,U_1) &= H(U_2|V,U,U_1) - H(U_2|V,U,U_1,Y)\notag\\
  &\overset{(a)}{\leq} H(U_2|S,U,U_1) - H(U_2|V,U,U_1,Y,S)\notag\\
  &\overset{(b)}{=}H(U_2|S,U,U_1) - H(U_2|U,U_1,Y,S) \notag\\
  \label{neq: example ine}&= I(U_2;Y|S,U,U_1)\overset{(c)}{=}I(U_2;X_1\oplus X_2|S,U,U_1),
\end{align}
where $(a)$ follows by the fact that  $U_2$ is independent to $U_1$ given $U$ and conditions decrease entropy, $(b)$ follows from the fact that $V$ is independent to anything else given $S$, $(c)$ follows by substituting the channel model into the mutual information. Applying the chain rule of mutual information yields
\begin{align*}
  I(U_2;X_1+X_2,X_2|S,U,U_1) &= I(U_2;X_2|S,U,U_1) + I(U_2;X_1\oplus X_2|S,U,U_1,X_2)\\
  &=I(U_2;X_2|S,U,U_1) + I(U_2;X_1|S,U,U_1,X_2)\\
  &\overset{(a)}{=}I(U_2;X_2|S,U,U_1)\\
  &=I(U_2;X_1\oplus X_2|S,U,U_1) + I(U_2;X_2|S,U,U_1,X_1\oplus X_2)
\end{align*}
where $(a)$ holds since $U_2$ is independent to $X_1$ given $(S,U,U_1,X_2)$. By the nonnegativity of the mutual information, we have $I(U_2;Z|U,S) =I(U_2;X_2|S,U,U_1)\geq I(U_2;X_1+X_2|S,U,U_1)=I(U_2;Y|S,U,U_1)$. Hence, region $\mathcal{R}_1$ reduces to $\mathcal{R}_{13}$. Now assign all the secret key rate to Sender 2, the achievable rate of Sender 2 satisfies
\begin{align*}
  R_2 &\leq I(V;Y) - I(V;Z,U,U_2) \\
  &=I(V;Y)-I(V;X_2,U,U_1)\\
  &\overset{(a)}{\leq}I(V;Y)\overset{(b)}{\leq} I(S;Y) = H(S) - H(S|Y)\overset{(c)}{\leq} h(p),
\end{align*}
where $(a)$ follows by the nonnegativity of mutual information, $(b)$ follows by data process inequality, the equality in $(c)$ holds when $V=S$ and $S$ can be determined by $Y$, which means $X_1\oplus X_2$ take some fixed numbers. Now consider the case that $V=S$ and $X_1\oplus X_2$ are fixed. Without loss generality, assume $X_1\oplus X_2=0$. The second constraint on Sender 2 in Coding scheme 1 is 
\begin{align*}
  I(U_2;Y|S,U,U_1) = H(Y|S,U,U_1) - H(Y|S,U,U_1,U_2) \overset{(a)}{=}0,
\end{align*}
where $(a)$ holds since $Y$ is determined by $S$ when $X_1\oplus X_2=0$ is fixed. Hence, the rate of Sender 2 is $R_2=0$ in this case. We conclude that $R_2 < h(p)$ using Coding scheme  1 for this model.

For Coding scheme 2, setting $X_1=U_1,X_2=U_2, U_1 \backsim Bernoulli(\frac{1}{2})$ and $U_2 \backsim Bernoulli(\frac{1}{2})$, the first constraint on Sender 2 is 
\begin{align*}
  I(U_2;Y|U_1)=I(X_2;Y|X_1)=H(Y|X_1) - H(Y|X_1,X_2)= 1 - h(p),
\end{align*}
The secret key rate of Sender 2 in this case is 
\begin{align*}
  H(S|Z,U_2) - H(S|Y,U_1,U_2) &= H(S) - H(S|Y,X_1,X_2) = h(p).
\end{align*}
When the distribution of channel states satisfies $1-h(p) \geq h(p)$, we have $(0,h(p))$ can be achieved by Coding scheme 2 but cannot be achieved by Coding scheme  1.
\end{proof}
In the rest of this section, we use superscript $`CSI-ED'$ to represent the causal CSI at encoder and decoder sides, use `$OCSI$' for one side causal CSI  and `$NCSI$' for noncausal CSI, and the subscript $`D'$ represents the degraded message sets. 
\subsection{One Sender}\label{sec: one sender case}
In this subsection, we consider the case that there is only one sender in the communication system. In this case, the SD-MAC reduces to an SD-DMC, and we are going to investigate the performance of the two coding schemes discussed in Sections \ref{sec: coding scheme for R11} and \ref{codin scheme for R21}. Without loss of generality, assume Sender 2 is removed and set $U_2=\emptyset$.
Define
\begin{align*}
  &\mathcal{R}(P_{VUU_1X_1|S}) = \min\{I(U_1;Y|V,U)-I(U_1;Z|S,U)+H(V|Z,U)-H(V|Y),I(U_1;Y|V,U),\\
  &\quad\quad\quad\quad\quad\quad\quad\quad I(V,U,U_1;Y)-I(V;S)-I(U_1;Z|S,U)+H(V|Z,U)-H(V|Y)\},
\end{align*} 
with joint distribution such that $P_SP_{VUU_1X_1|S}=P_SP_{V|S}P_UP_{U_1|U}P_{X_1|UU_1S}$.
By removing Sender 2, region $\mathcal{R}_1$ reduces to 
\begin{align*}
  \mathcal{R}_1 = \max_{P_{VUU_1X_1|S}}\mathcal{R}(P_{VUU_1X_1|S}).
\end{align*}
and region $\mathcal{R}_2$ reduces to $R_{CSI-1}(P_U,P_{X|US})$ in \cite{sasaki2019wiretap}. Region $\mathcal{R}_1$ and $\mathcal{R}_2$ in this case do not include each other in general. However, if we reveal the side information to the legitimate receiver and further assume the wiretap channel is a degraded version of the main channel, both coding schemes can achieve the secrecy capacity. Recall that the secrecy capacity of wiretap channels with causal/noncausal state information at both encoder and decoder, with a degraded eavesdropper is given as follows.
\begin{corollary}[Corollary 1 in \cite{sasaki2019wiretap}]\label{coro: single user capacity}
  For a degraded wiretap channel with causal/noncausal state information at both the encoder and decoder, it holds that
  \begin{align*}
    C^{CSI-ED}=C^{NCSI-ED}=\max_{P_{X|S}}\min\{I(X;Y|S)-I(X;Z|S)+H(S|Z),I(X;Y|S)\}
  \end{align*}
\end{corollary}
For this channel model, Coding schemes  1 and 2 are both optimal.
The proof of Coding scheme 2 can be found in \cite{sasaki2019wiretap} and is omitted. Here we give the proof of Coding scheme  1.
\begin{proof}
\emph{Achievability of Coding scheme  1.} Replacing $Y$ in $\mathcal{R}(P_{VUU_1X_1|S})$ by $(Y,S)$ and setting $V=S$ yield
\begin{align*}
  &I(U_1;Y,S|V,U)-I(U_1;Z|S,U)+H(V|Z,U)-H(V|Y,S)\\
  &= I(U_1;Y,S|S,U)-I(U_1;Z|S,U)+H(S|Z,U)-H(S|Y,S)\\
  &= I(U_1;Y|S,U)-I(U_1;Z|S,U)+H(S|Z,U).
\end{align*}
and 
\begin{align*}
  &I(V,U,U_1;Y,S)-I(V;S)-I(U_1;Z|S,U)+H(V|Z,U)-H(V|Y,S)\\
  &=I(S,U,U_1;Y,S)-I(S;S)-I(U_1;Z|S,U)+H(S|Z,U)-H(S|Y,S)\\
  &=I(U,U_1;Y|S) + H(S) - H(S)-I(U_1;Z|S,U)+H(S|Z,U)\\
  &=I(U,U_1;Y|S)-I(U_1;Z|S,U)+H(S|Z,U).
\end{align*}
Note that $I(U,U_1;Y|S)\geq I(U_1;Y|S,U)$.
By functional representation lemma\cite{el2011network}, for any $(X,S)$ there exists a $U_1$ independent of $S$ and a deterministic function $f:\mathcal{U}_1\times\mathcal{S}\to\mathcal{X}$. It follows that
\begin{align*}
  I(U_1;Y|S) = I(U_1,X;Y|S)=I(X;Y|S)
\end{align*}
and $I(U_1;Z|S) = I(U_1,X;Z|S)=I(X;Z|S)$ by the fact that $X$ is a deterministic function of $(U,S)$ and the Markov chain $U-(X,S)-Y$. This completes the proof of achievability.

\emph{Converse of Coding scheme  1. } To prove the converse, set $U_i=(Y^{i-1},S^{i-1},Z_{i+1}^n),U_{1i}=(U_i,M,S_{i+1}^n)$. Let $\delta$ be a positive number. Applying Fano's inequality gives
\begin{align*}
  nR_1 &= H(M) \\
  &=I(M;\boldsymbol{Y},\boldsymbol{S}) - I(M;\boldsymbol{Z}) + \delta\\
  &= I(M,\boldsymbol{S};\boldsymbol{Y},\boldsymbol{S}) - I(\boldsymbol{S};\boldsymbol{Y},\boldsymbol{S}|M) - I(M,\boldsymbol{S};\boldsymbol{Z}) + I(\boldsymbol{S};\boldsymbol{Z}|M) + \delta\\
  &=I(M,\boldsymbol{S};\boldsymbol{Y},\boldsymbol{S})- I(M,\boldsymbol{S};\boldsymbol{Z}) + H(\boldsymbol{S}|\boldsymbol{Y},\boldsymbol{S},M) - H(\boldsymbol{S}|\boldsymbol{Z},M)+\delta\\
  &\leq \sum_{i=1}^n I(M,\boldsymbol{S};Y_i,S_i|Y^{i-1},S^{i-1}) - I(M,\boldsymbol{S};Z_i|Z_{i+1}^n)+\delta\\
&\overset{(a)}{=}\sum_{i=1}^n I(M,\boldsymbol{S};Y_i,S_i|Y^{i-1},S^{i-1},Z_{i+1}^n) - I(M,\boldsymbol{S};Z_i|Y^{i-1},S^{i-1},Z_{i+1}^n)+\delta\\
&\overset{(b)}{=}\sum_{i=1}^n I(U_{1i},S_i;Y_i,S_i;U_i) - I(U_{1i},S_i;Z_i|U_i) + \delta\\
&=\sum_{i=1}^n I(U_{1i};Y_i|U_i,S_i) - I(U_{1i};Z_i|U_i,S_i) + H(S_i|Z_i,U_i) + \delta\\
&\overset{(c)}{=}n(I(U_1;Y|U,S)-I(U_1;Z|U,S)-H(S|Z,U)+\delta),
\end{align*}
where $(a)$ follows by applying Csisz\'ar's sum identity twice, $(b)$ follows by the definition of $U_i$ and $U_{1i}$, $(c)$ follows by introducing a time sharing random variable $J$ and setting $U=(U_J,J),U_1=U_{1J},Y=Y_J,S=S_J,Z=Z_J$ with the fact that $\boldsymbol{S}$ is stationary and memoryless. Now following the argument of formulas (135)-(139) in \cite[Proof of Theorem 4]{sasaki2021wiretap}, we have 
\begin{align*}
  I(U_1;Y|U,S)-I(U_1;Z|U,S)-H(S|Z,U) \leq I(X;Y|S) - I(X;Z|S) + H(S|Z).
\end{align*}
The proof of another bound is the same as that for Coding scheme 2 and can be found in \cite[Theorem 2]{sasaki2019wiretap}.
\end{proof}
As mentioned before, Coding scheme  1 and Coding scheme 2 do not outperform each other in general. It is worth exploring in what cases these two coding schemes have the same performance. One of the main differences between the two coding schemes is the way we construct the secret key, where in Coding scheme 1 Wyner-Ziv coding is used and the secret key is constructed by a lossy description of the state sequence. It is well-known that by setting $V=S$, the Wyner-Ziv coding theorem reduces to corner points of the Slepian-Wolf coding theorem. In this case, the difference in achievable regions $\mathcal{R}_1$ and $\mathcal{R}_2$ stems from the decoding step. Region $\mathcal{R}_1$ is derived by backward decoding, where the decoder first looks for $\boldsymbol{S}$ based on the observed channel output $\boldsymbol{Y}$. Region $\mathcal{R}_2$ adopts forward decoding, and the decoder first finds the transmitted sequence $\boldsymbol{U}_1$ and then the state sequence with the help of $(\boldsymbol{Y},\boldsymbol{U}_1)$. This leads to a lower cost of Slepian-Wolf coding for the decoder using Coding scheme 2. On the other hand, when decoding the transmitted sequence, the decoder in Coding scheme  1 has the state sequence, which can help the decoder find the transmitted sequence $\boldsymbol{U}_1$. This makes the achievable rate of the main channel higher than Coding scheme 2. When state sequences are available at the decoder side, the decoders in both coding schemes do not need to perform the Slepian-Wolf decoding and hence the cost for searching $\boldsymbol{S}$ is zero, and the codewords are found with the observation of channel output $\boldsymbol{Y}$ and state sequence $\boldsymbol{S}$, which makes it possible for both coding schemes to have the same performance.

Now, let's go back to the multi-user setting. In the following subsection, we consider multiple access wiretap channels with degraded messages. 

\subsection{SD-MAWCs with Degraded Message Sets}\label{sec: degraded message sets}
In this subsection, we consider SD-MAWCs with degraded message sets, where Sender 2 only sends a common message and Sender 1 sends both common message and private message. The private message is required to be kept secret from the eavesdropper.
Given $P_S$ and $P_{YZ|X_1X_2S}$, define sets of non-negative real numbers $(R_0,R_1)$ such that
\begin{equation}\label{neq: Coding scheme  1 for mac degraded message sets}
  \mathcal{R}^{CSI-E}_{D,11} = \bigcup_{P_{VUU_1U_2X_1X_2|S}}\left\{
    \begin{aligned}
      &R_1 \leq \min\{I(U_1;Y|U,U_2,V) - I(U_1;Z|U_2,U,S)+R_{SK},I(U_1;Y|U,U_2,V)\},\\
      &R_0 + R_1 \leq \min\{I(V,U,U_1,U_2;Y)-I(V;S)- I(U_1;Z|U_2,U,S)+R_{SK},\\
      &\quad\quad\quad\quad\quad\quad\quad\quad I(V,U,U_1,U_2;Y)-I(V;S)\} ,
    \end{aligned}
  \right.
\end{equation}
with joint distribution such that $P_SP_{VUU_1U_2X_1X_2|S}=P_{VS}P_{U}P_{U_1|U}P_{U_2|U}P_{X_1|UU_1S}P_{X_2|UU_2S}$, where $R_{SK}= I(V;Y)-I(V;Z,U,U_2)$ and
\begin{equation*}
  \mathcal{R}^{CSI-E}_{D,12} = \bigcup_{P_{VUU_1U_2X_1X_2|S}}\left\{
    \begin{aligned}
      &R_1 \leq \min\{R_{SK},I(U_1;Y|U,U_2,V)\},\\
      &R_0 + R_1 \leq I(V,U,U_1,U_2;Y)-I(V;S).
    \end{aligned}
  \right.
\end{equation*}

Region $\mathcal{R}_2$ in this case becomes
\begin{equation*}
  \mathcal{R}^{CSI-E}_{D,21} = \bigcup_{P_{UU_1U_2X_1X_2|S}}\left\{
    \begin{aligned}
      &R_1 \leq \min\left\{I(U_1;Y|U,U_2) - I(U_1;Z|U_2,U,S)-H(S|U,U_1,U_2,Y)+H(S|Z,U,U_2),\right.\\
   &\quad\quad\quad\quad\quad\quad\quad\quad\quad \left. I(U_1;Y|U,U_2)-H(S|U,U_1,U_2,Y)\right\}\\
      &R_0 + R_1 \leq \min\left\{I(U,U_1,U_2;Y) - I(U_1;Z|U_2,U,S)-H(S|U,U_1,U_2,Y)+H(S|Z,U,U_2),\right.\\
      &\quad\quad\quad\quad\quad\quad\quad\quad\quad \left. I(U,U_1,U_2;Y)-H(S|U,U_1,U_2,Y)\right\}
    \end{aligned}
  \right.
\end{equation*}
with joint distribution such that $P_SP_{UU_1U_2X_1X_2|S}=P_SP_{U}P_{U_1|U}P_{U_2|U}P_{X_1|UU_1S}P_{X_2|UU_2S}$
and
\begin{equation*}
  \mathcal{R}^{CSI-E}_{D,22} = \bigcup_{P_{UU_1U_2X_1X_2|S}}\left\{
    \begin{aligned}
    &R_1 \leq \min\{H(S|Z,U,U_2)-H(S|U,U_1,U_2,Y),I(U_1;Y|U,U_2)-H(S|U,U_1,U_2,Y)\}\\
      &R_0 + R_1 \leq I(U,U_1,U_2;Y)-H(S|U,U_1,U_2,Y).
    \end{aligned}
  \right.
\end{equation*}
The wiretap code region is
\begin{equation*}
  \mathcal{R}^{CSI-E}_{D,3} = \bigcup_{P_{VUU_1U_2X_1X_2|S}}\left\{
    \begin{aligned}
      &R_1 \leq I(U_1;Y|U,U_2,V) - I(U_1;Z|U_2,U),\\
      &R_0 + R_1 \leq I(V,U,U_1,U_2;Y)-I(V;S)- I(U_1;Z|U_2,U)
    \end{aligned}
  \right.
\end{equation*}
with joint distribution such that $P_{SV}P_{UU_1U_2}P_{X_1|UU_1S}P_{X_2|UU_2S}$. Let $\mathcal{R}^{CSI-E}_D$ be the convex closure of $\mathcal{R}^{CSI-E}_{D,11} \cup \mathcal{R}^{CSI-E}_{D,12}\cup\mathcal{R}^{CSI-E}_{D,21}\cup\mathcal{R}^{CSI-E}_{D,22}\cup\mathcal{R}^{CSI-E}_{D,3}$. 
\begin{theorem}\label{the: inner bound of degraded message sets}
  For an SD-MAWC with degraded message sets and causal CSI at both encoders, it holds that
  \begin{align*}
    \mathcal{R}^{CSI-E}_{D} \subseteq C^{CSI-E}_{D}.
  \end{align*}
\end{theorem}
Region $\mathcal{R}^{CSI-E}_{D,11}$ and $\mathcal{R}^{CSI-E}_{D,12}$ follows by setting $R_{20}=R_{21}=0$ and letting the common codeword $\boldsymbol{U}$ be determined by both $k_{0,b}$ and the common message in Section \ref{sec: coding scheme for R11}. We give the coding scheme in Appendix \ref{app: proof of degraded message sets} for completeness. Region $\mathcal{R}^{CSI-E}_{D,21}$ is obtained by using the block Markov coding scheme and secret key agreement as in Section \ref{codin scheme for R21}. In each block, generate an additional common message codebook $\{\boldsymbol{u}(m_0)\}_{m_0=1}^{2^{nR_0}}$ according to some fixed distribution $P_U$ with size $2^{nR_0}$. For each $\boldsymbol{u}(m_0)$, generate a codebook for Sender 1 according to a fixed distribution $P_{U_1|U}$ with size $2^{n\widetilde{R}_1}$(defined in Section \ref{codin scheme for R21}) and a codeword $\boldsymbol{u}_2(m_0)$ for Sender 2. Partitioning the codebooks for Sender 1 exactly the same as that in Section \ref{codin scheme for R21} and setting $R_{20}=R_{21}=R_{22}=0$ for Sender 2 give the desired region. Region $\mathcal{R}^{CSI-E}_{D,22}$ follows by using the same coding scheme and setting $R_{10}=0$. The wiretap coding region $\mathcal{R}^{CSI-E}_{D,3}$ follows by setting the key rate $R_{11}=0$ in Appendix \ref{app: proof of degraded message sets} and skipping the \emph{Key Message Codebook Generation}.
\begin{remark}
  Removing channel states by setting $V=S=\emptyset$ and $U_2=X_2$, the achievable rate result reduces to the secrecy capacity of multiple access channels with one confidential message\cite[Theorem 5]{liang2008multiple} considering strong secrecy constraint. Further set $R_0=0$ and $U=X_2 = \emptyset$, the result reduces to the secrecy capacity of broadcast channels with confidential messages with common message rate being zero \cite[Corollary 2]{csiszar1978broadcast} considering strong secrecy constraint.
\end{remark}

The degraded message SD-MAC also gives us an example that Coding scheme 1 outperforms Coding scheme 2. Here we consider first the channel model in Example \ref{example: R1 > R_2} with degraded messages.

\begin{example}\label{example: one region includes another}
  For SD-MAC with causal CSI described by 
    \begin{equation*}
      \begin{split}
        &Y=\left\{
          \begin{aligned}
          X_1, \quad\quad\quad\quad \text{if}\;S=0,\\
          X_2, \quad\quad\quad\quad \text{if}\;S=1.\\
        \end{aligned}
          \right.\\
      &Z=X_1.
      \end{split}
    \end{equation*}
with $p=Pr\{S=1\}=0.8$ and degraded message sets, Coding scheme 1 outperforms Coding scheme 2 and its achievable region includes the region of Coding scheme 2.
\end{example}
\begin{proof}
  By the definition of the wiretap channel and \eqref{neq: example ine}, it follows that
  \begin{align}
    I(U_1;Y|V,U,U_2) \leq I(U_1;Y|S,U,U_2)
  \end{align}
and
\begin{align*}
  &I(U_1;Y,X_1|S,U,U_2)\\
  &=I(U_1;X_1|S,U,U_2) + I(U_1;Y|S,U,X_1,U_2)\\
  &\overset{(a)}{=}I(U_1;X_1|S,U,U_2)\\
  &=I(U_1;Y|S,U,U_2) + I(U_1;X_1|S,U,U_2,Y),
\end{align*}
where $(a)$ follows by $I(U_1;Y|S,U,X_1,U_2)=Pr\{S=0\}I(U_1;X_1|S=0,U,X_1,U_2)+Pr\{S=1\}I(U_1;X_2|S=1,U,X_1,U_2)=0$. By $I(U_1;X_1|S,U,U_2)=I(U_1;Z|S,U,U_2)$ and $I(U_1;X_1|S,U,U_2,Y)\geq 0$, we have $I(U_1;Y|S,U,U_2)\leq I(U_1;Z|S,U,U_2)$ and hence, the region for Coding scheme 1 reduces to $\mathcal{R}^{CSI-E}_{D,12}$ and similarly the region for Coding scheme 2 reduces to $\mathcal{R}^{CSI-E}_{D,22}$. 

Now we proceed to compute the extreme points of the regions. For region $\mathcal{R}^{CSI-E}_{D,12}$ we set $V=S,U=\emptyset,U_1=X_1,U_2=X_2$ with $p_1=Pr\{X_1=0\}$ and $p_2=Pr\{X_2=0\}$. It follows that 
\begin{align*}
  I(S;Y)=H(Y)-H(Y|S) = h((1-p)\cdot p_1+p\cdot p_2)-(1-p)h(p_1)-ph(p_2)
\end{align*}
and
\begin{align*}
  I(U_1;Y|U,U_2,V)&=I(X_1;Y|X_2,S)=H(Y|X_2,S)-H(Y|X_1,X_2,S)=(1-p)h(p_1).
\end{align*}
Setting $p_1=0.2$ and $p_2=0.8$, the rate of Sender 1 in Coding scheme 1 satisfies
\begin{align*}
  R_1 \leq \min\{I(S;Y),I(X_1;Y|X_2,S)\} =\min\{0.182,0.144\}=0.144
\end{align*}
For region $\mathcal{R}^{CSI-E}_{D,22}$, we have
\begin{align*}
  I(U_1;Y|U,U_2)-H(S|U,U_1,U_2,Y) \leq \max_{u}I(U_1;Y|U=u,U_2).
\end{align*}
Now similar to Example \ref{example: R1 > R_2} and \cite[Example 4]{lapidoth2012multiple}, the maximum rate is equivalent to the capacity of Z-channel
\begin{align*}
  \log(1+(1-p)p^{\frac{p}{1-p}})=0.114.
\end{align*}
Hence, we have the rate $R_1$ in region $\mathcal{R}^{CSI-E}_{D,22}$ satisfying
\begin{align*}
  R_1 \leq I(U_1;Y|U,U_2)-H(S|U,U_1,U_2,Y) \leq 0.114.
\end{align*}
For the sum rate, by setting $V=S$, it follows that
\begin{align*}
  I(V,U,U_1,U_2;Y)-I(V;S) &= I(S,U,U_1,U_2;Y)-H(S)\\
  &=I(U,U_1,U_2;Y) + I(S;Y|U,U_1,U_2)-H(S)\\
  &=I(U,U_1,U_2;Y) + H(S|U,U_1,U_2) - H(S|Y,U,U_1,U_2)-H(S)\\
  &\overset{(a)}{=}I(U,U_1,U_2;Y)- H(S|Y,U,U_1,U_2),
\end{align*}
where $(a)$ follows by the independence between $S$ and $(U,U_1,U_2)$. It follows that the sum rate bound in region $\mathcal{R}^{CSI-E}_{D,12}$ covers the sum rate bound in region $\mathcal{R}^{CSI-E}_{D,22}$ and when we set the distributions of random variables as above, we have the sum rate
\begin{align*}
  &I(X_1,X_2;Y)- H(S|Y,X_1,X_2)\\
  &=H(Y)-H(Y|X_1,X_2)-H(S|Y,X_1,X_2)\\
  &\overset{(a)}{=}h((1-p)\cdot p_1+p\cdot p_2) -h(p)=0.182,
\end{align*} 
where $(a)$ holds since for $H(Y|X_1,X_2)$, $Y$ is determined if $X_1=X_2$, and when $X_1\neq X_2,Y$ is determined by $S$, which implies $H(Y|X_1,X_2)=(p_1(1-p_2)+(1-p_1)p_2)h(p)$. For $H(S|Y,X_1,X_2)$, $S$ can be determined given $Y$ and $X_1\neq X_2$, and this gives $H(S|Y,X_1,X_2)=(p_1p_2+(1-p_1)(1-p_2))h(p)$.

\begin{figure}
  \centering
  \includegraphics[scale=0.8]{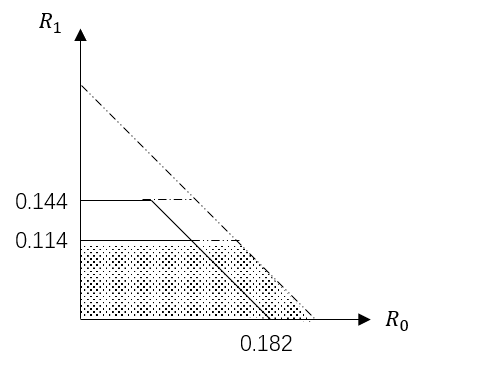}
  \caption{Achievable region of Example \ref{example: one region includes another}}
  \label{fig: one region includes another}
\end{figure}
Fig. \ref{fig: one region includes another} illustrates the achievable region in Example \ref{example: one region includes another}. The region within solid lines represents the rate region we actually achieve in this example by Coding scheme 1 and the shadowed part is the region of Coding scheme 2. The extended region within dotted lines is the rate region that is potentially achievable since we do not know the maximum rate of Sender 1 in Coding scheme 1 and the sum rate in both coding schemes. By changing the distribution, the sum rate may be larger and this gives the dotted line. However, since the maximum rate of Sender 1 in Coding scheme 2 is strictly smaller than that in Coding scheme 1, and the sum rate in Coding scheme 2 is always covered by that in Coding scheme 1, as implied by Fig. \ref{fig: one region includes another}, region $\mathcal{R}^{CSI-E}_{D,22}$ is included in region $\mathcal{R}^{CSI-E}_{D,12}$.
\end{proof}

In the following paragraphs, we consider several variations of the channel model. First, we consider the case that causal CSI is only available at the sender who sends both messages, while another sender knows nothing about the channel states. Then, we take one step further by revealing strictly causal CSI to the uninformed sender. Finally, the channel model that senders have knowledge about the state sequences but the channel is independent of states is discussed.

\subsubsection{One-Side Causal Channel State Information}\label{sec: one side causal CSI case}
Consider the case that the casual side information is only available at the sender that sends both messages. Such a model was studied in \cite{somekh2008cooperative} without secrecy constraint and the corresponding capacity result was given by using Shannon strategy. There is no cooperation about compressing state information between the senders since one of them cannot observe the channel states. Define $\mathcal{R}^{OCSI-E}_{D,21},\mathcal{R}^{OCSI-E}_{D,22},\mathcal{R}^{OCSI-E}_{D,3}$ by setting $U_2=X_2,V=\emptyset$ in $\mathcal{R}^{CSI-E}_{D,21},\mathcal{R}^{CSI-E}_{D,22},\mathcal{R}^{CSI-E}_{D,3}$ and the joint distribution $P_SP_{UU_1X_1X_2|S}=P_SP_{U}P_{U_1|U}P_{X_2|U}P_{X_1|UU_1S}$, and let $\mathcal{R}^{OCSI-E}_{D}$ be the convex closure of 
\begin{align*}
  \mathcal{R}^{OCSI-E}_{D,21}\cup \mathcal{R}^{OCSI-E}_{D,22}\cup\mathcal{R}^{OCSI-E}_{D,3}.
\end{align*}
\begin{proposition}
  For SD-MAWCs with degraded message sets and casual channel state information at the sender who sends both messages, it holds that $\mathcal{R}^{OCSI-E}_D\subseteq C^{OCSI-E}_D$.
\end{proposition}
$\mathcal{R}^{OCSI-E}_{D,21},\mathcal{R}^{OCSI-E}_{D,22}$ follows by setting $U_2=X_2$ in Theorem \ref{the: inner bound of degraded message sets}. Region $\mathcal{R}^{OCSI-E}_{D,3}$ follows by setting $V=\emptyset$ in Theorem \ref{the: inner bound of degraded message sets}. The details are omitted here. By setting $U=X_2$ and removing the wiretap channel, the achievable region reduces to the capacity of SD-MACs with degraded message sets and causal channel state information at the encoder sending both messages\cite[Theorem 4]{somekh2008cooperative}.

Next, we consider a special case of the SD-MAWCs with degraded message sets, where the wiretap channel is the degraded version of the main channel, i.e. $W(y,z|x_1,x_2,s)=W(y|x_1,x_2,s)P_{Z|Y}(z|y)$.

Given $P_S$ and $P_{YZ|X_1X_2S}$, define region
\begin{equation*}
  R^{OCSI}_{OUT} = \bigcup_{P_{UX_1X_2|S}} \left\{
    \begin{aligned}
      &R_1 \leq \min\{I(X_1,S;Y|U,X_2) - I(X_1,S;Z|U,X_2), I(X_1;Y|S,U,X_2)\},\\
      &R_0 + R_1 \leq \min\{I(X_1,X_2,S;Y) - I(X_1,S;Z|U,X_2), I(X_1,X_2;Y|S)\}
    \end{aligned}
  \right.
\end{equation*}
with joint distribution such that  $P_SP_{UX_1X_2|S}=P_SP_UP_{X_1|US}P_{X_2|U}$
and
\begin{equation*}
  R^{OCSI}_{IN} = \bigcup_{P_{UX_1X_2|S}}\left\{
    \begin{aligned}
      &R_1 \leq \min\{I(X_1,S;Y|U,X_2) - I(X_1,S;Z|U,X_2), I(X_1,S;Y|U,X_2)-H(S)\},\\
      &R_0 + R_1 \leq \min\{I(X_1,X_2,S;Y) - I(X_1,S;Z|U,X_2), I(X_1,X_2,S;Y)-H(S)\}
    \end{aligned}
  \right.
\end{equation*}
with joint distribution such that $P_{UX_1X_2S}=P_SP_UP_{X_1|US}P_{X_2|U}$.
\begin{theorem}\label{the: degraded mac with one-side causal CSI}
  The capacity of degraded SD-MAWCs with degraded message sets and one-side causal CSI at the encoder who sends both messages is bounded by 
  \begin{align*}
    R^{OCSI}_{IN} \subseteq C^{OCSI-E}_D\subseteq R^{OCSI}_{OUT}.
  \end{align*}
\end{theorem}
The proof is given in Appendix \ref{app: proof of the degraded mac with one-side causal CSI}. Note that $I(X_1;Y|S,U,X_2) = I(X_1,S;Y|U,X_2)-I(S;Y|U,X_2)$. The gap between inner and outer bounds is because $H(S|Y)$ and $H(S|Y,U,X_2)$. Hence, the following capacity result is straightforward. For given joint distribution $P_SP_{UX_1X_2|S}=P_SP_UP_{X_1|US}P_{X_2|U}$, define region
\begin{equation*}
  R^{CSI-ED}_D(P_{UX_1X_2|S}) = \left\{
    \begin{aligned}
      &R_1 \leq \min\{I(X_1;Y|U,X_2,S) - I(X_1;Z|U,X_2,S)+H(S|Z,U,X_2), I(X_1;Y|U,X_2,S)\},\\
      &R_0 + R_1 \leq \min\{I(X_1,X_2;Y|S) - I(X_1;Z|U,X_2,S)+H(S|Z,U,X_2), I(X_1,X_2;Y|S)\}
    \end{aligned}
  \right.
\end{equation*}
\begin{corollary}\label{coro: one side capacity}
  The capacity of degraded MACs with degraded message sets and one-side causal CSI and informed decoder is
  \begin{equation*}
    C^{OCSI-ED}_D = \bigcup_{P_{UX_1X_2|S}}R^{CSI-ED}_D(P_{UX_1X_2|S})
  \end{equation*}
where $|\mathcal{U}|$ satisfies $|\mathcal{U}|\leq |\mathcal{S}|\cdot|\mathcal{X}_1|\cdot|\mathcal{X}_2|+1$.
\end{corollary}
\begin{proof}
  The region is obtained by setting $Y=(Y,S)$ in Theorem \ref{the: degraded mac with one-side causal CSI}. The cardinality of the auxiliary random variable set is bounded by using the support lemma\cite{el2011network} to preserve the values of  $I(X_1;Y|U,X_2,S)$, $I(X_1;Y|X_2,S)$ and $I(X_1;Z|X_2,S)-H(S|Z,X_2)$, which are $|\mathcal{S}|\cdot|\mathcal{X}_1|\cdot|\mathcal{X}_2|+1$ continuous function of some distributions of $(U,S,X_1,X_2)$. This completes the proof.
\end{proof}
\begin{remark}
  Setting $U=X_2=\emptyset$, the capacity result reduces to Corollary \ref{coro: single user capacity} for the single-user case.
\end{remark}
Similar to the single-user case, the coding schemes achieve channel capacity when state information is also revealed to the decoder. In the following example, we consider a channel model that the capacity can be achieved without an informed decoder. The model originates from the binary example in \cite[Theorem 7]{liang2008multiple}, and hence a similar technique is used in the proof.

\begin{example}\label{example: state reproduced}
  Consider a state-dependent MAC with alphabets $\mathcal{X}_1=\mathcal{X}_2=\mathcal{S}=\mathcal{Z}=\mathcal{N}=\{0,1\}$ and $\mathcal{Y}=\{0,1\}\times\{0,1\}$. The main channel output contains two components, i.e. $Y=(Y_1,Y_2)$. Let $\odot$ be binary multiplication and $\oplus$ be modulo sum. The channel model is described by
  \begin{align*}
    Y_1 &= (X_1 \odot X_2) \oplus S,\\
    Y_2 &= X_1\odot X_2,\\
    Z &= Y_2 \oplus N,
  \end{align*}
  where $S\backsim Bernoulli(q),N\backsim Bernoulli(p), q,p\in [0,\frac{1}{2}]$, and the causal channel state information is available at the sender who sends both messages. The capacity of this channel is
  \begin{equation*}
    \bigcup_{\alpha\in[\frac{1}{2},1]}\left\{ 
      \begin{aligned}
        &R_1 \leq \min\{h(\alpha) + h(q) - h(p*\alpha) + h(p), h(\alpha) \},\\
        &R_1 + R_2 \leq \min\{1 + h(q) - h(\alpha*p) + h(p), 1\}.
      \end{aligned}
    \right.
  \end{equation*}
  where $h(\cdot)$ is the binary entropy such that $h(a)=-a\log(a)-(1-a)\log(1-a)$.
\end{example}
\begin{proof}
  \emph{Converse: } By formula \eqref{neq: R1 upper bound 1} in Appendix \ref{app: proof of the degraded mac with one-side causal CSI}, it follows that (for simplicity, we omit $\delta$.)
  \begin{align*}
    R_1 &\leq I(\boldsymbol{X}_1,\boldsymbol{S};\boldsymbol{Y}|M_0,\boldsymbol{X}_2) - I(\boldsymbol{X}_1,\boldsymbol{S};\boldsymbol{Z}|M_0,\boldsymbol{X}_2)\\
    &=(H(\boldsymbol{Y}|M_0,\boldsymbol{X}_2) - H(\boldsymbol{Y}|M_0,\boldsymbol{X}_2,\boldsymbol{X}_1,\boldsymbol{S})) - (H(\boldsymbol{Z}|M_0,\boldsymbol{X}_2)-H(\boldsymbol{Z}|M_0,\boldsymbol{X}_2,\boldsymbol{X}_1,\boldsymbol{S}))\\
    &=(H(\boldsymbol{Y}_2|M_0,\boldsymbol{X}_2) + H(\boldsymbol{Y}_1|M_0,\boldsymbol{X}_2,\boldsymbol{Y}_2)) - (H(\boldsymbol{Z}|M_0,\boldsymbol{X}_2)-nh(p))\\
    &\overset{(a)}{\leq} n + nh(q) - (H(\boldsymbol{Y}_2 \oplus N|M_0,\boldsymbol{X}_2)-nh(p)),
  \end{align*}
  where $(a)$ follows by $H(\boldsymbol{Y}_2|M_0,\boldsymbol{X}_2)\leq \sum_{i=1}^n H(Y_i)\leq n$.
  Hence, there exists some $\alpha \in [\frac{1}{2},1]$ such that
  \begin{align*}
    nR_1 \leq nh(\alpha) + nh(q) - (H(\boldsymbol{Y}_2 \oplus N|M_0,\boldsymbol{X}_2)-nh(p)).
  \end{align*}
  To bound $H(\boldsymbol{Y}_2 + N|M_0,\boldsymbol{X}_2)$, we invoke Lemmas 5 and 6 in \cite{liang2008multiple} (proposed in \cite{wyner1973theorem}).
  \begin{align*}
    &H(\boldsymbol{Y}_2 \oplus N|M_0,\boldsymbol{X}_2)\\
    &= \mathbb{E}_{M_0,\boldsymbol{X}_2} \left[ H(\boldsymbol{Y}_2 \oplus N|M_0=m_0,\boldsymbol{X}_2=\boldsymbol{x}_2) \right] \\
    &\overset{(a)}{\geq} \mathbb{E}_{M_0,\boldsymbol{X}_2} \left[ h\left( p*h^{-1}\left( \frac{H(\boldsymbol{Y}_2|M_0=m_0,\boldsymbol{X}_2=\boldsymbol{x}_2)}{n} \right) \right) \right]\\
    &\overset{(b)}{\geq} nh\left( p*h^{-1}\left(  \mathbb{E}_{M_0,\boldsymbol{X}_2} \left[\frac{H(\boldsymbol{Y}_2|M_0=m_0,\boldsymbol{X}_2=\boldsymbol{x}_2)}{n} \right]\right) \right) \\
    &=nh\left( p*h^{-1}\left(  \frac{H(\boldsymbol{Y}_2|M_0,\boldsymbol{X}_2)}{n}\right) \right) \\
    &=nh\left( p*h^{-1}\left(  \frac{nh(\alpha)}{n}\right) \right) =n h(p*\alpha).
  \end{align*}
  where $(a)$ follows by Lemma 6 in \cite{liang2008multiple}, $(b)$ follows by Lemma 5 in \cite{liang2008multiple}. Combining previous bounds yields
  \begin{align}
    \label{neq: numerical example upper bound R1 1}nR_1 \leq nh(\alpha) + nh(q) - nh(p*\alpha) + nh(p).
  \end{align}
  For the second constraint, by Appendix \ref{app: proof of the degraded mac with one-side causal CSI} we have
  \begin{align}
    nR_1 &\leq I(\boldsymbol{X}_1;\boldsymbol{Y}|M_0,\boldsymbol{X}_2,\boldsymbol{S})\notag\\
    &=I(\boldsymbol{X}_1,\boldsymbol{S};\boldsymbol{Y}|M_0,\boldsymbol{X}_2) - I(\boldsymbol{S};\boldsymbol{Y}|M_0,\boldsymbol{X}_2)\notag\\
    &\overset{(a)}{=}I(\boldsymbol{X}_1,\boldsymbol{S};\boldsymbol{Y}|M_0,\boldsymbol{X}_2) - H(\boldsymbol{S})\notag\\
    \label{neq: numerical example upper bound R1 2}&=nh(\alpha) + nh(q) - nh(q)=nh(\alpha).
  \end{align}
  where $(a)$ follows by the fact that $\boldsymbol{S}$ is independent to $(M_0,\boldsymbol{X}_2)$ and the receiver can always recover $\boldsymbol{S}$ by $S_i = Y_{1,i} \ominus Y_{2,i}$.
  
  For the first constraint of the sum rate, by Appendix \ref{app: proof of the degraded mac with one-side causal CSI} we have 
  \begin{align}
    n(R_1 + R_2) &\leq I(M_0;\boldsymbol{Y}) + nR_1 \notag\\
    &\leq I(M_0,\boldsymbol{X}_2;\boldsymbol{Y}) + nR_1\notag\\
    &\leq I(M_0,\boldsymbol{X}_1,\boldsymbol{X}_2,\boldsymbol{S};\boldsymbol{Y}) - I(\boldsymbol{X}_1,\boldsymbol{S};\boldsymbol{Z}|M_0,\boldsymbol{X}_2)\notag \\
    \label{neq: numerical example upper bound R1+R2 1}&\leq n + nh(q) - nh(\alpha*p) + nh(p).
  \end{align}
  The second constraint on the sum rate is bounded by 
  \begin{align}
    n(R_1 + R_2) &\leq I(M_0;\boldsymbol{Y}) + nR_1\notag\\
    &\overset{(a)}{\leq} I(M_0,\boldsymbol{X}_2;\boldsymbol{Y}|\boldsymbol{S}) + I(\boldsymbol{X}_1;\boldsymbol{Y}|M_0,\boldsymbol{X}_2,\boldsymbol{S})\notag\\
    &= I(M_0,\boldsymbol{X}_1,\boldsymbol{X}_2;\boldsymbol{Y}|\boldsymbol{S})\notag\\
    &=I(\boldsymbol{X}_1,\boldsymbol{X}_2,\boldsymbol{S};\boldsymbol{Y}) - I(\boldsymbol{S};\boldsymbol{Y})\notag\\
    &\overset{(b)}{=}H(\boldsymbol{Y})-H(\boldsymbol{S})\notag\\
    &=H(\boldsymbol{Y}_1,\boldsymbol{Y}_2)-H(\boldsymbol{S})\notag\\
    &=H(\boldsymbol{Y}_2) + H(\boldsymbol{Y}_1|\boldsymbol{Y}_2)-H(\boldsymbol{S})\notag\\
    \label{neq: numerical example upper bound R1+R2 2}&\leq n + nh(q)-nh(q)=n.
  \end{align}
  where $(a)$ follows by the nonnegativity of mutual information and the independence between $\boldsymbol{S}$ and $(M_0,\boldsymbol{X}_2)$, $(b)$ follows by the fact that the receiver can recover channel state sequences by $S_i = Y_{1,i} \ominus Y_{2,i}$. Combining formulas \eqref{neq: numerical example upper bound R1 1}-\eqref{neq: numerical example upper bound R1+R2 2} gives the outer bound
  \begin{equation*}
    \left\{ 
      \begin{aligned}
        &R_1 \leq \min\{h(\alpha) + h(q) - h(p*\alpha) + h(p), h(\alpha) \},\\
        &R_1 + R_2 \leq \min\{1 + h(q) - h(\alpha*p) + h(p), 1\}.
      \end{aligned}
    \right.
  \end{equation*}
\emph{Direct Part: } To prove the achievability, set $Pr\{U=1\}=\frac{1}{2},Pr\{X_2 = 1\}=1$. Define a random variable $X'$ independent to $(U,X_2,S,N)$ with $Pr\{X'=1\}=\beta$ such that $q*\beta=\alpha$ for some $\alpha\in[\frac{1}{2},1]$. Define $X_1 = U \oplus S \oplus X'$. Applying Theorem \ref{the: degraded mac with one-side causal CSI} to the model, we obtain 
\begin{align}
  R_1 &\leq I(X_1,S;Y|U,X_2) - I(X_1,S;Z|U,X_2)\notag\\
  &=H(Y|U,X_2=1)-H(Y|U,X_1,X_2,S) - (H(Z|U,X_2)-H(Z|U,X_1,X_2,S))\notag\\
  &=\frac{1}{2}H(Y|U=0,X_2=1)+\frac{1}{2}H(Y|U=1,X_2=1) \notag\\
  &\quad\quad\quad\quad\quad - (\frac{1}{2}H(Z|U=0,X_2=1)+\frac{1}{2}H(Z|U=1,X_2=1)-h(p)),\notag
\end{align}
where
\begin{align*}
  H(Y|U=0,X_2=1)=H(Y|U=1,X_2=1)&=H(Y_2|U=0,X_2=1) + H(Y_1|U=0,X_2=1,Y_2)\\
  &=h(S\oplus X'|U=0,X_2=1) + h(q) = h(\alpha) + h(q),
\end{align*}
and 
\begin{align*}
  H(Z|U=0,X_2=1) = H(Z|U=0,X_2=1) = h(\alpha*p).
\end{align*}
Combining the above inequalities together gives
\begin{align*}
  R_1 \leq h(\alpha) + h(q) - h(\alpha*p) + h(p).
\end{align*}\
The second constraint is 
\begin{align*}
  R_1 &\leq I(X,S;Y|U,X_2) - H(S)\\
  &=H(Y|U,X_2) - H(Y|U,X_1,X_2,S) - h(q)\\
  &=h(\alpha).
\end{align*}
The sum rates are bounded by 
\begin{align*}
  R_1 + R_2 &\leq I(X_1,X_2,S;Y) - I(X_1,S;Z|U,X_2)\\
  &=1 + h(q) - h(\alpha*p) + h(p)
\end{align*}
and
\begin{align*}
  R_1 + R_2 &\leq I(X_1,X_2,S;Y) - H(S) = 1 + h(q) - h(q) = 1.
\end{align*}
The proof is completed.
\end{proof}
In the above example, the inner and outer bounds meet each other even when the state information is not revealed to the decoder. This is because the legitimate receiver can always reproduce the channel state based on the channel output by the relation $S = Y_1 \ominus Y_2$. Such a `state-reproducing' scheme was also discussed in \cite{sasaki2019wiretap} for single-user case.
\subsubsection{Causal CSI at One Encoder and Strictly Causal CSI at the Other}\label{sec: causal and strictly causal CSI case}
In this section, we further extend the model to the case that causal CSI is available to the sender that sends both messages and strictly causal CSI is available at another sender. Both senders observe the channel state sequences and hence, they can send the messages by cooperation, which means Coding scheme 1 can be used here. Define $\mathcal{R}^{CSI-SCSI}_{D,11},\mathcal{R}^{CSI-SCSI}_{D,12},\mathcal{R}^{CSI-SCSI}_{D,3}$ by setting $U_2=X_2$ in $\mathcal{R}^{CSI-E}_{D,11},\mathcal{R}^{CSI-E}_{D,12},\mathcal{R}^{CSI-E}_{D,3}$ with joint distribution $P_SP_{VUU_1X_1X_2|S}=P_{VS}P_{UU_1}P_{X_1|UU_1S}P_{X_2|US}$. Let $\mathcal{R}^{CSI-SCSI}_D$ be the convex closure of $\mathcal{R}^{CSI-SCSI}_{D,11}\cup\mathcal{R}^{CSI-SCSI}_{D,12}\cup\mathcal{R}^{OCSI-E}_{D,21}\cup \mathcal{R}^{OCSI-E}_{D,22}\cup\mathcal{R}^{CSI-SCSI}_{D,3}$.
\begin{proposition}
  For SD-MAWCs with degraded message sets and causal CSI at the sender that sends both messages and strictly causal CSI at the other, it holds that $\mathcal{R}^{CSI-SCSI}_D \subseteq \mathcal{C}^{CSI-SCSI}_D$.
\end{proposition}
\begin{proof}
  The proof of $\mathcal{R}^{CSI-SCSI}_{D,11},\mathcal{R}^{CSI-SCSI}_{D,12}$ remains exactly as the coding scheme in Appendix \ref{app: proof of degraded message sets} for Theorem \ref{the: inner bound of degraded message sets} except the codewords of Sender 2 are generated according to distribution $P_{X_2|U}$ since the channel state of current transmission is not available.
\end{proof}
Now we prove both coding schemes are optimal for this model with some additional conditions. Suppose the wiretap channel is a degraded version of the main channel, and the state sequence is also available at the legitimate receiver. For simplicity, we name the sender that sends both messages as the main sender.
\begin{corollary}\label{coro: causal and strictly causal capacity}
    The capacity of degraded multiple access wiretap channels with degraded message sets, causal CSI at the main sender and strictly causal CSI at the other, and an informed decoder, is $$C^{CSI-SCSI-ED}_D = C^{OCSI-ED}_D.$$
\end{corollary}
The corollary proves that for degraded multiple access wiretap channels with informed decoder, degraded messages sets and causal CSI at the main sender, revealing strictly causal CSI to another sender does not improve the capacity.
\begin{proof}
  \emph{Converse:} The converse proof follows by setting $Y=(Y,S)$ in $R^{OCSI}_{OUT}$ with the fact that $S$ is independent to $(U,X_2)$.

  \emph{Achievability: } We first prove the capacity can be achieved by Coding scheme  1. This is shown by setting $V=S,Y=(Y,S)$ in $\mathcal{R}^{CSI-SCSI}_{D,11}$ and using strong functional representation lemma again. The achievability of Coding scheme 2 follows directly by setting $Y=(Y,S)$ in $R^{OCSI}_{IN}$ with the fact that $S$ is independent of $(U,X_2)$.
\end{proof}

\subsubsection{Channel Is Independent of the States}\label{sec: channel independent of state case}

Following \cite{chia2012wiretap}\cite{sasaki2019wiretap}, we consider a special case of multiple access wiretap channels such that $P_{YZ|X_1X_2S}=P_{YZ|X_1X_2}$. We further assume the wiretap channel is a degraded version of the main channel. 
For such a channel model, the following capacity result holds.
\begin{corollary}
  The capacity of state-independent degraded multiple access wiretap channels with degraded message sets is
  \begin{equation*}
    C^{CSI-E}_D = \bigcup_{P_UP_{X_1|U}P_{X_2|U}}\left\{
      \begin{aligned}
        &R_1 \leq I(X_1;Y|U,X_2) - I(X_1;Z|U,X_2),\\
        &R_0 + R_1 \leq I(X_1,X_2;Y) - I(X_1;Z|U,X_2)
      \end{aligned}
    \right.
  \end{equation*}
where $|\mathcal{U}|$ satisfies $|\mathcal{U}|\leq |\mathcal{X}_1|\cdot|\mathcal{X}_2|+1$.
\end{corollary}
\begin{proof}
  The achievability follows by setting $V=\emptyset$ and $U_1=X_1,U_2=X_2$ independent of $S$ in region $\mathcal{R}^{CSI-E}_{D,3}$. The converse is proved by letting $S$ independent of all the other random variables in $R^{OCSI}_{OUT}$.
\end{proof}
Note that similar to \cite{sasaki2019wiretap}, when the channel is independent of the states, although senders observe the channel state sequences as an additional source, the secret key constructions in both coding schemes are invalid since our key constructions are based on the noisy observation of state sequences at the decoder side (channel output). In the case that the channel is independent of the states, there is no correlation between the states and the channel output. Hence, the model behaves like a normal multiple access wiretap channel, and the capacity result reduces to the capacity of multiple access channels with one confidential message in \cite[Corollary 3]{liang2008multiple} considering strong secrecy constraint.

\section{Conclusion}
Although a series of works about point-to-point wiretap channels with causal CSI under different secrecy constraints have been reported,  the secure transmission over the SD-MAC with causal CSI has not been fully explored yet. The most related result to this paper was given in \cite{dai2014multiple}, where a single letter characterization of SD-MAWCs with causal CSI under weak secrecy constraint was provided as follows.
\begin{theorem}{\cite[Theorem 4]{dai2014multiple}}\label{the: existing result}
  A single-letter characterization of the capacity region of SD-MAWCs with causal CSI is as follows. Let $\mathcal{R}^{CSI-E} = \mathcal{R}_1 \cup \mathcal{R}_2,$
  where 
  \begin{equation*}
    \mathcal{R}_1 = \bigcup_{P_{U_1U_2X_1X_2|S}}\left\{  
      \begin{aligned}
        &R_1 \leq I(U_1;Y) - I(U_1;Z);\\
        &R_2 \leq I(U_2;Y|U_1) - I(U_2;Z|U_1)
      \end{aligned}
    \right.,
    \mathcal{R}_2 = \bigcup_{P_{U_1U_2X_1X_2|S}} \left\{  
      \begin{aligned}
        &R_1 \leq I(U_1;Y|U_2) - I(U_1;Z|U_2);\\
        &R_2 \leq I(U_2;Y) - I(U_2;Z),
      \end{aligned}
    \right.
  \end{equation*}

  with joint distribution $P_SP_{U_1U_2X_1X_2|S}=P_SP_{U_1}P_{U_2}P_{X_1|U_1S}P_{X_2|U_2S}$. 
  It follows that $\mathcal{R}^{CSI-E}\subseteq C^{CSI-E}$.
\end{theorem}
This region is included in our inner bound (region $\mathcal{R}_{3}$). Theorem \ref{the: existing result} relies on wiretap coding, which was then combined with block Markov coding, channel resolvability and secret key construction to establish the lower bound of SD-WTCs with causal CSI at encoder under strong secrecy constraint in \cite{sasaki2019wiretap}. The motivation of this paper is the fact that Shannon strategy is not optimal for SD-MACs with causal CSI at encoders, and the secret key agreement by Csisz\'ar and Narayan\cite{csiszar2000common}. The secret key rate in the new coding scheme (referred to as Coding scheme  1) in general does not outperform that in the coding scheme of \cite{sasaki2019wiretap} (referred to as Coding scheme 2 in this paper). However, the new coding scheme achieves a larger achievable region for the main channel compared to the Coding scheme 2 using Shannon strategy. The regions by different coding schemes do not include each other in general. We prove in some cases, both coding schemes are optimal and achieve the channel capacities. Although it was discussed in \cite{sasaki2019wiretap} that forward decoding has a lower decoding delay when both coding schemes achieve the same fundamental limits, our numerical examples show both coding schemes achieve some rate pairs that cannot be achieved by another scheme, and region of Coding scheme 1 can be strictly larger than that of Coding scheme 2. Hence, we cannot drop any of them. Our results are in accordance with existing results for multi-user channels and point-to-point channels.


%

\vspace{-0.25cm}
\appendices
\section{Typical Sequences}\label{app: properties of typical sequences}
This section provides some properties of typical sequences defined in Section \ref{sec: notations}. These properties are used throughout the paper and may not be invoked explicitly. The proof of the following lemmas can be found in \cite{csiszar2011information}.
\begin{lemma}{\cite[Lemma 2.10]{csiszar2011information}}\label{lem: typicality}
  If $\boldsymbol{x}\in T^n_{P_X,\delta},\boldsymbol{y}\in T^n_{P_{Y|X},\delta'}[\boldsymbol{x}]$, then $(\boldsymbol{x},\boldsymbol{y})\in T^n_{P_{XY},\delta+\delta'}$ and $\boldsymbol{y}\in T^n_{P_Y,|\mathcal{X}|(\delta+\delta')}$. 
\end{lemma}
Implied by Lemma \ref{lem: typicality}, it follows that
\begin{lemma}{\cite[Eq. 17.15]{csiszar2011information}}\label{lem: typicality 2}
  For positive numbers $\zeta,\delta,\sigma$ such that $\zeta<\sigma$, if $\boldsymbol{x}\in T^n_{P_X,\zeta}$, then $T^n_{P_{Y|X},\sigma-\zeta}[\boldsymbol{x}]\subseteq T^n_{P_{XY},\sigma}[\boldsymbol{x}]$.
\end{lemma}
For joint distribution $P_{XY}$ with $P_X$ being the marginal distribution, the following lemma holds.
\begin{lemma}{\cite[Corollary 17.9A]{csiszar2011information}}\label{lem: typical sequences independent sequences}
  Consider $N=2^{nR}$ sequences $\boldsymbol{x}\in\mathcal{X}^n$, each i.i.d. generated according to distribution $P_X$. If $I(X;Y)<R$, to any $\tau>0$ there exists $\zeta>0$ such that 
  \begin{align*}
    \left| \frac{1}{n}\log \left| \{i:\boldsymbol{x}(i)\in T^n_{P_{XY},\zeta}[\boldsymbol{y}]\} \right| - (R-I(X;Y)) \right| < \tau
  \end{align*}
  simultaneously for all $\boldsymbol{y}\in \mathcal{Y}^n$ with $T^n_{P_{XY},\zeta}[\boldsymbol{y}]\neq \emptyset$.
\end{lemma}
\section{Proof of lemma \ref{lem: wiretap codes}}\label{app: proof of wiretap codes}
 The proof here uses similar arguments as that in \cite{helal2020cooperative}. Consider distribution
\begin{align*}
  P_{U_1U_2SZ|U}(u_1,u_2,s,z|u)=\sum_{x_1,x_2}P_S(s)P_{U_1|U}(u_1|u)P_{U_2|U}(u_2|u)P_{X_1|U_1US}(x_1|u,u_1,s)P_{X_2|U_2US}(x_2|u,u_2,s)P_{Z|X_1X_2S}(z|x_1,x_2,s)
\end{align*}
such that $Q_{SZ|U}(s,z|u)=\sum_{u_1,u_2}P_{U_1U_2SZ|U}(u_1,u_2,s,z|u)$.
Now it follows that
\begin{align*}
  &\mathbb{E}_{\textbf{C}\textbf{C}_1\textbf{C}_2}\left[ D(P_{\boldsymbol{Z}\boldsymbol{S}|\boldsymbol{U}\textbf{C}\textbf{C}_1\textbf{C}_2}||Q_{\boldsymbol{Z}\boldsymbol{S}|\boldsymbol{U}}) \right]\\
  &\overset{(a)}{=}\mathbb{E}_{\textbf{C}\textbf{C}_1\textbf{C}_2}\left[ \sum_{\boldsymbol{z},\boldsymbol{s}}\sum_{m}\bar{P}(m)P_{\boldsymbol{Z}\boldsymbol{S}|\boldsymbol{U}\textbf{C}\textbf{C}_1\textbf{C}_2}(\boldsymbol{z},\boldsymbol{s}|\boldsymbol{U}(m))\log\frac{P_{\boldsymbol{Z}\boldsymbol{S}|\boldsymbol{U}\textbf{C}\textbf{C}_1\textbf{C}_2}(\boldsymbol{z},\boldsymbol{s}|\boldsymbol{U}(m))}{Q_{\boldsymbol{Z}\boldsymbol{S}|\boldsymbol{U}}(\boldsymbol{z},\boldsymbol{s}|\boldsymbol{U}(m))} \right]\\
  &=\mathbb{E}_{\textbf{C}\textbf{C}_1\textbf{C}_2}\left[ \sum_{\boldsymbol{z},\boldsymbol{s}}\sum_{m}\bar{P}(m)\sum_{m_1,m_2}2^{-n(R_1+R_2)}P_{\boldsymbol{Z}\boldsymbol{S}|\boldsymbol{U}\boldsymbol{U}_1\boldsymbol{U}_2}(\boldsymbol{z},\boldsymbol{s}|\boldsymbol{U}(m),\boldsymbol{U}_1(m,m_1),\boldsymbol{U}_2(m,m_2))\right.\\
  &\quad\quad\quad\quad\quad \left. \log\frac{\sum_{l_1',l_2'}2^{-n(R_1+R_2)}P_{\boldsymbol{Z}\boldsymbol{S}|\boldsymbol{U}\boldsymbol{U}_1\boldsymbol{U}_2}(\boldsymbol{z},\boldsymbol{s}|\boldsymbol{U}(m),\boldsymbol{U}_1(m,l_1'),\boldsymbol{U}_2(m,l_2'))}{Q_{\boldsymbol{Z}\boldsymbol{S}|\boldsymbol{U}}(\boldsymbol{z},\boldsymbol{s}|\boldsymbol{U}(m))} \right]\\
  &=\sum_{\substack{\boldsymbol{u}(1),\boldsymbol{u}_1(1,1)\dots,\boldsymbol{u}_1(1,2^{nR_1})\\\boldsymbol{u}_2(1,1)\dots,\boldsymbol{u}_2(1,2^{nR_2})}}\sum_{\substack{\boldsymbol{u}(2),\boldsymbol{u}_1(2,1)\dots,\boldsymbol{u}_1(2,2^{nR_1})\\\boldsymbol{u}_2(2,1)\dots,\boldsymbol{u}_2(2,2^{nR_2})}}\dots \sum_{\substack{\boldsymbol{u}(2^{nR_0}),\boldsymbol{u}_1(2^{nR_0},1)\dots,\boldsymbol{u}_1(2^{nR_0},2^{nR_1})\\\boldsymbol{u}_2(2^{nR_0},1)\dots,\boldsymbol{u}_2(2^{nR_0},2^{nR_2})}}\\
  &\quad\quad\quad\quad\quad\quad\quad\quad\quad\quad\cdot\prod_{(l,l_1,l_2)=(1,1,1)}^{2^{nR_0},2^{nR_1},2^{nR_2}}P_{UU_1U_2}(\boldsymbol{u}(l),\boldsymbol{u}_1(l,l_1),\boldsymbol{u}_2(l,l_2))\\
  &\quad\quad\quad\quad\quad \left[ \sum_{\boldsymbol{z},\boldsymbol{s}}\sum_{m}\bar{P}(m)\sum_{m_1,m_2}2^{-n(R_1+R_2)}P_{\boldsymbol{Z}\boldsymbol{S}|\boldsymbol{U}\boldsymbol{U}_1\boldsymbol{U}_2}(\boldsymbol{z},\boldsymbol{s}|\boldsymbol{u}(m),\boldsymbol{u}_1(m,m_1),\boldsymbol{u}_2(m,m_2))\right.\\
  &\quad\quad\quad\quad\quad \left. \log\frac{\sum_{l_1',l_2'}2^{-n(R_1+R_2)}P_{\boldsymbol{Z}\boldsymbol{S}|\boldsymbol{U}\boldsymbol{U}_1\boldsymbol{U}_2}(\boldsymbol{z},\boldsymbol{s}|\boldsymbol{u}(m),\boldsymbol{u}_1(m,l_1'),\boldsymbol{u}_2(m,l_2'))}{Q_{\boldsymbol{Z}\boldsymbol{S}|\boldsymbol{U}}(\boldsymbol{z},\boldsymbol{s}|\boldsymbol{u}(m))} \right]\\
  &=2^{-n(R_1+R_2)}\sum_{\boldsymbol{z},\boldsymbol{s}}\sum_{m}\sum_{m_1,m_2}\sum_{\boldsymbol{u}(m)}\sum_{\boldsymbol{u}_1(m,m_1)}\sum_{\boldsymbol{u}_2(m,m_2)}\bar{P}(m)P_{\boldsymbol{Z}\boldsymbol{S}\boldsymbol{U}\boldsymbol{U}_1\boldsymbol{U}_2}(\boldsymbol{z},\boldsymbol{s},\boldsymbol{u}(m),\boldsymbol{u}_1(m,m_1),\boldsymbol{u}_2(m,m_2))\\
  &\quad\quad\quad\quad\quad\sum_{\substack{(k,k_1,k_2)\\\neq(m,m_1,m_2)}}\sum_{\boldsymbol{u}(k)}\sum_{\boldsymbol{u}_1(k,k_1)}\sum_{\boldsymbol{u}_2(k,k_2)}\prod_{(l,l_1,l_2)\neq(m,m_1,m_2)} P_{UU_1U_2}(\boldsymbol{u}(l),\boldsymbol{u}_1(l,l_1),\boldsymbol{u}_2(l,l_2))\\
  &\quad\quad\quad\quad\quad \log\frac{\sum_{l_1',l_2'}2^{-n(R_1+R_2)}P_{\boldsymbol{Z}\boldsymbol{S}|\boldsymbol{U}\boldsymbol{U}_1\boldsymbol{U}_2}(\boldsymbol{z},\boldsymbol{s}|\boldsymbol{u}(m),\boldsymbol{u}_1(m,l_1'),\boldsymbol{u}_2(m,l_2'))}{Q_{\boldsymbol{Z}\boldsymbol{S}|\boldsymbol{U}}(\boldsymbol{z},\boldsymbol{s}|\boldsymbol{u}(m))}\\
  &=2^{-n(R_1+R_2)}\sum_{\boldsymbol{z},\boldsymbol{s}}\sum_{m}\sum_{m_1,m_2}\sum_{\boldsymbol{u}(m)}\sum_{\boldsymbol{u}_1(m,m_1)}\sum_{\boldsymbol{u}_2(m,m_2)}\bar{P}(m)P_{\boldsymbol{Z}\boldsymbol{S}\boldsymbol{U}\boldsymbol{U}_1\boldsymbol{U}_2}(\boldsymbol{z},\boldsymbol{s},\boldsymbol{u}(m),\boldsymbol{u}_1(m,m_1),\boldsymbol{u}_2(m,m_2))\\
  &\quad\quad\quad\quad\quad \mathbb{E}_{\backslash(m,m_1,m_2)}\log\frac{\sum_{l_1',l_2'}2^{-n(R_1+R_2)}P_{\boldsymbol{Z}\boldsymbol{S}|\boldsymbol{U}\boldsymbol{U}_1\boldsymbol{U}_2}(\boldsymbol{z},\boldsymbol{s}|\boldsymbol{U}(m),\boldsymbol{U}_1(m,l_1'),\boldsymbol{U}_2(m,l_2'))}{Q_{\boldsymbol{Z}\boldsymbol{S}|\boldsymbol{U}}(\boldsymbol{z},\boldsymbol{s}|\boldsymbol{U}(m))}\\
  &=2^{-n(R_1+R_2)}\sum_{\boldsymbol{z},\boldsymbol{s}}\sum_{m}\sum_{m_1,m_2}\sum_{\boldsymbol{u}(m)}\sum_{\boldsymbol{u}_1(m,m_1)}\sum_{\boldsymbol{u}_2(m,m_2)}\bar{P}(m)P_{\boldsymbol{Z}\boldsymbol{S}\boldsymbol{U}\boldsymbol{U}_1\boldsymbol{U}_2}(\boldsymbol{z},\boldsymbol{s},\boldsymbol{u}(m),\boldsymbol{u}_1(m,m_1),\boldsymbol{u}_2(m,m_2))\\
  &\quad\quad\quad\quad\quad \mathbb{E}_{\backslash(m,m_1,m_2)}\log\left[\frac{2^{-n(R_1+R_2)}P_{\boldsymbol{Z}\boldsymbol{S}|\boldsymbol{U}\boldsymbol{U}_1\boldsymbol{U}_2}(\boldsymbol{z},\boldsymbol{s}|\boldsymbol{u}(m),\boldsymbol{u}_1(m,m_1),\boldsymbol{u}_2(m,m_2))}{Q_{\boldsymbol{Z}\boldsymbol{S}|\boldsymbol{U}}(\boldsymbol{z},\boldsymbol{s}|\boldsymbol{u}(m))}\right.\\
  &\quad\quad\quad\quad\quad\quad\quad\quad\quad\quad +\frac{\sum_{l_1'\neq m_1}2^{-n(R_1+R_2)}P_{\boldsymbol{Z}\boldsymbol{S}|\boldsymbol{U}\boldsymbol{U}_1\boldsymbol{U}_2}(\boldsymbol{z},\boldsymbol{s}|\boldsymbol{u}(m),\boldsymbol{U}_1(m,l_1'),\boldsymbol{u}_2(m,m_2))}{Q_{\boldsymbol{Z}\boldsymbol{S}|\boldsymbol{U}}(\boldsymbol{z},\boldsymbol{s}|\boldsymbol{u}(m))}\\
  &\quad\quad\quad\quad\quad\quad\quad\quad\quad\quad  +\frac{\sum_{l_2'\neq m_2}2^{-n(R_1+R_2)}P_{\boldsymbol{Z}\boldsymbol{S}|\boldsymbol{U}\boldsymbol{U}_1\boldsymbol{U}_2}(\boldsymbol{z},\boldsymbol{s}|\boldsymbol{u}(m),\boldsymbol{u}_1(m,m_1),\boldsymbol{U}_2(m,l_2'))}{Q_{\boldsymbol{Z}\boldsymbol{S}|\boldsymbol{U}}(\boldsymbol{z},\boldsymbol{s}|\boldsymbol{u}(m))} \\
  &\quad\quad\quad\quad\quad\quad\quad\quad\quad\quad \left. +\frac{\sum_{l_1'\neq m_1,l_2'\neq m_2}2^{-n(R_1+R_2)}P_{\boldsymbol{Z}\boldsymbol{S}|\boldsymbol{U}\boldsymbol{U}_1\boldsymbol{U}_2}(\boldsymbol{z},\boldsymbol{s}|\boldsymbol{u}(m),\boldsymbol{U}_1(m,l_1'),\boldsymbol{U}_2(m,l_2'))}{Q_{\boldsymbol{Z}\boldsymbol{S}|\boldsymbol{U}}(\boldsymbol{z},\boldsymbol{s}|\boldsymbol{u}(m))} \right]\\
\end{align*}
\begin{align*}
  &\overset{(b)}{\leq} 2^{-n(R_1+R_2)}\sum_{\boldsymbol{z},\boldsymbol{s}}\sum_{m}\sum_{m_1,m_2}\sum_{\boldsymbol{u}(m)}\sum_{\boldsymbol{u}_1(m,m_1)}\sum_{\boldsymbol{u}_2(m,m_2)}\bar{P}(m)P_{\boldsymbol{Z}\boldsymbol{S}\boldsymbol{U}\boldsymbol{U}_1\boldsymbol{U}_2}(\boldsymbol{z},\boldsymbol{s},\boldsymbol{u}(m),\boldsymbol{u}_1(m,m_1),\boldsymbol{u}_2(m,m_2))\\
  &\quad\quad\quad\quad\quad \log \mathbb{E}_{\backslash(m,m_1,m_2)} \left[\frac{2^{-n(R_1+R_2)}P_{\boldsymbol{Z}\boldsymbol{S}|\boldsymbol{U}\boldsymbol{U}_1\boldsymbol{U}_2}(\boldsymbol{z},\boldsymbol{s}|\boldsymbol{u}(m),\boldsymbol{u}_1(m,m_1),\boldsymbol{u}_2(m,m_2))}{Q_{\boldsymbol{Z}\boldsymbol{S}|\boldsymbol{U}}(\boldsymbol{z},\boldsymbol{s}|\boldsymbol{u}(m))}\right.\\
  &\quad\quad\quad\quad\quad\quad\quad\quad\quad\quad +\frac{\sum_{l_1'\neq m_1}2^{-n(R_1+R_2)}P_{\boldsymbol{Z}\boldsymbol{S}|\boldsymbol{U}\boldsymbol{U}_1\boldsymbol{U}_2}(\boldsymbol{z},\boldsymbol{s}|\boldsymbol{u}(m),\boldsymbol{U}_1(m,l_1'),\boldsymbol{u}_2(m,m_2))}{Q_{\boldsymbol{Z}\boldsymbol{S}|\boldsymbol{U}}(\boldsymbol{z},\boldsymbol{s}|\boldsymbol{u}(m))}\\
  &\quad\quad\quad\quad\quad\quad\quad\quad\quad\quad  +\frac{\sum_{l_2'\neq m_2}2^{-n(R_1+R_2)}P_{\boldsymbol{Z}\boldsymbol{S}|\boldsymbol{U}\boldsymbol{U}_1\boldsymbol{U}_2}(\boldsymbol{z},\boldsymbol{s}|\boldsymbol{u}(m),\boldsymbol{u}_1(m,m_1),\boldsymbol{U}_2(m,l_2'))}{Q_{\boldsymbol{Z}\boldsymbol{S}|\boldsymbol{U}}(\boldsymbol{z},\boldsymbol{s}|\boldsymbol{u}(m))} \\
  &\quad\quad\quad\quad\quad\quad\quad\quad\quad\quad \left. +\frac{\sum_{l_1'\neq m_1,l_2'\neq m_2}2^{-n(R_1+R_2)}P_{\boldsymbol{Z}\boldsymbol{S}|\boldsymbol{U}\boldsymbol{U}_1\boldsymbol{U}_2}(\boldsymbol{z},\boldsymbol{s}|\boldsymbol{u}(m),\boldsymbol{U}_1(m,l_1'),\boldsymbol{U}_2(m,l_2'))}{Q_{\boldsymbol{Z}\boldsymbol{S}|\boldsymbol{U}}(\boldsymbol{z},\boldsymbol{s}|\boldsymbol{u}(m))} \right]\\
  &\leq 2^{-n(R_1+R_2)}\sum_{\boldsymbol{z},\boldsymbol{s}}\sum_{m}\bar{P}(m)\sum_{m_1,m_2}\sum_{\boldsymbol{u}(m)}\sum_{\boldsymbol{u}_1(m,m_1)}\sum_{\boldsymbol{u}_2(m,m_2)}P_{\boldsymbol{Z}\boldsymbol{S}\boldsymbol{U}_1\boldsymbol{U}_2\boldsymbol{U}}(\boldsymbol{z},\boldsymbol{s},\boldsymbol{u}_1(m,m_1),\boldsymbol{u}_2(m,m_2),\boldsymbol{u}(m))\\
  &\quad\quad\quad\quad\quad \log  \left[\frac{2^{-n(R_1+R_2)}P_{\boldsymbol{Z}\boldsymbol{S}|\boldsymbol{U}\boldsymbol{U}_1\boldsymbol{U}_2}(\boldsymbol{z},\boldsymbol{s}|\boldsymbol{u}(m),\boldsymbol{u}_1(m,m_1),\boldsymbol{u}_2(m,m_2))}{Q_{\boldsymbol{Z}\boldsymbol{S}|\boldsymbol{U}}(\boldsymbol{z},\boldsymbol{s}|\boldsymbol{u}(m))}\right.\\
  &\quad\quad\quad\quad\quad\quad\quad\quad\quad\quad +\frac{2^{-nR_2}P_{\boldsymbol{Z}\boldsymbol{S}|\boldsymbol{U}\boldsymbol{U}_2}(\boldsymbol{z},\boldsymbol{s}|\boldsymbol{u}(m),\boldsymbol{u}_2(m,m_2))}{Q_{\boldsymbol{Z}\boldsymbol{S}|\boldsymbol{U}}(\boldsymbol{z},\boldsymbol{s}|\boldsymbol{u}(m))}\\
  &\quad\quad\quad\quad\quad\quad\quad\quad\quad\quad \left. +\frac{2^{-nR_1}P_{\boldsymbol{Z}\boldsymbol{S}|\boldsymbol{U}\boldsymbol{U}_1}(\boldsymbol{z},\boldsymbol{s}|\boldsymbol{u}(m),\boldsymbol{u}_1(m,m_1))}{Q_{\boldsymbol{Z}\boldsymbol{S}|\boldsymbol{U}}(\boldsymbol{z},\boldsymbol{s}|\boldsymbol{u}(m))} +1 \right]\\
  &=\Psi_1 + \Psi_2,
\end{align*}
where $\bar{P}$ in $(a)$ is some distribution over the set $[1:|\textbf{C}|]$, $(b)$ follows by Jenson's inequality,
\begin{align*}
  \Psi_1 &=  2^{-n(R_1+R_2)}\sum_{m}\bar{P}(m)\sum_{m_1,m_2}\sum_{\substack{(\boldsymbol{u}(m),\boldsymbol{u}_1(m,m_1),\boldsymbol{u}_2(m,m_2),\boldsymbol{z},\boldsymbol{s})\\\notin T^n_{P_{UU_1U_2SZ},\delta}}}P_{\boldsymbol{Z}\boldsymbol{S}\boldsymbol{U}_1\boldsymbol{U}_2\boldsymbol{U}}(\boldsymbol{z},\boldsymbol{s},\boldsymbol{u}_1(m,m_1),\boldsymbol{u}_2(m,m_2),\boldsymbol{u}(m))\\
  &\quad\quad\quad\quad\quad \log  \left[\frac{2^{-n(R_1+R_2)}P_{\boldsymbol{Z}\boldsymbol{S}|\boldsymbol{U}\boldsymbol{U}_1\boldsymbol{U}_2}(\boldsymbol{z},\boldsymbol{s}|\boldsymbol{u}(m),\boldsymbol{u}_1(m,m_1),\boldsymbol{u}_2(m,m_2))}{Q_{\boldsymbol{Z}\boldsymbol{S}|\boldsymbol{U}}(\boldsymbol{z},\boldsymbol{s}|\boldsymbol{u}(m))}\right.\\
  &\quad\quad\quad\quad\quad\quad\quad\quad\quad\quad +\frac{2^{-nR_2}P_{\boldsymbol{Z}\boldsymbol{S}|\boldsymbol{U}\boldsymbol{U}_1\boldsymbol{U}_2}(\boldsymbol{z},\boldsymbol{s}|\boldsymbol{u}(m),\boldsymbol{u}_2(m,m_2))}{Q_{\boldsymbol{Z}\boldsymbol{S}|\boldsymbol{U}}(\boldsymbol{z},\boldsymbol{s}|\boldsymbol{u}(m))}\\
  &\quad\quad\quad\quad\quad\quad\quad\quad\quad\quad \left. +\frac{2^{-nR_1}P_{\boldsymbol{Z}\boldsymbol{S}|\boldsymbol{U}\boldsymbol{U}_1\boldsymbol{U}_2}(\boldsymbol{z},\boldsymbol{s}|\boldsymbol{u}(m),\boldsymbol{u}_1(m,m_1))}{Q_{\boldsymbol{Z}\boldsymbol{S}|\boldsymbol{U}}(\boldsymbol{z},\boldsymbol{s}|\boldsymbol{u}(m))} +1 \right]
\end{align*}
\begin{align*}
  \Psi_2 &=  2^{-n(R_1+R_2)}\sum_{m}\bar{P}(m)\sum_{m_1,m_2}\sum_{\substack{(\boldsymbol{u}(m),\boldsymbol{u}_1(m_1),\boldsymbol{u}_2(m_2),\boldsymbol{z},\boldsymbol{s})\\\in T^n_{P_{UU_1U_2SZ},\delta}}}P_{\boldsymbol{Z}\boldsymbol{S}\boldsymbol{U}_1\boldsymbol{U}_2\boldsymbol{U}}(\boldsymbol{z},\boldsymbol{s},\boldsymbol{u}_1(m,m_1),\boldsymbol{u}_2(m,m_2),\boldsymbol{u}(m))\\
  &\quad\quad\quad\quad\quad \log  \left[\frac{2^{-n(R_1+R_2)}P_{\boldsymbol{Z}\boldsymbol{S}|\boldsymbol{U}\boldsymbol{U}_1\boldsymbol{U}_2}(\boldsymbol{z},\boldsymbol{s}|\boldsymbol{u}(m),\boldsymbol{u}_1(m,m_1),\boldsymbol{u}_2(m,m_2))}{Q_{\boldsymbol{Z}\boldsymbol{S}|\boldsymbol{U}}(\boldsymbol{z},\boldsymbol{s}|\boldsymbol{u}(m))}\right.\\
  &\quad\quad\quad\quad\quad\quad\quad\quad\quad\quad +\frac{2^{-nR_2}P_{\boldsymbol{Z}\boldsymbol{S}|\boldsymbol{U}\boldsymbol{U}_1\boldsymbol{U}_2}(\boldsymbol{z},\boldsymbol{s}|\boldsymbol{u}(m),\boldsymbol{u}_2(m,m_2))}{Q_{\boldsymbol{Z}\boldsymbol{S}|\boldsymbol{U}}(\boldsymbol{z},\boldsymbol{s}|\boldsymbol{u}(m))}\\
  &\quad\quad\quad\quad\quad\quad\quad\quad\quad\quad \left. +\frac{2^{-nR_1}P_{\boldsymbol{Z}\boldsymbol{S}|\boldsymbol{U}\boldsymbol{U}_1\boldsymbol{U}_2}(\boldsymbol{z},\boldsymbol{s}|\boldsymbol{u}(m),\boldsymbol{u}_1(m,m_1))}{Q_{\boldsymbol{Z}\boldsymbol{S}|\boldsymbol{U}}(\boldsymbol{z},\boldsymbol{s}|\boldsymbol{u}(m))} +1 \right].
\end{align*}
By \cite[Theorem 1.1]{kramer2008topics}
\begin{align*}
  \sum_{\substack{(\boldsymbol{u}(m),\boldsymbol{u}_1(m,m_1),\boldsymbol{u}_2(m,m_2),\boldsymbol{z},\boldsymbol{s})\\\notin T^n_{P_{UU_1U_2SZ},\delta}}}P_{\boldsymbol{Z}\boldsymbol{S}\boldsymbol{U}_1\boldsymbol{U}_2\boldsymbol{U}}(\boldsymbol{z},\boldsymbol{s},\boldsymbol{u}_1(m,m_1),\boldsymbol{u}_2(m,m_2),\boldsymbol{u}(m))\leq 2|\mathcal{U}||\mathcal{U}_1||\mathcal{U}_2||\mathcal{S}||\mathcal{Z}|e^{-2n\delta_1^2\mu_{UU_1U_2SZ}}
\end{align*}
for some $\delta_1>0$, where $\mu_{UU_1U_2SZ}=\min_{(u,u_1,u_2,s,z)\in supp(P_{UU_1U_2SZ})}P_{UU_1U_2SZ}(u,u_1,u_2,s,z)$. Hence, $$\Psi_1\leq 2|\mathcal{U}||\mathcal{U}_1||\mathcal{U}_2||\mathcal{S}||\mathcal{Z}|e^{-2n\delta_1^2\mu_{UU_1U_2SZ}}\log\left( \frac{3}{\mu_{ZS|U}} + 1 \right),$$
where $\mu_{ZS|U}=\min_{(u,s,z)\in supp(P_{ZS|U})}P_{ZS|U}(z,s|u)$.
We further have
\begin{align*}
  \Psi_2 &\leq \log  \left[\frac{2^{-n(R_1+R_2)}P_{\boldsymbol{Z}\boldsymbol{S}|\boldsymbol{U}\boldsymbol{U}_1\boldsymbol{U}_2}(\boldsymbol{z},\boldsymbol{s}|\boldsymbol{u}(m),\boldsymbol{u}_1(m,m_1),\boldsymbol{u}_2(m,m_2))}{Q_{\boldsymbol{Z}\boldsymbol{S}|\boldsymbol{U}}(\boldsymbol{z},\boldsymbol{s}|\boldsymbol{u}(m))}\right.\\
  &\quad\quad\quad\quad\quad +\frac{2^{-nR_2}P_{\boldsymbol{Z}\boldsymbol{S}|\boldsymbol{U}\boldsymbol{U}_1\boldsymbol{U}_2}(\boldsymbol{z},\boldsymbol{s}|\boldsymbol{u}(m),\boldsymbol{u}_2(m,m_2))}{Q_{\boldsymbol{Z}\boldsymbol{S}|\boldsymbol{U}}(\boldsymbol{z},\boldsymbol{s}|\boldsymbol{u}(m))}\\
  &\quad\quad\quad\quad\quad \left. +\frac{2^{-nR_1}P_{\boldsymbol{Z}\boldsymbol{S}|\boldsymbol{U}\boldsymbol{U}_1\boldsymbol{U}_2}(\boldsymbol{z},\boldsymbol{s}|\boldsymbol{u}(m),\boldsymbol{u}_1(m,m_1))}{Q_{\boldsymbol{Z}\boldsymbol{S}|\boldsymbol{U}}(\boldsymbol{z},\boldsymbol{s}|\boldsymbol{u}(m))} +1 \right]\\
  &\leq \log  \left[2^{-n(R_1+R_2)}2^{n(I(U_1,U_2;Z,S|U)+\epsilon)}\right.\\
  &\quad\quad\quad\quad\quad +2^{-nR_2} 2^{n(I(U_2;Z,S|U)+\epsilon)}\\
  &\quad\quad\quad\quad\quad \left. +2^{-nR_1} 2^{n(I(U_1;Z,S|U)+\epsilon)} +1 \right].
\end{align*}
Combining $\Psi_1$ and $\Psi_2$ we have $\mathbb{E}_{\textbf{C}\textbf{C}_1\textbf{C}_2}\left[ D(P_{\boldsymbol{Z}\boldsymbol{S}|\boldsymbol{U}\textbf{C}\textbf{C}_1\textbf{C}_2}||Q_{\boldsymbol{Z}\boldsymbol{S}|\boldsymbol{U}}) \right]\to 0$ exponentially fast.
\section{Proof of lemma \ref{lem: secret key}}\label{app: proof of secret key}
In this section, we prove Lemma \ref{lem: secret key}. The proof is divided into two steps. In Step 1, we prove there exists an equal partition on $\mathcal{C}_\mathcal{V}=\{\mathcal{C}_\mathcal{V}(k_0)\}_{k_0=1}^k$, as required by the Wyner-Ziv Theorem, such that $I(K_0;\boldsymbol{Z})\leq\epsilon$. In Step 2, based on the partition constructed in Step 1, we further construct a secret key mapping such that the constructed secret key $K_1$ satisfies
\begin{align*}
  \mathbb{S}(K_1|\boldsymbol{Z},K_0) \leq \epsilon.
\end{align*}

\subsubsection*{Step 1}
The constructions of the partition and the secret key are based on the following extractor lemma by Csisz\'ar and K\"orner\cite{csiszar2011information}.
\begin{lemma}[Lemma 17.3 in \cite{csiszar2011information}]\label{lem:extractor}
  Let $P$ be a distribution on a finite set $\mathcal{V}$ and $\mathcal{F}$ be a subset of $\mathcal{V}$ such that $\mathcal{F}=\{v\in\mathcal{V}:  P(v)\leq 1/d\}$.
For some positive number $\epsilon$, if $P(\mathcal{F})\geq 1- \varepsilon$, then there exists a randomly selected mapping $G: \mathcal{V}\to\{1,\dots,k\}$ satisfies
\begin{align}
  \label{ineq:extractor}\sum_{m=1}^{k}\left| P(G^{-1}(m))-\frac{1}{k} \right| \leq 3\varepsilon
 \end{align}
with probability at least $1-2ke^{-\varepsilon^2(1-\varepsilon)d/2k(1+\varepsilon)}$.

Moreover, if each $P$ in a family $\mathcal{P}$ satisfies the hypothesis, then the probability that formula \eqref{ineq:extractor} holds for all $P\in\mathcal{P}$ is at least $1-2k|\mathcal{P}|e^{-\varepsilon^2(1-\varepsilon)d/2k(1+\varepsilon)}$.

Thus, the desired mapping exists if $k\log k < \frac{\varepsilon^2(1-\varepsilon)d\log e}{2(1+\varepsilon)\log 2|\mathcal{P}|}$. This realization of $G$ is denoted by $g$.
\end{lemma}
The proof in this step uses the technique in the proof of Lemma 17.5 and Theorem 17.21 in \cite{csiszar2011information} where the key idea is to construct set $\mathcal{F}$ and distribution family $\mathcal{P}$ satisfying conditions in Lemma \ref{lem:extractor}. We set parameters as follows. For some small positive $\tau$,
\begin{equation}
  \begin{split}
    d = I(V;S) - I(V;Y),\\
    k = I(V;S) - I(V;Y) - \tau \\
  \mathcal{P} = P_{V^n} \bigcup \{P^n_{V|\boldsymbol{z}}:\boldsymbol{z}\notin \mathcal{E} \},
  \end{split}
\end{equation}
where $\mathcal{E}$ is a subset of $\mathcal{Z}^n$ with exponentially small probability and will be defined later, and $P_{V^n}$ is a uniform distribution on the codebook $\mathcal{C}_{\mathcal{V}}$.
Let $\epsilon,\sigma,\delta,\zeta$ be positive real numbers such that $\zeta<\delta<\sigma$. Define set
\begin{align*}
  \mathcal{T}_1 = \{\boldsymbol{s}: T^n_{P_{SV},\delta}[\boldsymbol{s}]\neq\emptyset\}
\end{align*}
and let $f$ be a function on $\mathcal{T}_1$ such that $(f(\boldsymbol{S}),\boldsymbol{S})\in T_{P_{SV},\delta}, f(\boldsymbol{S})\in\mathcal{C}_{\mathcal{V}}$. Further extend $f$ to function on $\mathcal{S}^n$ by setting $f(\boldsymbol{s})$ to some fixed sequence in $\mathcal{V}^n$ for $\boldsymbol{s}\notin\mathcal{T}_1$. By Lemma \ref{lem: typical sequences independent sequences}, such a function $f$ always exists by setting the size of $|\mathcal{C}_{\mathcal{V}}|=2^{n(I(V;S)+\tau)}$. Define set
\begin{align*}
  \mathcal{T}_2 = \{(\boldsymbol{s},\boldsymbol{z}):\boldsymbol{s}\in \mathcal{T}_1,(f(\boldsymbol{s}),\boldsymbol{s},\boldsymbol{z})\in T^n_{P_{VSZ},\sigma}\}
\end{align*}
and let $\chi$ be the indicator function of set $\mathcal{T}_2$. Then, the joint distribution of $(\boldsymbol{v},\boldsymbol{s},\mu)$ is given by 
\begin{align}
  &P(\boldsymbol{v},\boldsymbol{z},\mu) = Pr\{f(\boldsymbol{S})=\boldsymbol{v},\boldsymbol{Z}=\boldsymbol{z},\chi(\boldsymbol{S},\boldsymbol{Z})=\mu\}\notag\\
  \label{neq: secret key lemma proof 1}&=\sum_{\boldsymbol{s}:f(\boldsymbol{s})=\boldsymbol{v},\chi(\boldsymbol{s},\boldsymbol{z})=\mu}P^n_{SZ}(\boldsymbol{s},\boldsymbol{z})
\end{align}
Define set $\mathcal{B}=\{(\boldsymbol{v},\boldsymbol{z},1):\boldsymbol{v}\in\mathcal{C}_{\mathcal{V}},\boldsymbol{z}\in T^n_{P_Z,\zeta},T^n_{P_{VSZ},\sigma}[\boldsymbol{v},\boldsymbol{z}]\neq\emptyset\}$. It follows that $P(\mathcal{B})\geq P^n_{SZ}(\mathcal{T}_2) - P^n_{Z}((T^n_{P_Z,\zeta})^c)\geq 1 - \eta^2$ for some exponentially small number $\eta$ and 
\begin{align*}
  |\mathcal{B}| &\overset{(a)}{\leq} \sum_{\boldsymbol{z}\in T^n_{P_Z,\zeta}} \left|\{\boldsymbol{v}:\boldsymbol{v}\in T^n_{P_{VZ},\sigma|\mathcal{S}|}[\boldsymbol{z}]\}\right|\\
  &\leq 2^{n(H(Z)+\epsilon)}2^{n(I(V;S) + \tau -I(V;Z) + \epsilon)},
\end{align*}
where $(a)$ follows by Lemma \ref{lem: typicality}.
For any $(\boldsymbol{v},\boldsymbol{z},1)\in\mathcal{B}$, by \eqref{neq: secret key lemma proof 1} and the definition of $\mathcal{T}_2$ we have
\begin{align*}
  P(\boldsymbol{v},\boldsymbol{z},1) &\leq \sum_{\boldsymbol{s}\in T^n_{P_{VSZ},\sigma}[\boldsymbol{v},\boldsymbol{z}]} P^n_{SZ}(\boldsymbol{s},\boldsymbol{z})\\
  &\leq 2^{n(H(S|VZ)+\epsilon)}2^{-n(H(SZ)-\epsilon)}<\frac{1}{\alpha|\mathcal{B}|},
\end{align*}
where $\alpha=2^{-n(5\epsilon + \tau)}$. Define $\mathcal{B}_{\boldsymbol{z},1}:= \{\boldsymbol{v} : (\boldsymbol{v},\boldsymbol{z},1)\in\mathcal{B}\}$. By Lemma \ref{lem: typicality 2}, $\boldsymbol{z}\in T^n_{P_Z,\zeta}$ implies $T^n_{P_{VZ},\delta}[\boldsymbol{z}]\neq \emptyset$, and $\boldsymbol{v}\in T^n_{P_{VZ},\delta}[\boldsymbol{z}]$ is a sufficient condition of $T^n_{P_{VSZ},\sigma}[\boldsymbol{v},\boldsymbol{z}]\neq\emptyset$. Thus, the size of $\mathcal{B}_{\boldsymbol{z},1}$ can be lower bounded by
\begin{align*}
  \left| \mathcal{B}_{\boldsymbol{z},1} \right| \geq \left| \{\boldsymbol{v}: \boldsymbol{v}\in T^n_{P_{VZ},\delta}[\boldsymbol{z}],\boldsymbol{v}\in\mathcal{C}_{\mathcal{V}}\} \right| \geq 2^{n(I(V;S)-I(V;Z)+\tau-\epsilon)}.
\end{align*}
Now define $\mathcal{D}:= \{\boldsymbol{z}: P_{Z}^n(\boldsymbol{z})\geq \frac{\alpha^2|\mathcal{B}_{\boldsymbol{z},1}|}{|\mathcal{B}|}\}$ and $\mathcal{B}':= \mathcal{B}\bigcap \{\mathcal{C}_{\mathcal{V}} \times \mathcal{D}\}$.  By our setting that $I(V;Y)-I(V;Z)>\tau_1>2\tau+16\epsilon>0$, for any $(\boldsymbol{v},\boldsymbol{z},1)\in \mathcal{B}'$ with $1$ being the value of the indicator function $\chi$,
\begin{align*}
  P_{\boldsymbol{V}|\boldsymbol{Z},1}(\boldsymbol{v}|\boldsymbol{z},1)\leq \frac{P(\boldsymbol{v},\boldsymbol{z},1)}{P_{\boldsymbol{Z},1}}(\boldsymbol{z},1)\leq \frac{1}{\alpha^3\min|\mathcal{B}_{\boldsymbol{z},1}|} < \frac{1}{d}
\end{align*}
and
\begin{align*}
  P_{\boldsymbol{V}\boldsymbol{Z}1}(\mathcal{B}')\geq P_{\boldsymbol{V}\boldsymbol{Z}1}(\mathcal{B})-P_{\boldsymbol{Z}1}(\mathcal{D}^c) \geq P_{\boldsymbol{V}\boldsymbol{Z}1}(\mathcal{B})-\alpha^2 \geq 1-\eta_1^2
\end{align*}
for exponentially small number $\eta_1$. Now for each $\boldsymbol{z}$, denote $\mathcal{B}_{\boldsymbol{z},1}'$ the set $\{\boldsymbol{v}:(\boldsymbol{v},\boldsymbol{z},1)\in \mathcal{B}'\}$ and $\mathcal{E}:=\{(\boldsymbol{z},1): P_{\boldsymbol{V}|\boldsymbol{Z},1}(\mathcal{B}_{\boldsymbol{z},1}'|\boldsymbol{z},1)\leq 1-\eta_1\}$, we have $P_{\boldsymbol{Z},1}(\mathcal{E})<\eta_1$. For each $\boldsymbol{z}\notin\mathcal{E}$, the distribution $P_{\boldsymbol{V}|\boldsymbol{Z},1}(\cdot|\boldsymbol{z},1)$ satisfies the condition in Lemma \ref{lem:extractor} with set $\mathcal{F}=\mathcal{B}_{\boldsymbol{z},1}'(\mathcal{F} \text{is defined in Lemma \ref{lem:extractor}})$ and hence there exists a mapping satisfying \eqref{ineq:extractor} for all $P_{\boldsymbol{V}|\boldsymbol{Z},1}(\cdot|\boldsymbol{z},1),(\boldsymbol{z},1)\notin \mathcal{E}$.

For the uniform distribution, we define the set $\mathcal{F}=\mathcal{C}_{\mathcal{V}}$. Since $P_{\boldsymbol{V}}(\boldsymbol{v})=\frac{1}{|\mathcal{C}_{\mathcal{V}|}}<\frac{1}{d}$ for any $\boldsymbol{v}\in\mathcal{C}_{\mathcal{V}}$ and $P_{\boldsymbol{V}}(\mathcal{C}_{\mathcal{V}})=1$, the conditions in Lemma \ref{lem:extractor} are also satisfied. The constructed mapping satisfies
\begin{align}
  \label{ineq: partition 1}&\sum_{m=1}^{k}\left| P_{\boldsymbol{V}}(g^{-1}(m))-\frac{1}{k} \right| \leq 3\varepsilon,\\
  \label{ineq: partition 2}&\sum_{m=1}^{k}\left| P_{\boldsymbol{V}|\boldsymbol{Z},1}(g^{-1}(m)|\boldsymbol{z},1)-\frac{1}{k} \right| \leq 3\varepsilon \;\; \text{for $(\boldsymbol{z},1)\notin\mathcal{E}$}
\end{align}
for some exponentially small number $\varepsilon$. By \eqref{ineq: partition 2} and the definition of the secure index $\mathbb{S}$, it follows that
\begin{align}
  \label{ineq: secret result of wyner-ziv index}I(g(\boldsymbol{V});\boldsymbol{Z})\leq\mathbb{S}(g(\boldsymbol{V})|\boldsymbol{Z}) \leq \varepsilon'
\end{align}
for some exponentially small number $\varepsilon'$.
The partition on the codebook arises from the mapping $g$. A codeword $\boldsymbol{v}$ belongs to bin $k$ if $g(\boldsymbol{v})=k$. Notice here by \eqref{ineq: partition 1}, the partition is not necessarily an equi-partition. Now let $g(\boldsymbol{V})$ be a random variable on $\mathcal{C}_{\mathcal{V}}$ following a nearly uniform distribution defined by the mapping $g$.

\begin{lemma}{\cite[Lemma 4]{he2016strong}}
  For any given codebook $\mathcal{C}$, if the function $g: \mathcal{C} \to [1:k]$ satisfies (\ref{ineq: partition 1}), there exists a partition  $\{\mathcal{C}_{m}\}_{m=1}^{k}$on $\mathcal{C}$ such that
  \begin{enumerate}
    \item \(|\mathcal{C}_{m}| = \frac{|\mathcal{C}|}{k}\) for all \(m \in [1: k]\), 
    \item \(H(K_0|g(\boldsymbol{V})) < 4\sqrt{\epsilon} \log{k}\),
  \end{enumerate}
where $K_0=g^{equal}(\boldsymbol{V})$ is the index of the bin containing $\boldsymbol{V}$, and $g^{equal}$ is the new mapping inducing the equal partition.
\end{lemma}
Now according to \eqref{ineq: secret result of wyner-ziv index} and the above lemma, we have
\begin{align*}
  I(K_0;\boldsymbol{Z}) &\leq I(K_0,g(\boldsymbol{V});\boldsymbol{Z})\\
  &=I(g(\boldsymbol{V});\boldsymbol{Z}) + I(K_0,\boldsymbol{Z}|g(\boldsymbol{V}))\\
  &=I(g(\boldsymbol{V});\boldsymbol{Z}) + H(K_0|g(\boldsymbol{V}))\\
  &\leq 3\varepsilon + 4\sqrt{\varepsilon} \log{k}.
\end{align*}
Note that $\varepsilon$ is an exponentially small number and hence, the information leakage is also exponentially small.

\subsubsection*{Step 2} The proof in Step 2 is almost the same as that in Step 1 and is in fact the direct part proof of Theorem 17.21 in \cite{csiszar2011information}. Here we replace the function $f$ in Step 1 by a pair of function $\phi_1 \times \phi_2: \mathcal{S}\to \mathcal{K}_0 \times \mathcal{K}_1$ on $\mathcal{T}_1$ such that $(\boldsymbol{V}(\phi_1(\boldsymbol{S}),\phi_2(\boldsymbol{S})),\boldsymbol{S})\in T^n_{P_{VS},\delta}$, where $\phi_1(\boldsymbol{S})=g^{equal}(f(\boldsymbol{S}))$ and $g^{equal}$ is the final mapping we construct in Step 1. The result of $\phi_2$ is the index of $\boldsymbol{V}$ in sub-bin $\mathcal{C}_{\mathcal{V}}(\phi_1(\boldsymbol{S}))$. Rewriting the set $\mathcal{B}$ as $\mathcal{B}=\{(k_0,k_1,\boldsymbol{z},1): k_0\in\mathcal{K}_0,k_1\in\mathcal{K}_1,\boldsymbol{z}\in T^n_{P_Z,\zeta},T^n_{P_{VSZ},\sigma}[\boldsymbol{v}(k_0,k_1),\boldsymbol{z}]\neq\emptyset\}$ and set $\mathcal{B}_{\boldsymbol{z},1}$ as $\mathcal{B}_{k_0,\boldsymbol{z},1}=\{k_1: (k_0,k_1,\boldsymbol{z},1)\in\mathcal{B}\}$, the remaining proof is the same as Step 1 and can be found in \cite[Proof of Theorem 17.21]{csiszar2011information}. The finally constructed mapping $\kappa$ satisfies
\begin{align*}
  \mathbb{S}(\kappa(\phi_2(\boldsymbol{S}))|\boldsymbol{Z},\phi_1(\boldsymbol{S}))\leq\epsilon
\end{align*}
and the proof is completed.

\section{proof of theorem \ref{the: outer bound}}\label{app: proof of outer bound}
In this section, we prove the outer bounds of the channel model. Starting with Sender 1 and Fano's inequality,
\begin{align*}
  nR_1 &= H(M_1)\\
  &=H(M_1|M_2)\\
  &\overset{(a)}{\leq} I(M_1;\boldsymbol{Y}|M_2) - I(M_1;\boldsymbol{Z}|M_2) + \delta\\
  &= \sum_{i=1}^n I(M_1;Y_i|M_2,Y_{i+1}^n) - I(M_1;Z_i|M_2,Z^{i-1}) + \delta\\
  &\overset{(b)}{=}\sum_{i=1}^n I(M_1,Z^{i-1};Y_i|M_2,Y_{i+1}^n) - I(M_1,Y_{i+1}^n;Z_i|M_2,Z^{i-1}) + \delta\\
  &\overset{(c)}{=}\sum_{i=1}^n I(M_1;Y_i|M_2,Y_{i+1}^n,Z^{i-1}) - I(M_1;Z_i|M_2,Y_{i+1}^n,Z^{i-1}) + \delta\\
  &\overset{(d)}{=}\sum_{i=1}^n I(U_{1i};Y_i|U_{2i},U_i) - I(U_{1i};Z_i|U_{2i},U_i) + \delta\\
  &\overset{(e)}{=} n(I(U_{1};Y|U_{2},U) - I(U_{1};Z|U_{2},U) + \delta),
\end{align*} 
where $(a)$ follows from the fact that $I(M_1;\boldsymbol{Z}|M_2)\leq I(M_1,M_2;\boldsymbol{Z})\leq \epsilon$, $(b)$ and $(c)$ follows by Csisz\'ar's sum identity, $(d)$ follows by setting $U_i=(Y_{i+1}^n,Z^{i-1}),U_{ji}=(M_j,U_i)$ for $j=1,2$, $(e)$ follows by introducing a time-sharing random variable $Q$ and setting $U=(U_Q,Q),U_1=U_{1Q},U_2=U_{2Q},Y=Y_Q,Z=Z_Q$. The bounds of $R_2$ and $R_1+R_2$ follow similarly and we have
\begin{align*}
  R_2 &\leq I(U_{2};Y|U_{1},U) - I(U_{2};Z|U_{1},U)\\
  R_1 + R_2 &\leq I(U_{1},U_{2};Y|U) - I(U_{1},U_{2};Z|U),
\end{align*}
with random variables $U,U_1,U_2,X_1,X_2,S,Y,Z$ satisfying a Markov chain relation $(U,U_1,U_2)-(X_1,X_2,S)-(Y,Z)$.

\section{Proof of Region $\mathcal{R}^{CSI-E}_{D,11}$ in theorem \ref{the: inner bound of degraded message sets}}\label{app: proof of degraded message sets}
In this section, we give the coding scheme for SD-MAWC with degraded message sets. Compared with the coding scheme in Section \ref{sec: coding scheme for R11}, secrecy of Sender 2 is not required and only Sender 1's codebook is partitioned as in \emph{Message Codebook Generation}, Section \ref{sec: coding scheme for R11}. The common message is encoded together with the Wyner-Ziv index, and the constructed secret key is only used to encrypt Sender 1's private message. Given $P_S$ and $P_{YZ|X_1X_2S}$, consider the positive real numbers $\widetilde{R}_1,R_0,R_{10},R_{11},R_{K_1}$ such that 
\begin{align*}
  &\widetilde{R}_1 \leq I(U_1;Y|U,U_2,V),\\
  &R_0 + \widetilde{R}_1 \leq I(V,U,U_1,U_2;Y)-I(V;S),\\
  &\widetilde{R}_1 - R_{10} > I(U_1;Z|U,U_2,S),\\
  &R_{K_1}=R_{11} \leq I(V;Y) - I(V;Z,U,U_2)
\end{align*}
under fixed joint distribution $$P_{S}P_{V|S}P_UP_{U_1|U}P_{U_2|U}P_{X_1|UU_1S}P_{X_2|UU_2S}P_{YZ|X_1X_2S}.$$

\emph{Key Message Codebook Generation: }  In each block $1\leq b \leq B$, the sender generates a codebook $\mathcal{C}_{K_b}=\{\boldsymbol{v}(l)\}_{l=1}^{2^{nR_{K}}}$ consisting of $2^{nR_K}$ codewords, each i.i.d. generated according to distribution $P_V$ such that $P_V(v)=\sum_{s\in\mathcal{S}}q(s)P_{V|S}(v|s)$ for any $v\in\mathcal{V}$ and $R_K=I(V;S)+\tau$. Partition the codebook $\mathcal{C}_{K_b}$ into $2^{nR_{K_0}}$ sub-codebooks $\mathcal{C}_{K_b}(k_{0,b})$, where $k_{0,b}\in[1:2^{nR_{K_0}}]$ and $R_{K_0}=I(V;S)-I(V;Y)+2\tau$. Let $\mathcal{T}$ be the index set of codewords in each subcodebook $\mathcal{C}_{K_b}(k_{0,b})$ such that $|\mathcal{T}|=|\mathcal{C}_{K_b}(k_{0,b})|$ for any $k_{0,b}\in[1:2^{nR_{K_0}}]$. For each codebook $\mathcal{C}_{K_b},$ construct a secret key mapping $\kappa:\mathcal{T}\to [1:2^{nR_{K_1}}].$ Denote the resulted secret key by $K_{1,b}$.

\emph{Common Message Codebook Generation: }  For each block $b$, generate common message codebook $\mathcal{C}_b=\{\boldsymbol{u}(m_A,m_0)\}_{m_A=1,m_0=1}^{2^{nR_A},2^{nR_0}}$ i.i.d. according to distribution $P_{U}$, where $R_A=I(V;S)-I(V;Y)+3\tau$, $m_A$ is used to convey the Wyner-Ziv index and $m_0$ is the common message.

\emph{Message Codebook Generation:} 
\begin{enumerate}
  \item Block $b, b\in[1:B]$. For each pair $(m_A,m_0)$, generate codebook $\mathcal{C}_{1,b}(m_A,m_0)=\{\boldsymbol{u}_1(m_A,m_0,l)\}_{l=1}^{2^{n\widetilde{R}_1}}$ containing $2^{n\widetilde{R}_1}$ codewords, each i.i.d. generated according to distribution $P^n_{U_1|U}(\cdot | \boldsymbol{u}(m_A,m_0))$. Partition each $\mathcal{C}_{1b}(m_A,m_0)$ into two-layer subcodebooks $\mathcal{C}_{1,b}(m_A,m_0,m_{10},m_{11})=\{\boldsymbol{u}_1(m_A,m_0,m_{10},m_{11},l_1)\}_{l_1=1}^{2^{nR_1'}}$ as in \emph{Message Codebook Generation}, Section \ref{sec: coding scheme for R11}, where $m_{10}\in[1:2^{nR_{10}}],m_{11}\in[1:2^{nR_{11}}],R_1'=\widetilde{R}_1 - R_{10} - R_{11}$. Generate codebook $\mathcal{C}_{2,b}=\{\boldsymbol{u}_2(m_A,m_0)\}_{m_A=1,m_0=1}^{2^{nR_A},2^{nR_0}}$ with each codeword generated by $P^n_{U_2|U}(\boldsymbol{u}_2(m_A,m_0)|\boldsymbol{u}(m_A,m_0))$. 
  \item Block $B+1$. For $k=1,2$, generate codebooks $\mathcal{C}_{k,B+1}$ as above with length $\widetilde{n}$ defined in \eqref{eq: codeword length of last block}.
\end{enumerate}


The above codebooks are all generated randomly and independently. Denote the set of random codebooks in each block $b$ by $\bar{\textbf{C}}_b$.

\emph{Encoding: } 
\begin{enumerate}
  \item Block 1. Setting $m_{A,1}=m_{0,1}=m_{10,1}=m_{11,1}=1$, Sender 1 selects $l_1\in[1:2^{nR_1'}]$ uniformly at random and Sender 2 selects codeword $\boldsymbol{u}_2(1,1)$. The codeword $\boldsymbol{x}_1$ is generated by $(\boldsymbol{u}(1,1),\boldsymbol{u}_1(1,1,1,1,l_1),\boldsymbol{s})$ according to $P^n_{X_1|UU_1S}$ and $\boldsymbol{x}_2$ is generated by $(\boldsymbol{u}(1,1),\boldsymbol{u}_2(1,1),\boldsymbol{s})$ according to $P^n_{X_2|UU_2S}$.
  \item Blocks $b\in[2:B]$. Upon observing the state sequence $\boldsymbol{s}_{b-1}$ in the last block, the encoders find a sequence $\boldsymbol{v}_{b-1}$ such that $(\boldsymbol{s}_{b-1},\boldsymbol{v}_{b-1})\in T^n_{P_{SV},\delta}$ and set $m_{A,b}=k_{0,b}$, where $k_{0,b}$ is the index of subcodebook $\mathcal{C}_{k_b}(k_{0,b})$ containing $\boldsymbol{v}_{b-1}$. We also write sequence $\boldsymbol{v}_{b-1}$ as $\boldsymbol{v}(k_{0,b},t_b)$ if $\boldsymbol{v}_{b-1}$ is the $t_b$ th sequence in sub-codebook $\mathcal{C}_{K_b}(k_{0,b})$. To transmit message $(m_{0,b},m_{1,b})$, the encoder $1$ splits $m_{1,b}$ into two independent parts $(m_{10,b},m_{11,b})$ and compute $c_{11,b} = m_{11,b} \oplus k_{1,b} \pmod{2^{nR_{11}}}$, where $k_{1,b}=\kappa(t_{b})$. The encoder selects $l_{1,b}\in[1:2^{nR_1'}]$ uniformly at random and generates the codeword $\boldsymbol{x}_1$ by
  \begin{align*}
    P^n_{X_1|UU_1S}(\boldsymbol{x}_1|\boldsymbol{u}(k_{0,b},m_{0,b}),\boldsymbol{u}_1(k_{0,b},m_{0,b},m_{10,b},c_{11,b},l_{1,b}),\boldsymbol{s}_b)=\prod_{i=1}^n P_{X_1|UU_1S}(x_{1i}|u_i,u_{1i},s_{i,b}).
  \end{align*}
  The codeword $\boldsymbol{x}_2$ for Sender 2 is generated by 
  \begin{align*}
    P^n_{X_2|UU_2S}(\boldsymbol{x}_2|\boldsymbol{u}(k_{0,b},m_{0,b}),\boldsymbol{u}_2(k_{0,b},m_{0,b}),\boldsymbol{s}_b)=\prod_{i=1}^n P_{X_2|UU_2S}(x_{2i}|u_i,u_{2i},s_{i,b}).
  \end{align*}
  \item Block $B+1$. Upon observing the state sequence $\boldsymbol{s}_{B}$ in the last block, the encoders find a sequence $\boldsymbol{v}(k_{0,B+1},t_{B+1})$ such that $(\boldsymbol{s}_{B},\boldsymbol{v}(k_{0,B+1},t_{B+1}))\in T^n_{P_{SV},\delta}$. The encoders then set $m_{0,B+1}=m_{10,B+1}=m_{11,B+1}=1$ and generate codeword $\boldsymbol{x}_1$ and $\boldsymbol{x}_2$ according to distributions
  \begin{align*}
    P^n_{X_1|UU_1S}(\boldsymbol{x}_1|\boldsymbol{u}(k_{0,B+1},1),\boldsymbol{u}_1(k_{0,B+1},1,1,1,1),\boldsymbol{s}_{B+1})=\prod_{i=1}^n P_{X_1|UU_1S}(x_{1i}|u_i,u_{1i},s_{i,B+1}),\\
    P^n_{X_2|UU_2S}(\boldsymbol{x}_2|\boldsymbol{u}(k_{0,B+1},1),\boldsymbol{u}_2(k_{0,B+1},1),\boldsymbol{s}_{B+1})=\prod_{i=1}^n P_{X_2|UU_2S}(x_{2i}|u_{i},u_{2i},s_{i,B+1}).
  \end{align*}
\end{enumerate}

\emph{Backward Decoding: } 
\begin{enumerate}
  \item Block $B+1$. The decoder tries to find a unique $\hat{k}_{0,B+1}$ such that $(\boldsymbol{u}_1(\hat{k}_{0,B+1},1,1,1,1),\boldsymbol{u}_2(\hat{k}_{0,B+1},1),\boldsymbol{y}_{B+1})\in T^n_{P_{U_1U_2Y},\delta}$ for some $\delta>0$.
  \item Blocks $b\in[1:B]$. The decoder has the knowledge about $\hat{k}_{0,b+1}$ from the last block. It tries to find a unique $\boldsymbol{v}_{b}=\boldsymbol{v}(\hat{k}_{0,b+1},\hat{t}_{b+1})$ such that $(\boldsymbol{v}(\hat{k}_{0,b+1},\hat{t}_{b+1}),\boldsymbol{y}_b)\in T^n_{P_{VY},\delta}$. Now the decoder computes $\hat{m}_{11,b+1}=\hat{c}_{11,b+1} \ominus \hat{k}_{1,b+1} \pmod{2^{nR_{11}}}$, where $\hat{k}_{1,b+1}=\kappa(\hat{t}_{b+1})$. 
  In Block $b, b\in[2:B]$, with the help of $\boldsymbol{v}(\hat{k}_{0,b+1},\hat{t}_{b+1})$, the decoder looks for a unique tuple $(\hat{k}_{0,b},\hat{m}_{0,b},\hat{m}_{10,b},\hat{c}_{11,b},\hat{l}_{1,b})$ such that 
  \begin{align*}
    (\boldsymbol{v}_b,\boldsymbol{u}(\hat{k}_{0,b},\hat{m}_{0,b}),\boldsymbol{u}_1(\hat{k}_{0,b},\hat{m}_{0,b},\hat{m}_{10,b},\hat{c}_{11,b},\hat{l}_{1,b}),\boldsymbol{u}_2(\hat{k}_{0,b},\hat{m}_{0,b}),\boldsymbol{y}_b)\in T^n_{P_{VUU_1U_2Y},\delta}.
  \end{align*}
  for some $\delta>0$.
  \item Block 1. The messages transmitted in Block 1 are dummy messages. Hence, the decoding of Block 1 need not be performed.
\end{enumerate}
The error analysis and information leakage analysis are almost the same as Section \ref{sec: coding scheme for R11} and is omitted here.

\section{proof of theorem \ref{the: degraded mac with one-side causal CSI}}\label{app: proof of the degraded mac with one-side causal CSI}
\subsection{Proof of Lower Bound}

By strong functional representation lemma\cite{li2018strong}, for any random variables $(X_1,S,U)$ with conditional distribution $P_{X_1S|U}$, one can construct a random variable $U_1$ such that $X_1$ can be specified by a deterministic function $x_1(u,u_1,s)$ and $U_1$ is independent of $S$ given $U$. Together with the fact that $S$ is independent of $U$, we have
\begin{align*}
  &I(U_1;Y|U,X_2) - I(U_1;Z|X_2,U,S) - H(S|U,U_1,X_2,Y) + H(S|Z,U,X_2)\\
  &\overset{(a)}{=}I(U_1;Y,S|U,X_2) - I(U_1;S|U,X_2,Y) - I(U_1;Z,S|X_2,U) - H(S|U,U_1,X_2,Y) + H(S|Z,U,X_2)\\
  &=I(U_1;Y|U,X_2,S) - I(U_1;Z|U,X_2,S) - H(S|U,X_2,Y)+ H(S|U,X_2,Z)\\
  &\overset{(b)}{=}I(U_1,X_1;Y|U,X_2,S) - I(U_1,X_1;Z|U,X_2,S) - H(S|U,X_2,Y)+ H(S|U,X_2,Z)\\
  &=I(X_1;Y|U,X_2,S) - I(X_1;Z|U,X_2,S) - H(S|U,X_2,Y)+ H(S|U,X_2,Z)\\
  &=I(X_1,S;Y|U,X_2) - I(X_1,S;Z|U,X_2).
\end{align*}
where $(a)$ follows by the fact that $S$ is independent of $(U,U_1,X_2)$, $(b)$ follows since $X_1$ is determined by $(U,U_1,S)$. For the second constraint on Sender 1,
\begin{align*}
  &I(U_1;Y|U,X_2)-H(S|U,U_1,X_2,Y)\\
  &=I(U_1;S,Y|U,X_2) - I(U_1;S|U,X_2,Y) - H(S|U,U_1,X_2,Y)\\
  &=I(U_1;Y|U,X_2,S)-H(S|U,X_2,Y)\\
  &\overset{(a)}{=}I(U_1,X_1;Y|U,X_2,S)-H(S|U,X_2,Y)\\
  &\overset{(b)}{=}I(X_1;Y|U,X_2,S)-H(S|U,X_2,Y)\\
  &=I(X_1,S;Y|U,X_2) - H(S|U,X_2)\\
  &\overset{(c)}{=}I(X_1,S;Y|U,X_2) - H(S)
\end{align*}
where $(a)$ follows since $X_1$ is determined by $(U,U_1,S)$, $(b)$ follows since $(U,U_1)-(X_1,X_2,S)-(Y,Z)$ form a Markov chain, $(c)$ follows by the independence between $S$ and $(U,X_2)$. Now we proceed to the sum rate. For the first constraint on the sum rate,
\begin{align*}
  &I(U,U_1,X_2;Y) - I(U_1;Z|X_2,U,S)-H(S|U,U_1,X_2,Y)+H(S|Z,U,X_2)\\
  &=I(U,U_1,X_2;S,Y) -I(U,U_1,X_2;S|Y) - I(U_1;Z|X_2,U,S)-H(S|U,U_1,X_2,Y)+H(S|Z,U,X_2)\\
  &=I(U,U_1,X_1,X_2;Y|S) - I(U_1,X_1;Z|X_2,U,S)-H(S|Y)+H(S|Z,U,X_2)\\
  &=I(X_1,X_2;Y|S) - I(X_1;Z|X_2,U,S)-H(S|Y)+H(S|Z,U,X_2)\\
  &\overset{(a)}{=}I(X_1,X_2,S;Y) - I(X_1,S;Z|U,X_2)
\end{align*}
where $(a)$ follows by the Markov chain $(U,U_1)-(X_1,X_2,S)-(Y,Z)$ and the fact that $S$ is independent of $U,X_2$. For the second achievable sum rate, it follows that
\begin{align*}
  &I(U,U_1,X_2;Y)-H(S|U,U_1,X_2,Y)\\
  &=I(U,U_1,X_2;S,Y)-I(U,U_1,X_2;S|Y)-H(S|U,U_1,X_2,Y)\\
  &=I(U,U_1,X_2;Y|S)-I(U,U_1,X_2;S|Y)-H(S|U,U_1,X_2,Y)\\
  &\overset{(a)}{=}I(U,U_1,X_1,X_2;Y|S)-I(U,U_1,X_2;S|Y)-H(S|U,U_1,X_2,Y)\\
  &=I(X_1,X_2;Y|S)-H(S|Y)=I(X_1,X_2,S;Y)-H(S).
\end{align*}
where $(a)$ follows since $X_1$ is determined by $(U,U_1,S)$.

\subsection{Proof of Upper Bound}
To prove the upper bound of Sender 1, by Fano's inequality and the fact that $\boldsymbol{X}_2$ is determined by $M_0$,
\begin{align}
  nR_1 &= H(M_1) \notag \\
  &=I(M_1;\boldsymbol{Y}|M_0,\boldsymbol{X}_2) - I(M_1;\boldsymbol{Z}|M_0,\boldsymbol{X}_2) + \delta \notag\\
  &\overset{(a)}{=}I(M_1;\boldsymbol{Y}|\boldsymbol{Z},M_0,\boldsymbol{X}_2) + \delta \notag\\
  &\leq I(M_1,\boldsymbol{X}_1,\boldsymbol{S};\boldsymbol{Y}|\boldsymbol{Z},M_0,\boldsymbol{X}_2) + \delta \notag\\
  &= I(\boldsymbol{X}_1,\boldsymbol{S};\boldsymbol{Y}|\boldsymbol{Z},M_0,\boldsymbol{X}_2) + I(M_1;\boldsymbol{Y}|\boldsymbol{Z},M_0,\boldsymbol{X}_2,\boldsymbol{X}_1,\boldsymbol{S}) + \delta \notag\\
  &\overset{(b)}{=} I(\boldsymbol{X}_1,\boldsymbol{S};\boldsymbol{Y}|\boldsymbol{Z},M_0,\boldsymbol{X}_2) + \delta \notag\\
  &=H(\boldsymbol{X}_1,\boldsymbol{S}|\boldsymbol{Z},M_0,\boldsymbol{X}_2) - H(\boldsymbol{X}_1,\boldsymbol{S}|\boldsymbol{Z},M_0,\boldsymbol{X}_2,\boldsymbol{Y}) + \delta \notag\\
  &\overset{(c)}{=}H(\boldsymbol{X}_1,\boldsymbol{S}|\boldsymbol{Z},M_0,\boldsymbol{X}_2) - H(\boldsymbol{X}_1,\boldsymbol{S}|M_0,\boldsymbol{X}_2,\boldsymbol{Y}) + \delta \notag\\
  \label{neq: R1 upper bound 1}&=I(\boldsymbol{X}_1,\boldsymbol{S};\boldsymbol{Y}|M_0,\boldsymbol{X}_2) - I(\boldsymbol{X}_1,\boldsymbol{S};\boldsymbol{Z}|M_0,\boldsymbol{X}_2) + \delta
\end{align}
where $(a),(c)$ follows by the degradedness of the wiretap channel, $(b)$ follows from the Markov chain $(M_1,M_0)-(\boldsymbol{X}_1,\boldsymbol{X}_2,\boldsymbol{S})-(\boldsymbol{Y},\boldsymbol{Z})$. We continue following the proof in \cite{liang2008multiple}. Setting
\begin{align*}
  U_i=(Y^{i-1},\boldsymbol{X}_2,M_0),
\end{align*}
it follows that $U_i$ is independent of $S_i$. Then, we have
\begin{align}
  &I(\boldsymbol{X}_1,\boldsymbol{S};\boldsymbol{Y}|M_0,\boldsymbol{X}_2) - I(\boldsymbol{X}_1,\boldsymbol{S};\boldsymbol{Z}|M_0,\boldsymbol{X}_2)\notag\\
  &= \sum_{i=1}^n I(\boldsymbol{X}_1,\boldsymbol{S};Y_i|Y^{i-1},M_0,\boldsymbol{X}_2) - I(\boldsymbol{X}_1,\boldsymbol{S};Z_i|Z^{i-1},M_0,\boldsymbol{X}_2)\notag\\
  &=\sum_{i=1}^n H(Y_i|Y^{i-1},M_0,\boldsymbol{X}_2) - H(Y_i|Y^{i-1},M_0,\boldsymbol{X}_2,\boldsymbol{X}_1,\boldsymbol{S})\notag\\
  &\quad\quad\quad\quad\quad - H(Z_i|Z^{i-1},M_0,\boldsymbol{X}_2) + H(Z_i|Z^{i-1},M_0,\boldsymbol{X}_2,\boldsymbol{X}_1,\boldsymbol{S})\notag\\
  &\overset{(a)}{\leq}\sum_{i=1}^n H(Y_i|U_i,X_{2,i}) - H(Y_i|U_i,X_{2,i},X_{1,i},S_i)\notag\\
  &\quad\quad\quad\quad\quad - H(Z_i|Y^{i-1},Z^{i-1},M_0,\boldsymbol{X}_2) + H(Z_i|Y^{i-1},M_0,\boldsymbol{X}_2,X_{1,i},S_{i})\notag\\
  &\overset{(b)}{=}\sum_{i=1}^n H(Y_i|U_i,X_{2,i}) - H(Y_i|U_i,X_{2,i},X_{1,i},S_i)\notag\\
  &\quad\quad\quad\quad\quad - H(Z_i|Y^{i-1},M_0,\boldsymbol{X}_2) + H(Z_i|Y^{i-1},M_0,\boldsymbol{X}_2,X_{1,i},S_{i})\notag\\
  \label{app: outer bound eq 1}&=\sum_{i=1}^n I(X_{1,i},S_i;Y_i|U_i,X_{2,i}) - I(X_{1,i},S_i;Z_i|U_i,X_{2,i})\\
  &=n(I(X_1,S;Y|U,X_2) - I(X_1,S;Z|U,X_2))\notag
\end{align}
where $(a)$ and $(b)$ follows by the Markov chain $(U_i,M_0)-(X_{1,i},X_{2,i},S_i)-(Y_i,Z_i)$ and the degradedness of the wiretap channel.  By Fano's inequality and $H(M_1|\boldsymbol{Y},M_0,\boldsymbol{X}_2,\boldsymbol{S})\leq H(M_1|\boldsymbol{Y})$, it follows that
\begin{align}
  nR_1 &= H(M_1) \notag\\
  &\leq I(M_1;\boldsymbol{Y}|M_0,\boldsymbol{X}_2,\boldsymbol{S}) + \delta\notag\\
  &= \sum_{i=1}^n I(M_1;Y_i|Y^{i-1},M_0,\boldsymbol{X}_2,\boldsymbol{S})+ \delta\notag\\
  &= \sum_{i=1}^n H(Y_i|Y^{i-1},M_0,\boldsymbol{X}_2,\boldsymbol{S}) - H(Y_i|Y^{i-1},M_0,M_1,\boldsymbol{X}_2,\boldsymbol{S})+ \delta\notag\\
  &\leq \sum_{i=1}^n H(Y_i|Y^{i-1},M_0,\boldsymbol{X}_2,X_{2,i},S_i) - H(Y_i|Y^{i-1},M_0,M_1,\boldsymbol{X}_1,\boldsymbol{X}_2,\boldsymbol{S})+ \delta\notag\\
  &\overset{(a)}{=} \sum_{i=1}^n H(Y_i|Y^{i-1},M_0,\boldsymbol{X}_2,X_{2,i},S_i) - H(Y_i|Y^{i-1},M_0,\boldsymbol{X}_2,X_{1,i},X_{2,i},S_i)+ \delta\notag\\
  \label{app: outer bound eq 2}&= \sum_{i=1}^n I(X_{1,i};Y_i|U_i,X_{2,i},S_i) + \delta\\
  &=n(I(X_1;Y|U,X_2,S)+\delta).\notag
\end{align}
where $(a)$ follows by the fact that conditions do not increase entropy and the Markov chain $(M_0,M_1)-(X_{1,i},X_{2,i},S_i)-(Y_i,Z_i)$. For the sum rate, we have
\begin{align*}
  n(R_0+R_1) &= H(M_0) + nR_1\\
  &\leq I(M_0;\boldsymbol{Y}) + nR_1 +\delta\\
  &\overset{(a)}{\leq} \sum_{i=1}^n I(M_0;Y_i|Y^{i-1}) + I(X_{1,i},S_i;Y_i|U_i,X_{2,i}) - I(X_{1,i},S_i;Z_i|U_i,X_{2,i}) + \delta \\
  &\leq \sum_{i=1}^n I(M_0,Y^{i-1},\boldsymbol{X}_2;Y_i)  + I(X_{1,i},S_i;Y_i|U_i,X_{2,i}) - I(X_{1,i},S_i;Z_i|U_i,X_{2,i}) + \delta \\
  &\overset{(b)}{=} \sum_{i=1}^n I(U_i,X_{1,i},X_{2,i},S_i;Y_i)- I(X_{1,i},S_i;Z_i|U_i,X_{2,i}) + \delta\\
  &\overset{(c)}{=}I(X_{1,i},X_{2,i},S_i;Y_i)- I(X_{1,i},S_i;Z_i|U_i,X_{2,i}) + \delta\\
  &=n(I(X_1,X_2,S;Y)-I(X_1,S;Z|U,X_2)+\delta).
\end{align*}
where $(a)$ follows by \eqref{app: outer bound eq 1}, $(b)$ follows by the definition of $U_i$, $(c)$ follows by the Markov chain $U_i-(X_{1,i},X_{2,i},S_i)-Y_i$,
and
\begin{align*}
  n(R_0+R_1)&= H(M_0) + nR_1\\
  &\leq I(M_0;\boldsymbol{Y}) + nR_1 +\delta\\
  &\overset{(a)}{\leq} \sum_{i=1}^n I(M_0;Y_i|Y^{i-1},S_i) + I(X_{1,i};Y_i|U_i,X_{2,i},S_i) + \delta\\
  &\leq \sum_{i=1}^n I(M_0,Y^{i-1},\boldsymbol{X}_2,X_{2,i};Y_i|S_i) + I(X_{1,i};Y_i|U_i,X_{2,i},S_i) + \delta\\
  &=\sum_{i=1}^n I(U_i,X_{1,i},X_{2,i};Y_i|S_i)+\delta\\
  &=\sum_{i=1}^n I(X_{1,i},X_{2,i};Y_i|S_i)+\delta = n(I(X_1,X_2;Y|S) +\delta)
\end{align*}
where $(a)$ follows by the independence of $S_i$ and $(M_0,Y^{i-1})$ and \eqref{app: outer bound eq 2}. The proof is completed.

\ifCLASSOPTIONcaptionsoff
  \newpage
\fi

\bibliographystyle{ieeetr} 
\bibliography{ref}

\begin{thebibliography}{10}

\bibitem{wyner1975wire}
A.~D. Wyner, ``The wire-tap channel,'' {\em Bell System Technical Journal},
  vol.~54, no.~8, pp.~1355--1387, 1975.

\bibitem{csiszar1978broadcast}
I.~Csisz{\'a}r and J.~K{\"o}rner, ``Broadcast channels with confidential
  messages,'' {\em IEEE Transactions on Information Theory}, vol.~24, no.~3,
  pp.~339--348, 1978.

\bibitem{liang2008multiple}
Y.~Liang and H.~V. Poor, ``Multiple-access channels with confidential
  messages,'' {\em IEEE Transactions on Information Theory}, vol.~54, no.~3,
  pp.~976--1002, 2008.

\bibitem{yassaee2010multiple}
M.~H. Yassaee and M.~R. Aref, ``Multiple access wiretap channels with strong
  secrecy,'' in {\em 2010 IEEE Information Theory Workshop}, pp.~1--5, IEEE,
  2010.

\bibitem{wiese2013strong}
M.~Wiese and H.~Boche, ``Strong secrecy for multiple access channels,'' in {\em
  Information Theory, Combinatorics, and Search Theory}, pp.~71--122, Springer,
  2013.

\bibitem{lai2008relay}
L.~Lai and H.~El~Gamal, ``The relay--eavesdropper channel: Cooperation for
  secrecy,'' {\em IEEE transactions on information theory}, vol.~54, no.~9,
  pp.~4005--4019, 2008.

\bibitem{oohama2007capacity}
Y.~Oohama, ``Capacity theorems for relay channels with confidential messages,''
  in {\em 2007 IEEE International Symposium on Information Theory},
  pp.~926--930, IEEE, 2007.

\bibitem{liu2008discrete}
R.~Liu, I.~Maric, P.~Spasojevic, and R.~D. Yates, ``Discrete memoryless
  interference and broadcast channels with confidential messages: Secrecy rate
  regions,'' {\em IEEE Transactions on Information Theory}, vol.~54, no.~6,
  pp.~2493--2507, 2008.

\bibitem{liang2009capacity}
Y.~Liang, A.~Somekh-Baruch, H.~V. Poor, S.~Shamai, and S.~Verd{\'u}, ``Capacity
  of cognitive interference channels with and without secrecy,'' {\em IEEE
  Transactions on Information Theory}, vol.~55, no.~2, pp.~604--619, 2009.

\bibitem{leung1978gaussian}
S.~Leung-Yan-Cheong and M.~Hellman, ``The gaussian wire-tap channel,'' {\em
  IEEE transactions on information theory}, vol.~24, no.~4, pp.~451--456, 1978.

\bibitem{shannon1958channels}
C.~E. Shannon, ``Channels with side information at the transmitter,'' {\em IBM
  Journal of Research and Development}, vol.~2, no.~4, pp.~289--293, 1958.

\bibitem{gel1980coding}
S.~Gel'fand and M.~Pinsker, ``Coding for channel with random parameters,'' {\em
  Probl. Control Inform. Theory}, vol.~9, no.~1, pp.~19--31, 1980.

\bibitem{lapidoth2012multiple}
A.~Lapidoth and Y.~Steinberg, ``The multiple-access channel with causal side
  information: Common state,'' {\em IEEE Transactions on Information Theory},
  vol.~59, no.~1, pp.~32--50, 2012.

\bibitem{lapidoth2012multiple2}
A.~Lapidoth and Y.~Steinberg, ``The multiple-access channel with causal side
  information: Double state,'' {\em IEEE Transactions on Information Theory},
  vol.~59, no.~3, pp.~1379--1393, 2012.

\bibitem{blackwell1959capacity}
D.~Blackwell, L.~Breiman, A.~Thomasian, {\em et~al.}, ``The capacity of a class
  of channels,'' {\em The Annals of Mathematical Statistics}, vol.~30, no.~4,
  pp.~1229--1241, 1959.

\bibitem{pereg2018arbitrarily}
U.~Pereg and Y.~Steinberg, ``The arbitrarily varying channel under constraints
  with side information at the encoder,'' {\em IEEE Transactions on Information
  Theory}, vol.~65, no.~2, pp.~861--887, 2018.

\bibitem{pereg2019arbitrarily}
U.~Pereg and Y.~Steinberg, ``The arbitrarily varying broadcast channel with
  causal side information at the encoder,'' {\em IEEE Transactions on
  Information Theory}, vol.~66, no.~2, pp.~757--779, 2019.

\bibitem{chen2008wiretap}
Y.~Chen and A.~H. Vinck, ``Wiretap channel with side information,'' {\em IEEE
  Transactions on Information Theory}, vol.~54, no.~1, pp.~395--402, 2008.

\bibitem{dai2012some}
B.~Dai and Y.~Luo, ``Some new results on the wiretap channel with side
  information,'' {\em Entropy}, vol.~14, no.~9, pp.~1671--1702, 2012.

\bibitem{goldfeld2019wiretap}
Z.~Goldfeld, P.~Cuff, and H.~Permuter, ``Wiretap channels with random states
  non-causally available at the encoder,'' {\em IEEE Transactions on
  Information Theory}, vol.~66, no.~3, pp.~1497--1519, 2019.

\bibitem{chen2021strong}
Y.~Chen, D.~He, and Y.~Luo, ``Strong secrecy of arbitrarily varying multiple
  access channels,'' {\em IEEE Transactions on Information Forensics and
  Security}, vol.~16, pp.~3662--3677, 2021.

\bibitem{chen2022strong}
Y.~Chen, D.~He, C.~Ying, and Y.~Luo, ``Strong secrecy of arbitrarily varying
  wiretap channel with constraints,'' {\em IEEE Transactions on Information
  Theory}, vol.~68, no.~7, pp.~4700--4722, 2022.

\bibitem{koga2013information}
T.~S. Han, {\em Information-spectrum methods in information theory}.
\newblock Springer, 2002.

\bibitem{bloch2013strong}
M.~R. Bloch and J.~N. Laneman, ``Strong secrecy from channel resolvability,''
  {\em IEEE Transactions on Information Theory}, vol.~59, no.~12,
  pp.~8077--8098, 2013.

\bibitem{goldfeld2016arbitrarily}
Z.~Goldfeld, P.~Cuff, and H.~H. Permuter, ``Arbitrarily varying wiretap
  channels with type constrained states,'' {\em IEEE Transactions on
  Information Theory}, vol.~62, no.~12, pp.~7216--7244, 2016.

\bibitem{goldfeld2016semantic}
Z.~Goldfeld, P.~Cuff, and H.~H. Permuter, ``Semantic-security capacity for
  wiretap channels of type ii,'' {\em IEEE Transactions on Information Theory},
  vol.~62, no.~7, pp.~3863--3879, 2016.

\bibitem{shannon1949communication}
C.~E. Shannon, ``Communication theory of secrecy systems,'' {\em The Bell
  system technical journal}, vol.~28, no.~4, pp.~656--715, 1949.

\bibitem{chia2012wiretap}
Y.-K. Chia and A.~El~Gamal, ``Wiretap channel with causal state information,''
  {\em IEEE Transactions on Information Theory}, vol.~58, no.~5,
  pp.~2838--2849, 2012.

\bibitem{fujita2016secrecy}
H.~Fujita, ``On the secrecy capacity of wiretap channels with side information
  at the transmitter,'' {\em IEEE Transactions on Information Forensics and
  Security}, vol.~11, no.~11, pp.~2441--2452, 2016.

\bibitem{sasaki2019wiretap}
T.~S. Han and M.~Sasaki, ``Wiretap channels with causal state information:
  Strong secrecy,'' {\em IEEE Transactions on Information Theory}, vol.~65,
  no.~10, pp.~6750--6765, 2019.

\bibitem{csiszar2011information}
I.~Csisz\'ar and J.~K{\"o}rner, {\em Information Theory: Coding Theorems for
  Discrete Memoryless Systems}.
\newblock Cambridge University Press, 2011.

\bibitem{sasaki2021wiretap}
T.~S. Han and M.~Sasaki, ``Wiretap channels with causal and non-causal state
  information: revisited,'' {\em IEEE Transactions on Information Theory},
  vol.~67, no.~9, pp.~6122--6139, 2021.

\bibitem{sonee2014wiretap}
A.~Sonee and G.~A. Hodtani, ``Wiretap channel with strictly causal side
  information at encoder,'' in {\em 2014 Iran Workshop on Communication and
  Information Theory (IWCIT)}, pp.~1--6, IEEE, 2014.

\bibitem{csiszar2000common}
I.~Csisz{\'a}r and P.~Narayan, ``Common randomness and secret key generation
  with a helper,'' {\em IEEE Transactions on Information Theory}, vol.~46,
  no.~2, pp.~344--366, 2000.

\bibitem{molavianjazi2009secure}
E.~MolavianJazi, ``Secure communications over arbitrarily varying wiretap
  channels,'' 2009.
\newblock M.S. thesis, Graduate School Univ. Notre Dame, Notre Dame, IN, USA,
  Dec. 2009.

\bibitem{gastpar2004wyner}
M.~Gastpar, ``The wyner-ziv problem with multiple sources,'' {\em IEEE
  Transactions on Information Theory}, vol.~50, no.~11, pp.~2762--2768, 2004.

\bibitem{li2012multiple}
M.~Li, O.~Simeone, and A.~Yener, ``Multiple access channels with states
  causally known at transmitters,'' {\em IEEE Transactions on Information
  Theory}, vol.~59, no.~3, pp.~1394--1404, 2012.

\bibitem{somekh2008cooperative}
A.~Somekh-Baruch, S.~Shamai, and S.~Verd{\'u}, ``Cooperative multiple-access
  encoding with states available at one transmitter,'' {\em IEEE Transactions
  on Information Theory}, vol.~54, no.~10, pp.~4448--4469, 2008.

\bibitem{el2011network}
A.~El~Gamal and Y.-H. Kim, {\em Network Information Theory}.
\newblock Cambridge University Press, 2011.

\bibitem{thomas2006elements}
T.~M. Cover and J.~A. Thomas, {\em Elements of information theory}.
\newblock Wiley-Interscience, 2006.

\bibitem{wyner1976rate}
A.~Wyner and J.~Ziv, ``The rate-distortion function for source coding with side
  information at the decoder,'' {\em IEEE Transactions on information Theory},
  vol.~22, no.~1, pp.~1--10, 1976.

\bibitem{helal2020cooperative}
N.~Helal, M.~Bloch, and A.~Nosratinia, ``Cooperative resolvability and secrecy
  in the cribbing multiple-access channel,'' {\em IEEE Transactions on
  Information Theory}, vol.~66, no.~9, pp.~5429--5447, 2020.

\bibitem{slepian1973noiseless}
D.~Slepian and J.~Wolf, ``Noiseless coding of correlated information sources,''
  {\em IEEE Transactions on information Theory}, vol.~19, no.~4, pp.~471--480,
  1973.

\bibitem{wyner1973theorem}
A.~Wyner and J.~Ziv, ``A theorem on the entropy of certain binary sequences and
  applications--i,'' {\em IEEE Transactions on Information Theory}, vol.~19,
  no.~6, pp.~769--772, 1973.

\bibitem{dai2014multiple}
B.~Dai, Y.~Wang, and Z.~Zhuang, ``Multiple-access wiretap channel with common
  channel state information at the encoders,'' {\em IET Communications},
  vol.~8, no.~5, pp.~597--606, 2014.

\bibitem{kramer2008topics}
G.~Kramer, {\em Topics in Multi-user Information Theory}.
\newblock Now Publishers Inc, 2008.

\bibitem{he2016strong}
D.~He and W.~Guo, ``Strong secrecy capacity of a class of wiretap networks,''
  {\em Entropy}, vol.~18, no.~7, p.~238, 2016.

\bibitem{li2018strong}
C.~T. Li and A.~El~Gamal, ``Strong functional representation lemma and
  applications to coding theorems,'' {\em IEEE Transactions on Information
  Theory}, vol.~64, no.~11, pp.~6967--6978, 2018.

\end{thebibliography}
\end{document}